\documentclass[tbtags,authorcolumns,numberwithinsect]{no-lipics-v2022}

\usepackage{amssymb}
\usepackage{mathtools}
\usepackage{bm}
\usepackage{stmaryrd}
\usepackage[draft]{fixme}
\usepackage{booktabs}
\usepackage{accents}
\usepackage{trimclip}

\usepackage[T1]{fontenc}
\usepackage[utf8]{inputenc}

\usetikzlibrary{calc}

\SetKwComment{Comment}{$\triangleright$\ \rm}{}

\usetikzlibrary{arrows,arrows.meta,shapes.misc, decorations.pathreplacing,decorations.pathmorphing,calligraphy, patterns.meta}
\usetikzlibrary{calc}
\tikzset{dotmark/.style={circle,fill,inner sep=1.5pt}}
\tikzset{emptymark/.style={circle,draw,fill=white,inner sep=1.5pt}}
\tikzset{crossmark/.style={thick,inner sep=1.5pt}}

\newcommand{\newhat}{\scalebox{1.5}[.75]{\trimbox{0pt 1.1ex}{\textasciicircum}}}

\newcommand{\stretchedhat}[1]{\accentset{\newhat}{#1}}

\def\fragmentco#1#2{{[}#1\,{.\,.}\,#2{)}}

\def\setn#1{{[}#1{]}}

\def\problembox#1{%
    \vspace{2mm}%
    \noindent\fbox{%
    \begin{minipage}{.985\linewidth}%
        #1
    \end{minipage}%
    }%
    \vspace{2mm}%
}

\newcommand{\defproblemp}[4]{%
    \problembox{%
        \textsc{#1}\\
        {\bf{Input:}} #2  \\
        {\bf{Output:}} #3\\
        {\bf{Parameter:}} #4
    }%
}

\makeatletter
\renewenvironment{cases}{%
  \matrix@check\cases\env@cases
}{%
  \endarray\right.%
}
\def\env@cases{%
  \let\@ifnextchar\new@ifnextchar
  \left\lbrace
  \def\arraystretch{1.1}%
  \array{@{\;}c@{\quad}l@{}}%
}
\makeatother

\def\mid{\ensuremath :}

\def\emptyset{\varnothing}

\newcommand{\nat}{\mathbb{N}}
\newcommand{\Q}{\mathbb{Q}}
\newcommand{\Z}{\mathbb{Z}}

\newcommand{\field}[1]{\smash{\mathbb{F}}_{\smash{#1}\vphantom{q}}}

\def\NUM{\text{\#}}
\newcommand{\NUMP}[1]{\NUM{}_#1}
\def\MOD{\text{MOD}_p}
\newcommand{\MODP}[1]{\text{MOD}_#1}
\def\W#1{\ensuremath {\sf W{\bm{[}\,#1\,\bm{]}}}}
\def\W#1{\ensuremath {\sf W{\bm{[}#1\bm{]}}}}
\def\w{\NUM\W1}
\def\wMOD{\MOD\W1}
\newcommand{\wMODP}[1]{\text{MOD}_{#1}\W1}
\newcommand{\wP}[1]{\NUM_{#1}\W1}

\newcommand{\sylow}{\textup{Syl}}

\newcommand{\push}{\overrightarrow{\varphi}}

\newcommand{\wrGraph}{\circ}

\newcommand{\prodGroup}{\times}
\newcommand{\BigProdGroup}{\bigtimes}

\makeatletter
\def\bigtimes{%
  \DOTSB\mathop{\mathpalette\mattos@bigtimes\relax}\slimits@
}
\newcommand\mattos@bigtimes[2]{%
  \vcenter{\hbox{%
    \sbox\z@{$#1\sum$}%
    \resizebox{!}{0.9\dimexpr\ht\z@+\dp\z@}{\raisebox{\depth}{$\m@th#1\times$}}%
  }}%
  \vphantom{\sum}%
}
\def\bigbowtie{%
  \DOTSB\mathop{\mathpalette\mattos@bigbowtie\relax}\slimits@
}
\newcommand\mattos@bigbowtie[2]{%
  \vcenter{\hbox{%
    \sbox\z@{$#1\sum$}%
    \resizebox{!}{0.9\dimexpr\ht\z@+\dp\z@}{\raisebox{\depth}{$\m@th#1\triangledown$}}%
  }}%
  \vphantom{\sum}%
}

\makeatother

\DeclareMathOperator*{\BigJoinGraph}{\bigbowtie}

\newcommand{\UnionGraph}{\uplus}

\newcommand{\unionSet}{\uplus}

\DeclareMathOperator{\IS}{IS}

\DeclareMathOperator{\aut}{Aut}
\DeclareMathOperator{\fp}{FP}

\DeclareMathOperator{\sylelm}{\overline{\varphi}}

\newcommand{\homs}[2]{\mbox{\ensuremath{\mathrm{Hom}(#1 \to #2)}}}
\newcommand{\indsub}{\ensuremath{\mathrm{IndSub}}}
\newcommand{\indsubs}[2]{\mbox{\ensuremath{\mathrm{IndSub}(#1 \to #2)}}}
\newcommand{\auts}[1]{\ensuremath{\mathrm{Aut}(#1)}}

\newcommand{\subs}[2]{\mbox{\ensuremath{\mathrm{Sub}(#1 \to #2)}}}

\newcommand{\primepower}[1]{q_{#1}}

\newcommand{\cphom}{\ensuremath{\mathrm{cp}\text{-}\mathrm{Hom}}}
\newcommand{\cphoms}[2]{\ensuremath{\mathrm{cp}\text{-}\mathrm{Hom}}(#1 \to #2)}
\newcommand{\cpindsub}{\ensuremath{\mathrm{cp}\text{-}\mathrm{IndSub}}}
\newcommand{\cpindsubs}[2]{\ensuremath{\mathrm{cp}\text{-}\mathrm{IndSub}}(#1 \to #2)}

\newcommand{\sat}[1]{\ensuremath{#1\text{-}\textsc{SAT}}}
\newcommand{\UniSat}[1]{\ensuremath{\textsc{Unique}\text{-}#1\text{-}\textsc{SAT}}}
\newcommand{\UniClique}{\ensuremath{\textsc{Unique}\text{-}\textsc{Clique}}}
\newcommand{\UniKClique}[1]{\ensuremath{\textsc{Unique}\text{-}#1\text{-}\textsc{Clique}}}
\newcommand{\clique}{\ensuremath{\textsc{Clique}}}
\newcommand{\homsprob}{\ensuremath{\textsc{Hom}}}
\newcommand{\cphomsprob}{\ensuremath{\textsc{cp-Hom}}}
\newcommand{\indsubsprob}{\ensuremath{\textsc{IndSub}}}
\newcommand{\cpindsubsprob}{\ensuremath{\textsc{cp-IndSub}}}
\newcommand{\Scat}[2]{\ensuremath #1_{\text{Sc}, #2}}

\newcommand{\ScatSet}[2]{\ensuremath \text{Sc}_{#1, #2}}
\newcommand{\CoSet}[2]{\ensuremath \text{Co}_{#1, #2}}

\newcommand{\subprob}{\ensuremath{\textsc{Sub}}}
\newcommand{\primebase}{\ensuremath{p}^{-}}

\newcommand{\colorT}{\ensuremath T}
\newcommand{\colorF}{\ensuremath F}
\newcommand{\colorB}{\ensuremath B}

\DeclareMathOperator{\wrLevel}{\varepsilon}
\DeclareMathOperator{\Hasselevel}{\ell}

\def\fps#1{\ensuremath\fp(\Gamma, #1)}
\def\fpb#1#2{\ensuremath\fp(#1, #2)}

\def\hl#1{\ensuremath \Hasselevel(#1)}

\def\nt#1{\ensuremath M_{#1}}

\def\aename#1{\ensuremath \stretchedhat{#1}}
\def\ae#1#2{\ensuremath \aename{#1}(#2)}
\def\aek#1#2#3{\ensuremath \aename{#1}_{#2}(#3)}
\def\ess#1#2{\ensuremath #1{\{}#2{\}}}

\newcommand{\fpt}{\leq_{\mathrm{T}}^{\mathrm{fpt}}}

\def\od#1{\ensuremath\mathbb{O}(#1)}

\def\sym#1{\ensuremath \mathfrak{S}_{#1}}

\def\cgr#1#2{\ensuremath \smash{C_{#1}^{#2}}\vphantom{{}_q^d}}

\newcommand{\graphs}[1]{\mathcal{G}_{#1}} 
\newcommand{\image}[2]{\mathrm{Im}_{#1}^{#2}} 
\newcommand{\metaGraph}[2]{#1\big\langle#2\big\rangle} 

\newcommand{\reductionGraph}[3]{\Tilde{#1}}

\newcommand{\subCoeff}[3]{\alpha_{#3}^{#2}}

\newcommand{\VC}[1]{\mathrm{VC}(#1)}

\newcommand{\codo}[2]{f_{#1}(#2)}

\makeatletter
\renewenvironment{cases}{%
  \matrix@check\cases\env@cases
}{%
  \endarray\right.%
}
\def\env@cases{%
  \let\@ifnextchar\new@ifnextchar
  \left\lbrace
  \def\arraystretch{1.1}%
  \array{@{\;}c@{\quad}l@{}}%
}
\makeatother

\def\mid{\ensuremath :}

\def\emptyset{\varnothing}

\def\epsilon{\varepsilon}

\title{From Graph Properties to Graph Parameters: Tight~Bounds~for~Counting~on~Small~Subgraphs}

\author{Simon Döring}{Max Planck Institute for Informatics and\\Saarbrücken Graduate School
of Computer Science\\Saarland Informatics Campus\\Saarbrücken, Germany}{sdoering@mpi-inf.mpg.de}{https://orcid.org/0009-0002-6667-5257}{}

\author{Dániel Marx}{CISPA Helmholtz Center for Information Security\\Saarbrücken, Germany}
{marx@cispa.de}{https://orcid.org/0000-0002-5686-8314}{}
\author{Philip Wellnitz}{National Institute of Informatics\\The Graduate University for Advanced Studies, SOKENDAI\\Tokyo, Japan}{wellnitz@nii.ac.jp}{https://orcid.org/0000-0002-6482-8478}{}

\authorrunning{S. Döring, D. Marx, and P. Wellnitz}

\nolinenumbers

\begin{document}
\pagenumbering{roman}
\maketitle
\begin{abstract}
    A {\em graph property} is a function $\Phi$ that maps every graph to $\{0,1\}$ and is
    invariant under isomorphism. In the $\NUM{}\indsubsprob(\Phi)$ problem, given a graph
    $G$ and an integer $k$, the task is to count the number of $k$-vertex induced
    subgraphs $G'$ with $\Phi(G')=1$. For example, this problem family includes counting
    $k$-cliques or induced $k$ vertex subgraphs that are connected, among others. There has been
    extensive work on determining the parameterized complexity of
    $\NUM{}\indsubsprob(\Phi)$ for various properties $\Phi$. Very recently, Döring, Marx,
    and Wellnitz [STOC 2024] showed the general result that $\NUM{}\indsubsprob(\Phi)$ is
    \NUM{}\W1-hard for every nontrivial edge-monotone property $\Phi$ and, assuming ETH, cannot
    be solved in time $f(k)\, n^{o(\log k)}$.

    $\NUM{}\indsubsprob(\Phi)$ can be naturally generalized to {\em graph parameters,}
    that is, to functions $\Phi$ on graphs that do not necessarily map to $\{0,1\}$: now
    the task is to compute the sum $\sum_{G'}\Phi(G')$ taken over all $k$-vertex induced
    subgraphs $G'$. This problem setting can express a wider range of counting problems
    (for instance, counting $k$-cycles or $k$-matchings) and can model problems involving expected
    values (for instance, the expected number of components in a subgraph induced by $k$ random
    vertices).
    Our main results are lower bounds on $\NUM{}\indsubsprob(\Phi)$ in this setting, which
    simplify, generalize, and tighten the lower bounds of Döring, Marx, and Wellnitz in
    various ways.
    \begin{enumerate}
        \item We show a lower bound for every nontrivial edge-monotone graph parameter $\Phi$ with finite codomain
        \\ (not only for parameters that take value in $\{0,1\}$).
        \item The lower bound is tight: we show that, assuming ETH, there is no $f(k)n^{o(k)}$ time algorithm.
        \item The lower bound applies also to the modular counting versions of the problem.
        \item The lower bound applies also to the multicolored version of the problem.
        \end{enumerate}
      We can extend the \NUM{}\W1-hardness result to the case when the codomain of $\Phi$ is not finite,
      but has size at most $(1-\epsilon)\sqrt{k}$ on $k$-vertex graphs. However, if there is no bound on
      the size of the codomain, the situation changes significantly: for example, there is a nontrivial
      edge-monotone function   $\Phi$ where the size of the codomain is $k$ on $k$-vertex graphs and
      $\NUM{}\indsubsprob(\Phi)$ is FPT.
\end{abstract}

\clearpage
\thispagestyle{plain}
\tableofcontents
\clearpage
\pagenumbering{arabic}

\section{Introduction}

Many scientific fields deal with \emph{data}---very often in the form of networks or \emph{graphs}.
To understand said data, it is typically instructive to \emph{count} how often certain
substructures or patterns appear in a given input graph;
for instance in the study of database systems \cite{10.1145/380752.380867, 10.1145/3651614},
neural and social networks \cite{doi:10.1126/science.298.5594.824, 10.1007/978-3-319-21233-3_5}, and
computational biology \cite{10.1007/11599128_7, Uncovering}, to name just a few examples.
Unfortunately, this pattern analysis
becomes an algorithmically challenging task whenever our graphs grow rapidly in size.
Therefore, there is already a need for algorithms
that are able to count just small patterns efficiently.
Hence, a long line of research is devoted to efficiently count certain types of small
subgraphs in large graphs
\cite{hom:basis,sub,10.1007/978-3-642-39206-1_30,topo,alge,Roth2021,hamming,peyerimhoff2021parameterized,hereditary,DBLP:journals/algorithmica/GoldbergR24,edge_monotone}.

For example, if we wish
to compute a function that depends somehow on patterns that have at most $k$ vertices in a graph $G$ with
$n$ vertices, then one could do so by simply considering all induced subgraphs of
size $k$ in $G$.
However, there are $\binom{n}{k}$ many induced subgraphs of size \(k\), meaning that enumerating
all of them would require a running time of $O(n^k)$ which is unfeasible if the
pattern size $k$ becomes too large.
Instead, we require an algorithm
whose complexity does not explode with our pattern size $k$.

We use concepts from parameterized complexity theory to better understand the complexity of such problems.
To be more precise, we would like to know whether our problems are {\em fixed-parameter
tractable (FPT)}, that is, if they can be solved in time $f(k)n^{O(1)}$ for some computable
function $f$. An algorithm with such a running time might be inefficient for large patterns.
However, for small patterns (as they typically occur in practice), an FPT algorithm is
efficient even for large input graphs.

Also, a polynomial running time is often not achievable since many
counting problems are NP-complete and it
is widely believed that many famous counting problems, such as
\textsc{\NUM{}Clique} are {\em not} FPT. The class of \w-hard problems is defined as the
class of problems that are at least as hard as \textsc{\NUM{}Clique}, meaning that
it is very unlikely that such problems are FPT.

A common formalism to unify different counting problem is the following \cite{hard_families,connected,hereditary,hamming,edge_monotone,alge}. Let $\Phi$ be a {\em graph property} such as bipartite, clique, or planar; formally, $\Phi$ is a function that maps every graph to $\{0,1\}$ and closed under isomorphism. Then given a graph $G$ and integer $k$, the task is to count the number of $k$-vertex induced subgraphs that satisfy property $\Phi$, or more formally, to compute the value
\[\NUM{}\indsubs{(\Phi, k)}{G} \coloneqq \sum_{A \subseteq \binom{V(G)}{k}} \Phi(G[A]).\]
For a fixed graph property $\Phi$, the problem $\NUM{}\indsubsprob(\Phi)$ is defined
as the problem that gets as input an integer $k$ and a graph $G$ and computes
$\NUM{}\indsubs{(\Phi, k)}{G}$.

To express a wider range of problems, it is natural to generalize $\NUM{}\indsubs{(\Phi, k)}{G}$ from graph properties to {\em graph parameters:} that is, the codomain of $\Phi$ is not necessarily $\{0,1\}$. For example, when $\Phi$ is the number of perfect matchings or the number of Hamiltonian cycles, then $\NUM{}\indsubs{(\Phi, k)}{G}$ is the number of $k$-matchings and the number of $k$-cycles in $G$, respectively. If $\Phi$ is, say, the chromatic number of the graph, then $\NUM{}\indsubs{(\Phi, k)}{G}$ can be interpreted (after dividing by $\binom{|V(G)|}{k}$) as the expected chromatic number of a subgraph induced by a subset of $k$ random vertices. Thus the $\NUM{}\indsubsprob(\Phi)$ problem allows us to formulate problems involving the expected value of graph parameters over random subsets of vertices.

For a fixed graph property $\Phi$, Jerrum and Meeks \cite{connected,
hard_families} were the first who study the $\NUM{}\indsubsprob(\Phi)$ problem.
In later years, the $\NUM{}\indsubsprob(\Phi)$ problem has attracted a lot of research and it is now known that
$\NUM{}\indsubsprob(\Phi)$ is \w-hard for many classes of
graph properties \cite{even_odd, hom:basis, topo, hamming, alge,
hereditary, edge_monotone}.

For certain properties like \emph{being a complete graph} it is easy to see that $\NUM{}\indsubsprob(\Phi)$ is \w-hard.
On the other hand, there are \emph{trivial} graph parameters for which $\NUM{}\indsubsprob(\Phi)$ is
definitely FPT computable. We say that a graph property or a graph parameter
{\em $\Phi$ is trivial on $k$} if $\Phi$ is constant on $k$-vertex graphs. It is obvious that $\NUM{}\indsubsprob(\Phi)$ is easy if $\Phi$ is trivial on all $k$.

It is known that $\NUM{}\indsubsprob(\Phi)$ is \w-hard for certain $\Phi$ and
FPT computable for other $\Phi$. An implicit result by Curticapean, Dell, and Marx (see \cite{hom:basis}) shows
that for all graph properties $\Phi$ the problem $\NUM{}\indsubsprob(\Phi)$ is always either FPT or \w-hard. This
result can easily extend to graph parameters. However, this result does not give any easy criteria
to check which case holds. For many large classes of graph properties $\Phi$
it is known that $\NUM{}\indsubsprob(\Phi)$ is \w-hard (see \cref{sec:prior})
and current research suggests that $\NUM{}\indsubsprob(\Phi)$ is always \w-hard
whenever $\Phi$ is a nontrivial graph property (see \cite[Conjecture 1]{hereditary,edge_monotone}).

Recently, Döring, Marx and Wellnitz~\cite{edge_monotone} showed that
$\NUM{}\indsubsprob(\Phi)$ is \w-hard for graph properties that are \emph{edge monotone} (closed under
deletion of edges). They also provided some quantitative lower bounds on the exponent of $n$ under
the {\em Exponential-Time Hypothesis (ETH)} \cite{IMPAGLIAZZO2001512}. We extend their results in the following four ways.

\begin{enumerate}[(1)]
    \item We show \w-hardness for nontrivial edge-monotone graph parameters with finite codomain.
    Here, edge-monotone means that $\Phi(H) \geq \Phi(G)$ whenever $H$ is an edge-subgraph
    of $G$. This extends the main result of \cite{edge_monotone} to the much
    larger space of graph parameters. However, we also show in \cref{sec:sub:base}
    that there are nontrivial edge-monotone graph
    parameters with unbounded codomain that are FPT.
    \item We show that $\NUM{}\indsub(\Phi)$ has tight lower bounds under ETH.
    In \cite{edge_monotone}, the authors archived tight lower bounds only for some
    special cases (like $\Phi$ is nontrivial on infinitely many prime powers $p^t$).
    \item We show that $\NUM{}\indsub(\Phi)$ stays hard under modulo counting.
    \item We show that (1) to (3) also hold if
    we consider the colored version $\NUM{}\cpindsub(\Phi)$ of $\NUM{}\indsubsprob(\Phi)$, where the graph $G$ is partitioned into $k$ classes in the input, and we are considering only induced $k$-vertex subgraphs with exactly one vertex in each class.
\end{enumerate}

\subparagraph*{Finite Codomain.}
Our first main result generalizes and tightens the result of Döring, Marx and Wellnitz~\cite{edge_monotone} by considering graph parameters with arbitrary finite codomain (not only $\{0,1\}$) and making the lower bound tight in all cases.

\begin{restatable*}{mtheorem}{thmedgemon}\label{theo:edge_mono}
    Let $\Phi \colon \graphs{} \to \Q$ denote some edge-monotone graph parameter with finite codomain.
    \begin{itemize}
        \item If $\Phi$ is nontrivial, then $\NUM{}\indsub(\Phi)$ (and $\NUM{}\cpindsub(\Phi)$) is $\w$-hard.
        \item Further, assuming ETH, if the codomain of $\Phi$
        is equal to $\{0, \dots, c\}$ then there is a universal constant $\delta_c > 0$
        (that only depends on the size of the codomain $c$)
        such that for any $k \geq 3$ on which $\Phi$ is nontrivial, no algorithm (that reads the whole input)
        computes for every graph $G$ the number $\NUM{}\indsubs{(\Phi, k)}{G}$
        (or $\NUM{}\cpindsubs{(\Phi, \graphs{k})}{G}$) in time $O(|V(G)|^{\delta_c  k})$.
        \qedhere
    \end{itemize}
\end{restatable*}

An important feature of edge-monotone graph parameters is that it is very well possible that the parameter is nontrivial only on graphs of certain sizes. This can be contrasted with the case of, say, hereditary graph classes, where if the class is nontrivial on sizes $k_1$ and $k_2$, then it is nontrivial on all sizes in between.

Given an edge-monotone graph parameter $\Phi$, let $\nt{\Phi}$
contain every size $k$ such that $\Phi$ is nontrival on $k$-vertex graphs.
The proof of Döring, Marx and Wellnitz~\cite{edge_monotone} considers
two main cases and handles them with very different techniques:
(1) when $\nt{\Phi}$ contains arbitrary large primes $p$ and (2)
when $\nt{\Phi}$ contains arbitrary large prime powers $p^t$ for some fixed prime $p$.
The proof is completed with a delicate reduction: if $k\in \nt{\Phi}$ and $p^t$ is the largest prime factor of $k$, then a nontrivial edge-monotone problem on $p^t$ vertices can be reduced to computing $\NUM{}\indsubs{(\Phi, k)}{G}$. While this proves $\NUM{}$W[1]-hardness, the largest prime power $p^t$ dividing $k$ can be as small as logarithmic in $k$, resulting in weak quantitative lower bounds.

Our main technical contribution is a more refined reduction technique that allows reducing a problem on $p^t$ vertices to computing $\NUM{}\indsubs{(\Phi, k)}{G}$, where $p^t$ is now the largest prime power not larger than $k$ (and not necessarily dividing $k$). This allows us to completely bypass case (1) above an reuse (2) essentially verbatim, resulting in a much simpler proof.

\subparagraph*{Infinite Codomain.} The behavior of the $\NUM{}\indsub(\Phi)$ problem
can change significantly if $\Phi$ has infinite domain. Apparently, the behavior of
the problem is influenced by how large the codomain can be on $k$-vertex graphs.
We can extend the $\NUM{}$W[1]-hardness of \cref{theo:edge_mono} to the more general
case when the codomain of $\Phi$ on $k$-vertex graphs is of size at most $(1-\varepsilon)\sqrt{k}$.
However, the quantitative lower bound on the exponent becomes weaker. Here,
$\Phi \colon \graphs{k} \to \{0, \dots, \codo{\Phi}{k}\}$ means that the codomain
of $\Phi$ is $\{0, \dots, \codo{\Phi}{k}\}$ when restricted to $k$-vertex graphs.

\begin{restatable*}{mtheorem}{thmedgemonDomain}\label{theo:edge_mono:domain}
    Let $\Phi \colon \graphs{k} \to \{0, \dots, \codo{\Phi}{k}\}$ be an edge-monotone graph parameter.
    Further, let $\alpha \in \nat$ with $\alpha \geq 2$ and $\varepsilon > 0$ such that
    $\codo{\Phi}{k} = \lfloor (1 - \varepsilon)k^{1/\alpha} \rfloor$.
    \begin{itemize}
        \item If $\Phi$ is nontrivial then $\NUM{}\indsubsprob(\Phi)$ and $\NUM{}\cpindsub(\Phi)$ are both $\w$-hard.
        \item There exists a constant $\delta > 0$ (that only depends on $\alpha$ and $\varepsilon$)
        such that for all nontrivial $k \geq 3$,
         no algorithm (that reads the whole input) computes the function $\NUM{}\indsubs{(\Phi, k)}{\star}$
         (or $\NUM{}\cpindsubs{(\Phi, \graphs{k})}{\star}$) in time $O(|V(G)|^{\delta k^\beta})$,
         where $\beta \coloneqq 1 - {1}/{\alpha}$.
         \qedhere
    \end{itemize}
\end{restatable*}

When the codomain of $\Phi$ on $k$-vertex graphs is $\{0,\dots, k\}$, then it is no longer
true that $\NUM{}\indsub(\Phi)$ is hard for every edge-monotone graph parameter $\Phi$.
A simple example is when $\Phi$ is the number of isolated vertices, in which case an easy
formula gives a formula for $\NUM{}\indsubs{(\Phi, k)}{G}$.
In contrast, we show that our techniques can be
used to obtain hardness results for specific graph parameters with unbounded codomains.
In \cref{theo:hard:common}, we show hardness for $\NUM{}\indsubsprob(\Phi)$ for some common graph parameters like the number
of connected components, chromatic number and maximal degree.

Beyond the example of isolated vertices above, there seem to be many nontrivial graph parameters for which $\NUM{}\indsubsprob(\Phi)$ is FPT.
For example, in \cref{sec:sub:base} we show that, for every fixed $c\ge 1$,  $\NUM{}\indsubsprob(\Phi_c)$ is FPT,  where $\Phi_c(G) = \NUM{}E(G)^c$.
Thus a simple analog of \cref{theo:edge_mono} for edge-monotone $\Phi$ with infinite codomain is out of question. In \cref{sec:sub:base},
we present a large class of nontrivial graph parameters $\Phi$ for which $\NUM{}\indsubsprob(\Phi)$ is FPT. This class could be a
conjectured candidate for the easy cases in an eventual dichotomy generalizing \cref{theo:edge_mono} to arbitrary edge-monotone $\Phi$.
We define this class to contain every $\Phi$ that is a a linear combination of subgraph counts of a set of graphs with bounded
vertex cover number. The idea of using different bases to analyze counting problems
was originally used in \cite{hom:basis}.

Finally, another interesting aspect of unbounded codomain is that it is no longer true that $\NUM{}\cpindsubsprob(\Phi)$
and $\NUM{}\indsubsprob(\Phi)$ have the same complexity. As an example, let $\Phi$ be the number of perfect matchings.
Then $\NUM{}\indsubsprob(\Phi)$ is the problem of counting $k$-matchings, which is well known to be $\NUM{}$W[1]-hard.
On the other hand, given a graph $G$ partitioned into $k$ classes, counting $k$-matchings where each of vertices contains
exactly one endpoint of the matching can be done in FPT time, showing that $\NUM{}\cpindsubsprob(\Phi)$ is FPT.

\subparagraph*{Modulo Counting.}
In \cref{sec:mod}, we further extend the results of \cite{edge_monotone} to modulo counting.
For a fixed prime number $p$ and a counting problem $\NUM{}A$, we define
$\MODP{p}A$ as the problem that decides on input $x$ whether $\NUM{}P(x)$ is
divisible by $p$. Further, we define $\NUMP{p}A$ as the problem that computes
$A(x) \mod p$. If $\NUM{}A$ is parameterized by $\kappa$, then
$\MODP{p}A$ and $\NUMP{p}A$ are both parameterized by $\kappa$.

We say that a problem $\MODP{p}A$ is $\wMODP{p}$-hard if there is a parameterized Turing
reduction from $\MODP{p}\clique{}$ to $\MODP{p}A$. Further, we say that a problem $\NUMP{p}A$
is $\wP{p}$-hard if there is a parameterized Turing
reduction from $\NUMP{p}\clique{}$ to $\NUMP{p}A$. Observe that there is always an easy Turing
reduction from $\NUMP{p}A$ to $\MOD{}A$ (simply return true if the output of $\MOD{}A(x)$ is zero).

The complexity of modulo counting gives us additional information about the hardness of our counting problem.
For example, it is well-known that computing the permanent of a matrix is $\NUM{}P$-complete
(see \cite{VALIANT1979189}), thus there is probably no algorithm that can compute the
permanent in polynomial time. However, observe that the permanent modulo $2$ is equal to the determinant modulo $2$.
Hence, this value can be computed in polynomial time. Hence, a counting problem that is
$\wMOD$-hard or $\wP{p}$-hard for all prime numbers $p$ can therefore be regarded as
a particularly hard counting problem. We show that this is the case for nontrivial
edge-monotone graph parameters. Further, we show tight bounds assuming
the randomized version of ETH (rETH).

\begin{restatable*}{mtheorem}{thmedgemonMOD}\label{theo:edge_mono:mod}
    Let $\Phi \colon \graphs{} \to \Z$ denote an edge-monotone graph parameter with finite codomain and
    $p$ be a prime number.
    \begin{itemize}
        \item If $\Phi$ is nontrivial, then $\NUMP{p}\indsubsprob(\Phi)$ and $\NUMP{p}\cpindsubsprob(\Phi)$
        are both $\wMODP{p}$-hard and $\wP{p}$-hard.
        \item Further, assuming rETH, if the codomain of $\Phi$
        is equal to $\{0, \dots, c\}$, with $c < p$, then there is a universal constant $\delta_c > 0$
        (that only depends on the size of the codomain $c$)
        such that for any $k \geq 3$ on which $\Phi$ is nontrivial, no algorithm (that reads the whole input)
        solves for every graph $G$ the problem $\NUMP{p}\indsubsprob(\Phi)$
        (or $\NUMP{p}\cpindsubsprob(\Phi, \graphs{k})$) in time $O(|V(G)|^{\delta_c  k})$.
        \qedhere
    \end{itemize}
\end{restatable*}

If we apply \cref{theo:edge_mono:mod} to a nontrivial edge-montone graph parameter
$\Phi \colon \graphs{} \to \{0, 1\}$, then we obtain that the computation of
$\NUM{}\indsubsprob(\Phi)$ modulo any prime number $p$ is still hard.

\begin{restatable}{corollaryq}{coredgemonMOD}\label{cor:hardness:mod}
    Let $\Phi \colon \graphs{} \to \{0, 1\}$ denote some edge-monotone graph property and
    $p$ be a prime number.
    \begin{itemize}
        \item If $\Phi$ is nontrivial, then $\NUMP{p}\indsubsprob(\Phi)$ (and $\NUMP{p}\cpindsubsprob(\Phi)$)
        is $\wMODP{p}$-hard and $\wP{p}$-hard.
        \item Further, assuming rETH, universal constant $\delta > 0$
        (independent of $\Phi$) such that for any prime number p and integer $k \geq 3$, on which $\Phi$ is nontrivial,
        no algorithm (that reads the whole input) solves for every graph $G$ the problem
        $\NUMP{p}\indsubsprob(\Phi)$ (and $\NUMP{p}\cpindsubsprob(\Phi)$) in time
        $O(|V(G)|^{\delta  k})$.
        \qedhere
    \end{itemize}
\end{restatable}

Observe that \cref{cor:hardness:mod} does not follow directly from the results of~\cite{edge_monotone}.
Imagine that we have an edge-monotone graph property $\Phi$ that is nontrivial only on
prime numbers.
Then for this case, the authors of~\cite{edge_monotone} rely on the  hardness of
$\NUM{}\homsprob(\mathcal{H})$ for a class of graphs $\mathcal{H}$ with unbounded treewidth
(see~\cite[Theorem 5.16]{edge_monotone}). However, observe that $\NUM{}\homsprob(\mathcal{H})$
might be trivial modulo $p$.

For example, fix a prime number $p$ and let $\mathcal{H} \coloneqq \{K_n : n \geq 2\}$.
Now, for all \(n \ge p\), we have
\[\NUM{}\homs{K_n}{G} = n!~\cdot~\NUM{}\subs{K_n}{G} \equiv_p 0.\]
This
implies that $\MOD{}\homsprob(\mathcal{H})$ is not $\wMOD$-hard although $\NUM{}\homsprob(\mathcal{H})$
is $\w$-hard. Hence, we cannot rely on the \w-hardness of $\NUM{}\homsprob(\mathcal{H})$ to
show the $\wMOD$-hardness of $\MOD{}\indsubsprob(\Phi)$. The parameterized complexity of
$\MOD{}\homsprob(\mathcal{H})$ was analyzed in \cite[Theorem 1.8]{peyerimhoff2021parameterized}.

\subsection{Prior Work}\label{sec:prior}

In this section, we summarize recent results for the parameterized complexity of the
$\NUM{}\indsubsprob(\Phi)$ problem. We order the results by their date of first
announcement.

\begin{enumerate}[(a)]
    \item The problem $\NUM{}\indsubsprob(\Phi)$ was first studied for some specific graph
    parameters $\Phi$. It was shown that $\NUM{}\indsubsprob(\Phi)$ is \w-hard for
    \begin{enumerate}[(i)]
        \item $\Phi(G) = 1$ if and only if $G$ is connected \cite{connected}
        \item $\Phi$ has a low edge densities \cite{hard_families}
        \item $\Phi(G) = 1$ if and only if $|E(G)|$ is even,  $\Phi(G) = 1$ if and only if
            $|E(G)|$ is odd (both in \cite{even_odd})
    \end{enumerate}

    \item In \cite{alge}, Dörfler, Roth, Schmitt, and Wellnitz proved that
        $\NUM{}\indsubsprob(\Phi)$ is \w-hard if there are infinitely many prime powers $t$
        such that $\Phi(K_{t, t}) \neq \Phi(\IS_{2t})$. We rely heavily on their techniques
        to prove our results.
    \item In \cite{hamming}, Roth, Schmitt, and Wellnitz proved
    that $\NUM{}\indsubsprob(\Phi)$ is \w-hard if $\Phi$ is
        \emph{monotone}, meaning closed under taking subgraphs.
\item In \cite{hereditary}, Focke and Roth proved that $\NUM{}\indsubsprob(\Phi)$ is
        \w-hard whenever $\Phi$ is a nontrivial and hereditary graph property.
        A say that $\Phi$ is \emph{hereditary} if it is closed under
        vertex-deletion, meaning that if $G$
        satisfies $\Phi$, then each induced subgraph of $G$ also satisfies $\Phi$.
    \item Lastly in \cite{edge_monotone}, Döring, Marx and Wellnitz showed that $\NUM{}\indsubsprob(\Phi)$ is
    \w-hard if $\Phi$ is nontrivial and edge-monotone. In this work, we refine their work by simplifying,
    improving, and extending their main results.
\end{enumerate}

For modulo counting the following results are known.

\begin{enumerate}[(a)]
    \item Beigel, Gill, and Hertrampf were the first who considered
    the decision problem $\MODP{p}A$ in \cite{Beigel_MOD_p}
    by introducing the complexity class $\MODP{p}P$ (the set of problems
    $\MODP{p}A$ that can be solved in polynomial time). The showed the
    surprising result that $\MODP{p}P = \MODP{{p^t}}P$. There work
    was later extended by Hertrampf in \cite{HERTRAMPF1990325}.

    \item Faben introduced in \cite{Faben_sharp_p} the idea of
    modulo counting by introducing $\NUMP{p}A$. He used this
    class to give a dichotomy result for the Generalised Satisfiability Problem
    modulo $p$.

    \item In \cite{alge}, Dörfler, Roth, Schmitt, and Wellnitz proved that
    $\NUMP{p}\indsubsprob(\Phi)$ is $\wP{p}$-hard if there are infinitely many prime powers $t$
    such that $\Phi(K_{t, t}) \not \equiv_p \Phi(\IS_{2t})$. Note that they used
    $\MODP{p}\indsubsprob(\Phi)$ to denote $\NUMP{p}\indsubsprob(\Phi)$.

    \item In \cite{peyerimhoff2021parameterized}, Peyerimhoff, Roth, Schmitt, Stix, and Vdovina
    analyzed the complexity of $\NUMP{p}\homsprob(\mathcal{H})$ for
    recursively enumerable classes of graphs $\mathcal{H}$. They showed
    that  $\NUMP{p}\homsprob(\mathcal{H})$ is $\wP{p}$-hard if and
    only if the class $\mathcal{H}^\ast_p$ has unbounded treewidth.
    Otherwise, the problem is FPT solvable.

    \item In \cite{graphom_mod}, Bulatov, and Kazeminia analyzed the complexity
    of the problem $\NUMP{p}\texttt{GraphHom}(H)$ that for a fixed graph $H$
    and an input graph $G$ asks to compute the number $\NUM{}\homs{G}{H} \mod p$.
    They complete classified when $\NUMP{p}\texttt{GraphHom}(H)$ is $\NUMP{p}P$-complete
    for square-free graphs $H$.

\end{enumerate}

Besides the aforementioned results, there is a lot of research specifically about the complexity class $\MODP{2}P$
which is also know as \emph{parity P} and is usually denoted by $\oplus P$ (see \cite{10.5555/647211.720053, DBLP:journals/siamdm/FockeGRZ21, DBLP:journals/algorithmica/GoldbergR24}).

\subsection{Technical Overview}

This section gives an overview of the techniques and ideas that we use to prove our main theorems.
We start by discussing techniques that allow us to get hardness results and lower bounds under ETH.

\paragraph*{The Alternating Enumerator and Reductions}

We start by considering the $\clique{}$ problem that is $\w$-hard by definition and has tight
lower bounds under ETH (see \cref{theo:k-clique:ETH}). Next,
we construct a reduction from $\clique{}$ to $\NUM{}\cpindsub(\Phi)$ that preserves the
tight lower bounds. Later, we use a reduction from $\NUM{}\cpindsub(\Phi)$ to
$\NUM{}\indsub(\Phi)$. Such reductions were first developed for \emph{graph properties}
in \cite{alge} and later refined in \cite{edge_monotone}. We present the reduction for
\emph{graph parameters} in \cref{sec:reduction}.

Our reduction relies on the crucial observation
that we can find, for each edge-monotone graph parameter $\Phi \colon \graphs{} \to \Q$,
graphs whose \emph{alternating enumerator} is nonvanishing, where the alternating enumerator is defined as
\begin{align*}
    \ae{\Phi}{G} \coloneqq \sum_{S \subseteq E(G)} \Phi(\ess{G}{S}) (-1)^{\#S}.
\end{align*}
We say that a graph $G$ is \emph{nonvanishing} if the alternating enumerator $\ae{\Phi}{G}$ is not zero.
The graph parameter $\Phi$ is always clear from context. The alternating enumerator was originally
introduced in~\cite{alge} and later used in~\cite{edge_monotone}.

The main observation of \cref{sec:reduction} is that we can solve the $\ell$-$\clique{}$ problem by using an oracle for
$\NUM{}\indsubs{(\Phi,k)}{\star}$ whenever we can find a nonvanishing graph $H$ that has $k$ vertices that
contains $K_{\ell, \ell}$ as a subgraph (see \cref{theo:lower:bound:bicliques}). Thus, to get tight
bounds, we try to find nonvanishing graphs $\{H_k\}$ that contain bicliques of unbounded size.
This leads to the main result of \cref{sec:reduction}.

\begin{restatable*}{corollary}{cortightbounds}\label{cor:tight:bounds}
    Let $\Phi \colon \graphs{} \to \Q$ be a graph parameter and $0 < \psi \leq 1$.
    Assume that there is a set of graphs $A \coloneqq \{G_k\}$ and a function $h \colon A \to \nat$
    such that each $G_k$ has $k$ vertices, contains $K_{h(G_k), h(G_k)}$
    as a subgraph and has a nonvanishing alternating enumerator.
    \begin{itemize}
        \item If $h(G) \in \omega(1)$, then the problem $\NUM{}\indsub(\Phi)$ (and $\NUM{}\cpindsub(\Phi)$) is $\w$-hard.
        \item If $h(G) \in \Omega(|V(G)|^\psi)$, then there is a global constant $\gamma > 0$
            such that for each fixed $k = |V(G_j)|$, no algorithm (that reads the whole input)
            that for every \(G\) computes
            $\NUM{}\indsubs{(\Phi, k)}{G}$ (or $\NUM{}\cpindsubs{(\Phi, \graphs{k})}{G}$) in time $O(|V(G)|^{\gamma  k^\psi})$, unless ETH fails.
            Here, the constant $\gamma$ depends on $h$ but is independent of $\Phi$ and $k$.
        \item If $h(G) \in \Omega(|V(G)|^\psi)$, then there is no algorithm (that reads the whole input) that for every \(G\) computes
            $\NUM{}\indsubs{(\Phi, k)}{G}$ (or $\NUM{}\cpindsubs{(\Phi, \graphs{k})}{G}$) in time $O(f(k) \cdot |V(G)|^{o(k^\psi)})$ for all computable functions $f$ unless ETH fails.
            \qedhere
    \end{itemize}
\end{restatable*}

\paragraph*{Computing the Alternating Enumerator using Fixed Points}

When using the alternating enumerable, we have the problem that the computation of this alternating sum over all edge subgraphs
is highly nontrivial. Nevertheless, there are some cases where the computation of $\ae{\Phi}{H}$ is straightforward.
For example, assume that $\Phi(H) = a$ and $\Phi(H') = b$ for all proper edge-subgraphs $H'$ of $H$. Now, using
the binomial theorem, it is easy to see
that $\ae{\Phi}{H} = (-1)^{\NUM{}E(H)}(b - a)$. Of course, this observation is only useful in a few special cases.

Thus, we rely on techniques that were originally developed
in~\cite{alge} and later refined in~\cite{edge_monotone}. These techniques compute the alternating enumerator
modulo a prime number $p$ by using the fixed points of a
specific group action that is defined using $p$-subgroup of the automorphism group of $K_n$. A $p$-group is a group
whose order is a $p$-power. Let $\Gamma \subseteq \aut(H)$ be $p$-subgroup of the automorphism group $\aut(G)$ of $G$,
then we say that an edge-subgraph $A$ is an
\emph{fixed point of $H$ with respect to $\Gamma$} if $g A = A$ for all $g \in \Gamma$, where $g A$ is the graph
with vertex set $A$ and edge set $\{\{g(u), g(v)\} : \{u, v\} \in E(A)\}$. We use $\fp(\Gamma, H)$
to denote the set of all fixed points.
If $A$ and $B$ are two fixed points with
respect to the same group $\Gamma$ and $B$ is an edge-subgraph of $A$, then we say that $B$ is a \emph{sub-point} of $A$.
This is sometimes denoted by $B \subseteq A$. Now, we exploit that these fixed points form a \emph{lattice
structure} that we can analyze. A complete introduction to this topic can be found in \cref{app:fix:basics}.

The advantage of using fixed points is that they allow us to compute $\ae{\Phi}{H}$ modulo a prime number $p$ by
considering only the fixed points of this specific group action (see \cref{lem:chi:comp}). Usually, the number
fixed is much lower than the number subgraphs. Further, the set of fixed points has a lot of structure that
helps us with our analyses. Now, similar to the example
given at the beginning of this section, we can consider the case that we have a fixed point $A$ with $\Phi(A) = a$
such that $\Phi(B) = b$ for all fixed points \(B\) that lie in \(A\). As it turns out, we again
get $\ae{\Phi}{A} \equiv_p (-1)^{\NUM{}E(H)}(b - a)$,
which is enough to show that $\ae{\Phi}{A}$ is nonvanishing. This is our main tool to show
that we have a nonvanishing alternating enumerator.

\begin{restatable*}[See {\cite[Lemma 4.8]{edge_monotone}}]{lemma}{lemcompae}\label{lem:ae:nonvanishing:sub-points}
    Let \(H\) denote a graph, let \(\Gamma \subseteq \auts{H}\) denote a \(p\)-group, and
    let \(\Phi \colon \graphs{} \to \{0, \dots, c\} \) denote a graph parameter with $c < p$.
    Further, let \(A \in \fps{H}\) denote a fixed point with \(\Phi(A) = a\),
    such that \(\Phi(B) = b\) for all of the proper sub-points \(B\) of \(A\).
    If $b \neq a$, then $\ae{\Phi}{A}$ is nonvanishing and $\ae{\Phi}{A} \not \equiv_p 0$.
\end{restatable*}

\paragraph*{The Main Steps of the Proof of \cref{theo:edge_mono}}

Now, to prove \cref{theo:edge_mono}, we need to find graphs that contain large bicliques and whose
alternating enumerator is nonvanishing. However, we cannot always find such graphs.
Thus, we rely on additional ideas. The proof works in multiple steps.

\subparagraph*{The Prime Power Case: Sylow Groups.}

In \cref{sec:prime:power}, we show that \cref{theo:edge_mono} is true for the prime power case. Similar to \cite{edge_monotone},
we archive this result by proving that if $\Phi$ is nontrivial on a prime power $p^t$ then we can find a nonvanishing graph $H$
that contains $K_{p^{t-1}, p^{t-1}}$ as a subgraph. This is enough due to \cref{cor:tight:bounds}.
We find the graph $H$ by analyzing the fixpoints
of the $p$-sylow group of $\aut(K_{p^t})$ (see \cref{sec:sylow}). In \cref{lem:treewidth:wreath:level} we show that these fixed points
contain large bicliques.

\subparagraph*{Reducing to the Prime Power Case.}

In order to make use of the prime power case, we have to find a way to go from a nonprime power $k$
to a prime power $p^t$. This is archived by using \emph{inhabited graphs}.
Inhabited graphs allow us to take a graph $G$, a graph $C$, and graphs $H_2, \dots H_s$ and combine them
into a larger graph (see {\cite[Definition 6.1]{edge_monotone}}).

\begin{restatable*}{definition}{definhab}\label{def:10.31-1}
    For graphs $G_1, \dots, G_m$ and $C \in \graphs{m}$,
    we define the \emph{inhabited graph} $\metaGraph{C}{G_1, \dots, G_m}$ via
    \begin{align*}
        V(\metaGraph{C}{G_1, \dots, G_m}) &\coloneqq V(G_1) \unionSet \dots \unionSet V(G_m) \\
        E(\metaGraph{C}{G_1, \dots, G_m}) &\coloneqq \{\{(i, v_i), (j, u_j)\} \mid \{i,
        j\} \in E(C) \text{ or } (i = j \text{ and } \{v_i, u_i\} \in E(G_i)) \}. \tag*{\qedhere}
    \end{align*}
\end{restatable*}

Assume that a graph parameter $\Phi$ is nontrivial on $k \coloneqq c + p^t$. Then, we could try
to find a graph $C$ and graphs $H_2, \dots, H_s$ with $c = \sum_{i = 2}^s |V(H_i)|$ such that the graph
parameter $\Tilde{\Phi}(G) = \Phi(\metaGraph{C}{G, H_2, \dots, H_s})$ is nontrivial on $p^t$. Here,
we use the fact that the graph $\metaGraph{C}{G, H_2, \dots, H_s}$ has $k$ vertices and
that $\Phi$ is nontrivial on $k$. Also $\Tilde{\Phi}(G)$
is still edge-monotone, and we get a reduction from $\NUM{}\cpindsubs{\Tilde{\Phi}, \graphs{p^t}}{\star}$ to
$\NUM{}\cpindsubs{\Phi, \graphs{k}}{\star}$.

\begin{restatable*}{lemma}{lemincexc}\label{lem:inc:exc}
    Let \(\Phi \colon \graphs{} \to \Q\) denote a graph parameter and suppose that there is an algorithm
    \(\mathbb{A}\) that computes for each graph \(H\) and each $H$-colored graph $G$ the value
    \(\NUM{}\cpindsubs{(\Phi,H)}{G}\) in time \(g(|V(H)|, |V(G)|)\) for some computable
    function \(g\) that is monotonically increasing.
    Further, for a graph \(C \in \graphs{s}\) and graphs $H_2, H_3, \dots, H_s$, write
    \(
        \reductionGraph{\Phi}{C}{H_2, \dots, H_s}(G) = \Phi(\metaGraph{C}{G, H_2 \dots, H_s})
    \)
    for the graph parameter that inserts $G$ into $\metaGraph{C}{\cdot, H_2, \dots, H_s}$ and evaluates the result on \(\Phi\).

    Then, there is an algorithm \(\mathbb{B}\) with oracle access to \(\mathbb{A}\) that computes for
    each graph \(H\) and each $H$-colored graph $G$ the value $\NUM{}\cpindsubs{(\reductionGraph{\Phi}{C}{H_2, \dots, H_s}, H)}{G}$
    in time $O(g(c + k, c + |V(G)|) + (c + |V(G)|)^2)$,
    where $c \coloneq \sum_{i = 2}^s \#V(H_i)$.
    The algorithm \(\mathbb{B}\) queries \(\mathbb{A}\) only graphs on $H$ with
    \(c + |V(H)|\) many vertices.
\end{restatable*}

\subparagraph*{Concentrated and Reducible Parameters.}

In the following, let $p$ be a prime number and let $\Phi$ be nontrivial on $k$.
We now have two techniques to get hardness. Either, we can directly find a graph $H$ with
$k$ vertices that contains a large enough biclique and has nonvanishing alternating enumerator; or we can
find a graph $C$ and graphs $H_2, \dots, H_s$ such that $\Tilde{\Phi}(G) = \Phi(\metaGraph{C}{G, H_2, \dots, H_s})$
is nontrivial on $p^t$, where $k = p^t + c$ and $c \coloneqq \sum_{i = 2}^s |V(H_i)|$. In order to obtain
good lower bounds, we have to ensure that, for the first case, the graph $H$ contains a large enough biclique. For
the second case, we have to ensure that the value $p^t$ is large enough (i.e. that we do not \emph{lose} too much
when doing the reduction). Otherwise, our reduction would not propagate the tight lower bounds under ETH.
These two approaches motivate the following two definitions.

\begin{restatable*}{definition}{defconsca}
    Let $p$ be a prime number and let $\Phi \colon \graphs{k}  \to \{0, \dots, \codo{\Phi}{k}\} $ denote a graph parameter. We say that
    \begin{itemize}
        \item $\Phi$ is \emph{concentrated on $k$ with respect to $p^t$} if  there is a
        graph $H \in \graphs{k}$ with $\ae{\Phi}{H} \not \equiv_p 0$ and $H$
        contains $K_{p^t, p^t}$ as a subgraph.
        \item $\Phi$ is \emph{reducible on $k$ with respect to $p^t$} if $\codo{\Phi}{k} < p$ and there is a graph $C \in \graphs{s}$ and graphs
        $H_2, \dots, \dots, H_{s} \in \graphs{}$ such that the graph
        parameter $\reductionGraph{\Phi}{C}{H_2, \dots, H_{s}}(G) \coloneq \Phi(\metaGraph{C}{G, H_2 \dots, H_{s}})$ is
        nontrivial on $p^t$ where $k = p^t + \sum_{i = 2}^s |V(H_i)|$.
        \qedhere
    \end{itemize}
\end{restatable*}

To obtain lower bounds that are as tight as possible, we need for all $k$ a prime power $p^t$ such that the
quotient $k/p^t$ is as small as possible. For this, we use the following lemma that ensures that this
is always possible as long as we pick $p^t$ carefully.

\begin{restatable*}{lemma}{lemlargebi}\label{lem:param:large:biclique}
    Let $p$ be prime and let $\Phi \colon \graphs{k} \to \{0, \dots, c\}$ be an edge-monotone graph parameter with $c < p$.
    If $\Phi$ is nontrivial on $k$ and $p^t + p^{t+1} \leq m$, then $\Phi$ is \emph{concentrated on $k$ with respect to $p^t$} or
    $\Phi$ is \emph{reducible on $k$ with respect to $p^{t+1}$}.
\end{restatable*}

\cref{lem:param:large:biclique} is one of the main technical tools that we develop to prove  our main
theorems since this lemma allows us to use a win-win strategy where we either always simply use \cref{cor:tight:bounds}
when $\Phi$ is concentrated or a reduction to the prime power case when $\Phi$ is reducible.

A similar strategy was used by the authors of \cite{edge_monotone} to obtain their lower bounds.
Their approach was to analyze the largest prime power inside a given number $k$ and to use their
techniques on this specific prime power (see \cite[Definition 6.17]{edge_monotone}). However, observe
that the largest prime power of $k$ could be logarithmically small which is why the authors of \cite{edge_monotone}
only archived a lower bound of $O(|V(G)|^{\gamma \sqrt{\log k}/ \log \log k })$. On the other side,
our \cref{lem:param:large:biclique} is much stronger by allowing arbitrary prime powers as long
as $p^t + p^{t+1} \leq k$. It is easy to see that for a fixed prime number $p$, the largest
prime power $p^t$ with $p^t + p^{t+1} \leq k$ has to property that $p^t \geq k/p$ meaning that $p^t$
grows linear with respect to $k$. Thus, if $p$ is fixed then we can archive tight lower bounds
since $k/p^t$ is bounded by a constant.

\subparagraph*{Bringing the Parts Together.}

Assume that $\Phi \colon \graphs{k} \to \{0, \dots, \codo{\Phi}{k}\}$
is edge-monotone and nontrivial. \cref{lem:param:large:biclique}
gives us the tools to prove \cref{theo:edge_mono} by using the following
win-win approach. Let $q(k)$ be a function that returns on input $k$ a prime power $q(k)$ with base $g(k)$
(i.e $q(k) = g(k)^t$ for some $t$) such that $\codo{\Phi}{k} < g(k)$ and $q(k) + q(k) \cdot g(k) \leq k$. Now,
\cref{lem:param:large:biclique} ensures that $\Phi$ is either
concentrated on infinitely many $k$ with respect to $q(k)$ or $\Phi$ is reducible on infinitely $k$ with respect to $q(k) \cdot g(k)$.

If $q(k) \in \Omega(k^{\psi})$ for a fixed $0 < \psi \leq 1$, then we can use \cref{cor:tight:bounds}
to show \w-hardness for the concentrated case.
Further, since each graph $H_k$ contains the graph $K_{q(k), q(k)}$ as a subgraph, we also get a lower bound
of $O(|V(G)|^{\delta k^{\psi}})$.  Otherwise, we are in the reducible case.
In this case, we can construct a graph property $\Scat{\Phi}{q}$ that is nontrivial on infinitely many primer powers $q(k)\cdot g(k)$.
We can use \cref{theo:edge_mono:prime:power:bicliques} to show that $\NUM{}\cpindsub(\Scat{\Phi}{q})$
is \w-hard and to get a lower bound of $O(|V(G)|^{\delta k^{\psi}})$ (see \cref{lem:scat:graph:parameter}).
Observe that we lose a prime factor during this process which is why we need that $\Scat{\Phi}{q}$
is nontrivial on $q(k)\cdot g(k)$ to mitigate this effect.
Now, we can use the reduction from $\NUM{}\cpindsub(\Scat{\Phi}{q})$ to $\NUM{}\cpindsub(\Phi)$
to show that $\NUM{}\cpindsub(\Phi)$ is also \w-hard and has the same lower bounds as
 $\NUM{}\cpindsub(\Scat{\Phi}{q})$ under ETH (see \cref{theo:scat:hard}).

If $\codo{\Phi}{k} \equiv c$, then we can choose the smallest prime
number $p$ with $c < p$ and set $q(k) \coloneqq p^{t-1}$ and $g(k) \equiv p$. Here, $p^t$ is the largest
prime power $p^t$ with $p^t \leq k$. Now, we get that $q(k) \in \Omega(k)$ which is why we get tight
lower bounds. This leads us to \cref{theo:edge_mono}. On the other side,
if $\codo{\Phi}{k} = \lfloor (1 - \varepsilon) k^{1/\alpha} \rfloor$, then we can choose
$q(k)$ in way that $q(k) \in \Omega(k^{1 - {1}/{\alpha}})$ which leads us to \cref{theo:edge_mono:domain}.

In \cref{sec:mod}, we show that our hardness results remain true under modulo counting
as long as $\codo{\Phi}{k}$ is bounded by a constant. These
results use the fact that our reduction from \cref{sec:reduction}
also holds for modulo counting as long as the alternating enumerator
does not vanish modulo $p$, which leads to a reduction from $\MOD{}\clique{}$
to $\MOD{}\indsubsprob(\Phi)$. Lastly, we show that $\MOD{}\indsubsprob(\Phi)$
has tight bounds under rETH. For this, we rely on the fact that
the problem $\UniKClique{k}$ has tight bounds
under rETH which is shown in \cref{sec:appendix:mod}.

\section{Preliminaries}

For natural numbers $n$,
we write $\setn{n}$ for the set $\{1, \dots, n\}$.
For natural numbers \(a\), \(b\) and \(m\), we write $a \equiv_m b$ as a shorthand for $a \equiv b \bmod m$.
Next, we write $\binom{n}{k}$ for the binomial coefficient.
If $A$ is a set then $\binom{A}{k}$ is the set of all subsets $B \subseteq A$ that have size $k$.
For a function $f \colon A \to B$ and a set $C \subseteq A$, we write $f(C)$ for the
set $\{f(c) : c \in C\}$. We say that a set $A$ is \emph{computable} if there exists
an algorithm that decides on input $x$ if $x \in A$. We say that a function $f \colon \nat \to \nat$
is \emph{unbounded} if $\lim_{n \to \infty} f(n) = \infty$. If $G$ is a graph and $h \colon \nat \to \nat$
is function, then $h(G)$ is defined as $h(|V(G)|)$.

\paragraph*{Graphs}

For the whole paper, we only consider {simple} graphs, meaning undirected graphs that have
neither weights, loops, nor parallel edges.
We denote with $\graphs{}$ the set of all (simple) graphs, and use \(\graphs{n}\) for the
set of all (simple) graphs with vertex set \(\setn{n}\). Further, we order
the set $\graphs{}$ lexicographically. This way, we construct the set \(\graphs{}^*\) that
constant of the first graph, in each isomorphism class of graphs, that appears $\graphs{}$. We use \(\graphs{n}^*\) to denote
the set of graphs in \(\graphs{}^*\) with vertex set \(\setn{n}\).

For a graph \(G\), we use \(V(G)\) to denote the set of vertices and we use \(E(G)\) to denote the set of edges.
For a subset of edges $A \subseteq E(G)$, we write $\ess{G}{A}$ for the graph
with vertex set $V(\ess{G}{A}) \coloneqq V(G)$ and edge set $E(\ess{G}{A}) \coloneqq A$.
We say that such a graph is an \emph{edge-subgraph} of $G$. For a set
of vertices $S \subseteq V(G)$,
we write $G[S]$ for the graph induced by $S$. For a graph $H$ and a graph $G$,
we write $\subs{H}{G}$ for the set of
all subgraphs of $G$ that are isomorphic to $H$.

We write $\IS_{k}$ for the independent set that has $k$ vertices.
Further, we use \(K_n\) to denote the complete graph with vertex set $\setn{n}$,
and we write \(K_{n, m}\) for the complete bipartite graph with vertex set
\(\setn{n} \times \setn{m}\).

\paragraph*{Group Theory and Morphisms between Graphs}\label{sec:group:theory}

For a finite set \(X\) of size \(n\), the set of all bijections from $X$ to $X$ forms
together with function composition a group that is usually called the symmetric group \(\sym{n}\).
We use the notation \(\sym{X}\) to emphasize the set \(X\).

We say that a group $\Gamma$ is a \emph{permutation group} of $X$ if \(\Gamma\) is a subgroup of the
symmetric group \(\sym{X}\) for a finite set \(X\).
This group is \emph{transitive} if every \(x \in X\) can be mapped to any other element \(y \in X\) via
a group element \(g \in \Gamma\).
Finally, we say that $\Gamma$ is a \emph{$p$-group} if its order is a power of $p$
(that is, the number of elements that appear in \(\Gamma\) is a power of \(p\)).

A graph \emph{homomorphism} $h \colon V(H) \to V(G)$ from $H$ to $G$ is a function between the
vertex sets of two graphs $H$ and $G$ that preserves adjacencies (but not necessarily
non-adjacencies). This means that \(h\) maps the vertices of every edge $\{u, v\} \in E(H)$
to an edge $\{h(u), h(v)\} \in E(G)$.
We denote the set of all homomorphisms from $H$ to $G$ by $\homs{H}{G}$.
We say that a graph homomorphism $h \colon V(H) \to V(G)$ is \emph{surjective}
if all vertices and edges of $G$ are in the image of $h$. The \emph{spasm} of $G$
is the set of all graphs $G$ such that there is a surjective homomorphism
from $H$ to $G$.

A {homomorphism} $h : V(H) \to V(G)$ is called an
\emph{isomorphism} if $h$ is a bijection on the vertex sets and $h^{-1} : V(H) \to V(G)$
is a homomorphism, too. This definition is equivalent to saying that
$\{u, v\} \in E(H)$ if and only if $\{h(u), h(v)\} \in E(G)$.
We call two graphs \(G\) and \(H\) \emph{isomorphic},
if there is an {isomorphism} between them, which is denoted by \(G \cong H\).
An \emph{automorphism} of $G$ is an {isomorphism} from $G$ to itself.
The set of all \emph{automorphism} of $G$ is denoted by $\aut(G)$.

The set of automorphism $\auts{G}$ together with composition
forms a group. Observe that the automorphism
group of the clique \(\auts{K_n}\) is just the symmetric group \(\sym{n}\).

\subparagraph*{Parameterized Problems and Parameterized Complexity}
A \emph{parameterized (counting) problem} is defined via a function $P \colon \Sigma^\ast \to \nat$ and
a computable parameterization $\kappa \colon \Sigma^\ast \to \nat$. To goal is to find an
algorithm that computes in input $x \in \Sigma^\ast$ to value $P(x)$.
We say that a parameterized problem \((P, \kappa)\) is \emph{fixed-parameter
tractable} (FPT) if there is a computable function $f$ and a deterministic algorithm
\(\mathbb{A}\) such that the algorithm \(\mathbb{A}\) computes the value $P(x)$
in time  $f(\kappa(x)) |x|^{c}$, for a fixed constant $c$, for all $x \in \Sigma^\ast$.

A parameterized Turing reduction from $(P, \kappa)$ to $(P', \kappa')$ is a deterministic
FPT algorithm with oracle access to $P'$ that computes $P(x)$ such that there is a
computable function $g$ with the property that $\kappa'(y) \leq g(\kappa(y))$ holds for
each oracle access to $P'$. We write $A \fpt B$ to denote that there is a parameterized Turing
reduction from $A$ to $B$.

The counting version $\NUM{}P$ of a decision problem $P$ is the problem to compute the
number of valid solutions for a given input $x$.
One very important example of such a counting problem is the \NUM{}\clique problem.
The input is a graph $G$ and a natural number $k$ and
the output is the number of induced
subgraphs of $G$ of size $k$ which form a $k$-clique.
We parameterize \(\NUM{}\clique\) by $\kappa(G, k) \coloneqq k$.
Much research has been focused on understanding the complexity of \(\NUM{}\clique\) and it is
widely believed that no FPT algorithm exists \cite{path_cycle, CHEN2005216}.
We say therefore that a problem $(P',\kappa')$ is \w-hard if there is parameterized Turing reduction
from $(\NUM{}\clique, \kappa)$ to $(P', \kappa')$, which intuitively means that $(P',\kappa')$ is
at least as hard as \(\NUM{}\clique\). Thus, \w-hardness rules out FPT algorithms of $(P',\kappa')$
based on said beliefs.

In the following, we further define parameterized counting problems.
For a recursively enumerable class of graphs $\mathcal{H}$,
a problem instance of $\NUM{}\homsprob(\mathcal{H})$
is a graph $H \in \mathcal{H}$ and a graph
$G \in \mathcal{G}$. The output is the number of homomorphisms from $H$ to $G$ (that
is, $\NUM{}\homs{H}{G}$); the parameterization is given by $\kappa(H, G) \coloneqq |V(H)|$.
In \cite{DALMAU2004315}, Dalmau and Jonsson proved that
$\NUM{}\homsprob(\mathcal{H})$ is \w-hard if and only if the treewidth of
the set $\mathcal{H}$ is unbounded (that is, there is no constant $c$ such that the
treewidth of all elements in $\mathcal{H}$ is below $c$).

Next, we define the colored version of $\NUM{}\homsprob(\mathcal{H})$. Here, we get
two graphs $G$ and $H$, with a coloring $c \colon V(G)\to V(H)$. Here, $c$ is simply
a homomorphism from $G$ to $H$. A~color-prescribed homomorphism $h$ from $H$ to \(G\)
is a homomorphism from $H$ to $G$ such that $c(h(v)) = v$ for
all $v \in V(H)$.

We write \(\cphoms{H}{G}\) for the set of all color-prescribed homomorphism from $H$ to
a graph \(G\) that is \(H\)-colored via \(c\). Observe that \(\cphoms{H}{G}\) might
change if we use a different coloring.
For a recursively enumerable class of graphs $\mathcal{H}$,
the problem $\NUM{}\cphomsprob(\mathcal{H})$ gets as input a
graph $H \in \mathcal{H}$ and a graph $G$ that is \(H\)-colored via \(c\), the task is
to compute the value $\NUM{}\cphoms{H}{G}$.
We parameterize $\NUM{}\cphomsprob(\mathcal{H})$ by $\kappa(H, G) \coloneqq |V(H)|$.

\paragraph*{Graph Parameters and the $\NUM{}\indsubsprob(\Phi)$ problem}

\emph{Graph parameters} (also known as graph invariants) are
functions $\Phi$ that assign a value to a given graph
and are invariant under graph isomorphisms that is $\Phi(G) = \Phi(H)$ whenever $G \cong H$. In this paper,
we consider only graph parameters that map graphs to rational numbers. Further, we exclusively consider
computable graph parameters. A graph parameter $\Phi$ is considered \emph{computable} if there is an algorithm
that gets as input a graph $G$ and computes the value $\Phi(G)$. We say that $\Phi$ is \emph{nontrivial on $n$}
if $\Phi$ is not constant on the set of $n$-vertex graphs, and we say that $\Phi$ is \emph{nontrivial} if it
is nontrivial on infinitely many values $n$. \emph{Graph properties} are graph parameters with codomain $\{0, 1\}$.

We sometimes write $\Phi \colon \graphs{k} \to \{0, \dots, \codo{\Phi}{k}\}$ where
$\codo{\Phi}{k} \colon \nat \to \nat$ is some computable function. This notation means
that the graph parameter $\Phi$ has codomain $\{0, \dots, \codo{\Phi}{k}\}$ when restricted
to $k$-vertex graphs. For instance, if $\Phi(G) = \#E(G)$, then we could
use $\codo{\Phi}{k} = k^2$.

Next, we say that a graph parameter
$\Phi \colon \graphs{} \to \Q$ is \emph{edge-monotone} if we have $\Phi(G) \leq \Phi(H)$ whenever $H$ is an edge-subgraph
of $G$. Note that the definition of edge-monotone graph parameter coincides with the definition of edge-monotonicity
for graph properties given in~\cite{hereditary, edge_monotone}, meaning that we can view each edge-monotone graph property
$\Phi \colon \graphs{} \to \{0, 1\}$ as an edge-monotone graph parameter. Observe that if $\Phi$ is edge-monotone
and nontrivial on $k$, then $\Phi(\IS_k) > \Phi(K_k)$.
For a fixed graph parameter $\Phi \colon \graphs{} \to \Q$, we define the following parameterized problem.

\defproblemp{$\NUM{}\indsubsprob(\Phi)$}{A natural number $k$ and a graph $G$.}
{The value
\[\NUM{}\indsubs{(\Phi, k)}{G} \coloneqq \sum_{A \in \binom{V(G)}{k}} \Phi(G[A]).\]}
{The problem
is parameterized by the size of the induced subgraphs $k$.}

Observe that for a graph property $\Phi$ the definition
of $\NUM{}\indsubsprob(\Phi)$ considers with the definition of the counting problem $\NUM{}\indsubsprob(\Phi)$ given in \cite{topo, alge, edge_monotone}.
As with $\NUM{}\homsprob$, we also need a colored variant of \(\NUM{}\indsubsprob(\Phi)\).
Imagine that the vertices of a given input graph $G$ are partitioned into $k$ color classes and
that we would like to compute the sum $\Phi(G[A])$ that goes over all $k$-vertex induced subgraphs $G[A]$ of $G$
such that $A$ contains exactly one vertex from each color class (this is also known as \emph{colorful}).
Observe that the sum depends on the coloring of $G$, and
remember that we can describe a coloring using a homomorphism $c \colon V(G)\to V(H)$. Now, a graph parameter
$\Phi \colon \graphs{} \to \Q$ and for a recursively enumerable class of graphs $\mathcal{H}$, we define the following problem.

\defproblemp{$\NUM{}\cpindsubsprob( \Phi,\mathcal{H})$}{A a graph $H \in \mathcal{H}$ and a $H$-graph colored $G \in \graphs{}$.}
{The value
\[\NUM{}\cpindsubs{(\Phi, H)}{G} \coloneqq \sum_{\substack{A \in \binom{V(G)}{|V(H)|} \\
G[A] \text{ is colorful} }} \Phi(G[A]).\]}
{The problem
is parameterized by $\kappa(H, G) \coloneqq |V(H)|$.}

We write $\NUM{}\cpindsubsprob(\Phi)$ instead of $\NUM{}\cpindsubsprob(\Phi, \graphs{})$, meaning that
we can also all graphs for coloring. We use $\NUM{}\cpindsubs{(\Phi, \graphs{k})}{\star}$ to denote
the problem that get as input a $k$-vertex graph $H$ and a graph $G$ and is asked to compute
$\NUM{}\cpindsubs{(\Phi, H)}{G}$.

\paragraph*{Basic Properties of Graph Parameters, $\NUM{}\indsubsprob(\Phi)$, and \(\NUM{}\cpindsubsprob(\Phi,\mathcal{H})\)  }

Next, we collect basic properties of graph parameters,
$\NUM{}\indsubsprob(\Phi)$, and \(\NUM{}\cpindsubsprob(\Phi,\mathcal{H})\). We start by constructing a graph property
given a graph parameter.

\begin{definition}
    For all graph parameters $\Phi \colon \graphs{} \to \Q$ and numbers $k$, we define the following objects.
    \begin{itemize}
        \item The set $\image{\Phi}{k}$ is the image of $\Phi$ on the set $\graphs{k}$ (this is the set of graphs with $k$ vertices).
        Further, we write $\image{\Phi}{}$ for the image of $\Phi$.
        \item For all $b \in \Q$, we define $\Phi^b(G) = 1$ if $\Phi(G) = b$ and $\Phi^b(G) = 0$ otherwise. \qedhere
    \end{itemize}
\end{definition}

Observe that $\image{\Phi}{k}$ is always a finite set. It is not hard to see that we can write each $\Phi$ as a linear
combination of $\Phi^b$ where $b \in \image{\Phi}{k}$. This allows us to also view
$\NUM{}\indsubsprob(\Phi)$ and \(\NUM{}\cpindsubsprob(\Phi,\mathcal{H})\) as a linear combination of
$\NUM{}\indsubsprob(\Phi^b)$ and \(\NUM{}\cpindsubsprob(\Phi^b,\mathcal{H})\) respectively. We use this to prove
statements about $\NUM{}\indsubsprob(\Phi)$ and \(\NUM{}\cpindsubsprob(\Phi,\mathcal{H})\) by first looking at
graph properties.

\begin{remark}\label{lem:param:to:property}
    For all graph parameters $\Phi \colon \graphs{} \to \Q$, non negative integers $k \in \nat$, graphs $G$ and $H$, values $b \in \Q$,
    sets $A \supseteq \image{\Phi}{\#V(G)}$, $B \supseteq \image{\Phi}{k}$, and $C \supseteq \image{\Phi}{\#V(H)}$, we obtain
    \begin{align}
        \Phi(G) &= \sum_{b \in A} b \cdot \Phi^b(G) \label{eq:phi:param} \\
        \NUM{}\indsubs{(\Phi, k)}{G} &= \sum_{b \in B} b \cdot  \NUM{}\indsubs{(\Phi^b, k)}{G}  \label{eq:indsub:param} \\
        \NUM{}\cpindsubs{(\Phi, H)}{G} &= \sum_{b \in C} b \cdot  \NUM{}\cpindsubs{(\Phi^b, H)}{G}  \label{eq:cpindsub:param}
    \end{align}
    For the last equality, we assume that $G$ is $H$-colored.
\end{remark}

Lastly, we show that we can transform each edge-monotone graph parameter $\Phi \colon \graphs{} \to D$ with finite codomain $D \subseteq \Q$
into an edge-monotone graph parameter $\Phi' \colon \graphs{} \to \{0, \dots, c\}$ such that we can compute $\NUM{}\indsubsprob(\Phi)$ using
$\NUM{}\indsubsprob(\Phi')$. This allows us to use $\{0, \dots, c\}$ as our codomain for most theorems.

\begin{lemma}\label{lem:finite:codomain}
    For graph parameter $\Phi \colon \graphs{} \to D$ with finite codomain $D \subseteq \Q$,
    there is a constant $c$ and a graph parameter $\Phi' \colon \graphs{} \to \{0, \dots, c\}$ such that
    \[\NUM{}\indsubs{(\Phi, k)}{G} = f(k, |V(G)|, \NUM{}\indsubs{(\Phi', k)}{G}) \quad \text{for some function $f(k,n,m)$.} \]
    The function $f(k, |V(G)|, \NUM{}\indsubs{(\Phi', k)}{G})$ can be computed in time $O(|V(G)|)$.
    Further, if $\Phi$ is edge-monotone then $\Phi'$ can be chosen to be edge-monotone.
\end{lemma}
\begin{proof}
        Let $n \coloneqq |V(G)|$. Since $D = \{d_1, \dots d_m\}$ is finite, there exists a constant $d > 0$ such that $d \cdot d_i \in \Z$
        for all $d_i \in D$. Further, since the set $\{d \cdot d_i : d_i \in \Q\}$ is finite, there is a constant
        $s$ and $c \in \nat$ such that $s + d \cdot d_i \in \{0, \dots, c\}$. We define the graph parameter
        $\Phi'(G) \coloneq s + d\cdot\Phi(G)$. Observe that $\Phi'$ is still computable and that we have
        $\NUM{}\indsubs{(\Phi', k)}{G} = s \cdot \binom{n}{k} + d \cdot \NUM{}\indsubs{(\Phi, k)}{G}$. Hence, by setting
        \[f(k, n, m) \coloneqq \frac{m - s \cdot \binom{n}{k}}{d},\]
        we obtain $\NUM{}\indsubs{(\Phi, k)}{G} = f(k, n, \NUM{}\indsubs{(\Phi', k)}{G})$.
        Next, note that we can assume $k \leq n$ and $|\NUM{}\indsubs{(\Phi, k)}{G}| \leq 2^n \cdot \max_{d_i \in D} |d_i|$. Thus,
        the function $f(k, n, |\NUM{}\indsubs{(\Phi, k)}{G}|)$ can be computed in time $O(n)$.
        Lastly, if $\Phi$ is edge-monotone, then $\Phi'$ is also edge-monotone since $d > 0$.
\end{proof}

\section{Reduction from $\NUM{}\clique$ to $\NUM{}\indsubsprob(\Phi)$}\label{sec:reduction}

In this section, we show a reduction from $\NUM{}\clique$ to $\NUM{}\indsubsprob(\Phi)$ which is an important step
for showing \cref{theo:edge_mono}. We use the following chain of reductions which is a modification of \cite{alge}.
\begin{align}\label{reduction:chain}
    \NUM{}\clique{}
    &\overset{\text{\cref{lem:tight:biclique:construction}}}{\fpt}
    \NUM{}\cphomsprob(\mathcal{H})\\\nonumber
    &\overset{\text{\cite[Lemmas 7 and 8]{alge}}}{\fpt}
    \NUM{}\cpindsubsprob(\Phi,\mathcal{H})
    \overset{\text{\cite[Lemma 10]{alge}}}{\fpt}
    \NUM{}\indsubsprob(\Phi);
\end{align}

This chain starts with the $\w$-hard  problem $\NUM{}\clique{}$. Further, it is well-known that
this problem has tight lower bounds under ETH. The proof of the following statement can be found in \cite[Lemma B.2]{edge_monotone}.

\begin{lemmaq}[{Modification of 14.21 in \cite{param_algo}}]\label{theo:k-clique:ETH}
    Assuming ETH, there is a constant $\alpha > 0$ such that for $k \geq 3$, no
    algorithm (that reads the whole input) solves $k$-$\clique$ on graph $G$ in time $O(|V(G)|^{\alpha k})$.
\end{lemmaq}

The first step of (\ref{reduction:chain}) is a reduction from $\NUM{}\clique$ to $\NUM{}\cphomsprob(\mathcal{H})$. To that end, we use the
following statement that allows us to count
$k$-cliques by using an oracle for $\NUM{}\cphoms{F}{\star}$ whenever $F$ contains $K_{k, k}$ as a subgraph.

\begin{lemmaq}[Modification of {\cite[Lemma 11]{alge}}]\label{lem:tight:biclique:construction}
    There is an algorithm that
    given a positive integer $\ell > 1$, a graph $F$ (that contains $K_{\ell, \ell}$ as assuming
    subgraph), and a graph $G$; computes an $F$-colored graph $G'$ with $2 \ell |V(G)|
    + (|V(F)| - 2 \ell)$ vertices. Further, the number of cliques of size $\ell$ in $G$
    equals $\NUM{}\cphoms{F}{G'}$.
    The running time of the algorithm is $O(|V(F)|^{|V(F)| + 2} + |V(F)|^2 |V(G)|^2 )$.
\end{lemmaq}

\cref{lem:tight:biclique:construction} allows us to focus on a reduction from
$\NUM{}\cphomsprob(\mathcal{H})$ to $\NUM{}\indsubsprob(\Phi)$. To that end we use the following
 key lemma from \cite{alge} that we adopt for graph parameters.

\begin{lemma}[Modification of {\cite[Lemma 18]{alge}}]\label{lem:chi:cpindub}
    Let $H$ denote a graph, let $\Phi \colon \graphs{} \to \Q$ denote a graph parameter,
    and let $G$ denote an $H$-colored graph.
    Then, we have
    \begin{align*}
        \NUM{}\cpindsubs{(\Phi, H)}{G}
        &= \sum_{S \subseteq E(H)} \Phi(\ess{H}{S}) \sum_{J \subseteq E(H) \setminus S}
        (-1)^{\NUM{}J} \cdot \NUM{}\cphoms{\ess{H}{S \cup J}}{G} \\
        &= \sum_{A \subseteq E(H)}  (-1)^{\#A} \cdot \ae{\Phi}{\ess{H}{A}} \cdot \NUM{}\cphoms{\ess{H}{A}}{G}.
    \end{align*}
    Further, the absolute values of $\ae{\Phi}{H}$ and of the coefficient of
    $\NUM{}\cphoms{H}{G}$ are equal.
\end{lemma}
\begin{proof}
    We define $k \coloneq \#V(H)$ and use \cref{lem:param:to:property} to rewrite
    the equation from \cite[Lemma 18]{alge} that was derived for graph properties. We obtain
    \begin{align*}
        \NUM{}\cpindsubs{(\Phi, H)}{G} &\stackrel{(\ref{eq:cpindsub:param})}{=} \sum_{b \in \image{\Phi}{k}} b \cdot \NUM{}\cpindsubs{(\Phi^b, H)}{G} \\
        &\overset{(\text{Lemma 18})}{=} \sum_{b \in \image{\Phi}{k}} b \sum_{S \subseteq E(H)} \Phi^b(\ess{H}{S}) \sum_{J \subseteq E(H) \setminus S}
        (-1)^{\NUM{}J} \cdot \NUM{}\cphoms{\ess{H}{S \cup J}}{G} \\
        &= \sum_{S \subseteq E(H)}  \sum_{J \subseteq E(H) \setminus S}
        (-1)^{\NUM{}J} \cdot \NUM{}\cphoms{\ess{H}{S \cup J}}{G} \sum_{b \in \image{\Phi}{k}} b \cdot \Phi^b(\ess{H}{S}) \\
        &\stackrel{(\ref{eq:phi:param})}{=} \sum_{S \subseteq E(H)}  \sum_{J \subseteq E(H) \setminus S}
        (-1)^{\NUM{}J} \cdot \NUM{}\cphoms{\ess{H}{S \cup J}}{G} \cdot \Phi(\ess{H}{S})               \\
        &= \sum_{S \subseteq E(H)} \Phi(\ess{H}{S}) \sum_{J \subseteq E(H) \setminus S}
        (-1)^{\NUM{}J} \cdot \NUM{}\cphoms{\ess{H}{S \cup J}}{G}
    \end{align*}
    Thus $\NUM{}\cpindsubs{(\Phi, H)}{G}$ is equal to a linear combination of $\NUM{}\cphoms{\ess{H}{A}}{G}$ for $A \subseteq E(H)$.
    By collecting terms, we obtain
    \[\NUM{}\cpindsubs{(\Phi, H)}{G}
    = \sum_{A \subseteq E(H)}   \left(\sum_{S \subseteq A}\Phi(\ess{H}{S}) (-1)^{\#A - \#S} \right) \cdot \NUM{}\cphoms{\ess{H}{A}}{G}.  \]
    By using the definition of the alternating enumerator, we can further rewrite this into
    \begin{align*}
        \NUM{}\cpindsubs{(\Phi, H)}{G} = \sum_{A \subseteq E(H)}  (-1)^{\#A} \cdot \ae{\Phi}{\ess{H}{A}} \cdot \NUM{}\cphoms{\ess{H}{A}}{G},
    \end{align*}
    which shows that the absolute values of $\ae{\Phi}{H}$ and of the coefficient of
    $\NUM{}\cphoms{H}{G}$ are equal.
\end{proof}

In particular, the second part of \cref{lem:chi:cpindub} yields
that the term $\NUM{}\cphoms{H}{G}$ appears in the of \cref{lem:chi:cpindub} if and only if $\ae{\Phi}{H} \neq 0$.
Observe that \cref{lem:chi:cpindub} itself does not suffice to obtain the claimed reduction
\(\NUM{}\cphomsprob(\mathcal{H}) \fpt \NUM{}\cpindsubsprob(\Phi)\):
the oracle for $\NUM{}\cpindsubsprob(\Phi)$ computes only a sum in which $\NUM{}\cphoms{H}{G}$ occurs
as some term---we still need to extract the value $\NUM{}\cphoms{H}{G}$ out of the result of
the oracle.
Fortunately for us, \cite[Lemma~7]{alge} does exactly that by showing a generalization of
the \emph{Complexity Monotonicity} of~\cite{hom:basis}. This means that the computation
of a finite sum $\sum_{H \in \graphs{}} \alpha(H) \cdot \cphoms{H}{\star}$ is precisely as hard as the computation of the hardest term with
nonvanishing coefficient. This result is summarized in the following lemma.

\begin{lemma}[Modification of {\cite[Lemma A.7]{edge_monotone}}]\label{lem:cphom:to:indsub}\label{claim:hom:tw}
    Let $\Phi \colon \graphs{} \to \Q$ denote a graph parameter and
    let $H$ be a $k$-vertex graph with $\ae{\Phi}{H} \neq 0$. Any algorithm that
    computes for each $H$-colored \(G'\) the number $\NUM{}\cpindsubs{(\Phi, H)}{G'}$
    in time $O(f(k) \cdot |V(G')|^\beta)$ implies an algorithm that for each graph
    \(G\)
    \begin{itemize}
        \item computes the number $\NUM{}\cphoms{H}{G}$ in time $O(h(k)\cdot |V(G)|^{\beta+1})$
        for some computable function $h$,
        \item and our algorithm queries $\NUM{}\cpindsubs{(\Phi, H)}{G'}$ only on
     $H$-colored graphs $G'$ with $V(G') \leq f(k) \cdot |V(G)|$, for some computable function $f$.
    \end{itemize}
\end{lemma}
\begin{proof}
    Write \(k \coloneqq |V(H)|\). Suppose that we are given a graph \(G\) that is $H$-colored via \(c\).
    We wish to compute the number \(\NUM{}\cphoms{H}{G}\).
    \begin{enumerate}[(1)]
        \item First, we use \cref{lem:chi:cpindub} to write $\NUM{}\cpindsubs{(\Phi, H)}{\star}$ as
        \[\sum_{A \subseteq E(H)} (-1)^{\#A} \cdot \ae{\Phi}{\ess{H}{A}} \cdot \NUM{}\cphoms{\ess{H}{A}}{\star},\]
        where
        the term $(-1)^{\#E(H)} \cdot \ae{\Phi}{H}  \cdot \NUM{}\cphoms{H}{\star}$ appears with non-zero coefficient
        since $H$ is nonvanishing. According to \cite[Lemmas~7]{alge}, we can compute $\NUM{}\cphoms{H}{G}$ in time $O(f(k) \cdot |V(G)|)$ using an oracle for $\NUM{}\cpindsubs{(\Phi, H)}{\star}$. Here $f$ is a computable function.
        In said reduction, each graph $G'$ that is used inside an oracle call satisfies
        $|V(G')| \leq f(k) \cdot |V(G)|$.

        \item Second, we compute for each \(G'\) the number $\NUM{}\cpindsubs{(\Phi, H)}{G'}$
            in time  $O(|V(G')|^{\beta})$ using the algorithm which we
            assumed to exist.
    \end{enumerate}
    We combine the two steps to obtain an algorithm that computes $\NUM{}\cphoms{H}{G}$
    in time $O(f(k) \cdot |V(G)| \cdot (f(k)\cdot |V(G)|)^{\beta} )$,
    which can be rewritten into $O(h(k) \cdot |V(G)|^{\beta + 1})$ for some computable function $h$. Finally,
    observe that we query our algorithm only on graphs $G'$ with $V(G') \leq f(k) \cdot |V(G)|$.
\end{proof}

Let us quickly recap what we have done. We first showed that we can count the number of $k$-cliques in a graph $G$ by finding
a graph $H$ that contains $K_{k, k}$ and computing $\NUM{}\cphoms{H}{G'}$ for some auxiliary graph $G$.
Further, we showed that we can compute $\NUM{}\cphoms{H}{G'}$
using a oracle for $\NUM{}\cpindsubs{(\Phi, H)}{\star}$ whenever the alternating enumerator of $H$ is nonvanishing.
That is, we can count cliques by finding graphs that contain a biclique and have a nonvanishing alternating enumerator.
Since there are tight bounds for counting cliques assuming ETH, we obtain the following result.

\begin{theorem}[Modification of \cite{alge}]\label{theo:lower:bound:bicliques}
    Assuming ETH, there is a global constant $\beta > 0$ and a positive integer $N$ such that for all
    graph parameters $\Phi \colon \graphs{} \to \Q$, functions $h$, numbers $k$ with
    \begin{itemize}
        \item $h(k) \geq N$
        \item there is a graph $F$ with $k$ vertices and $\ae{\Phi}{F} \neq 0$,
        \item and $F$ contains $K_{h(k), h(k)}$ as a subgraph
    \end{itemize}
    there is no algorithm (that reads the whole input) that for every \(G\) computes
    $\NUM{}\cpindsubs{(\Phi, \graphs{k})}{G}$ in time $O(|V(G)|^{\beta h(k)})$. Here,
    input consists of a $k$-vertex graph $H$ and $H$-colored graph $G$.
\end{theorem}
\begin{proof}
    We show how to use $\NUM{}\cpindsubs{(\Phi, \graphs{k})}{\star}$ to solve $h(k)$-\textsc{Clique}.

    \begin{claim}\label{claim:clique:to:indsub}
        For a fixed $k$ such that there is a
        graph $F$ with $k$ vertices, $\ae{\Phi}{F} \neq 0$, and $F$ contains $K_{h(k),
        h(k)}$ as a subgraph; if there is an algorithm that computes for each
        $F$-colored graph $G'$ the value $\NUM{}\cpindsubs{(\Phi,
        F)}{G'}$ in time $O(|V(G')|^{\gamma})$, then $h(k)$-\textsc{Clique} can be
        computed for each graph $G$ in time
        $O(|V(G)|^{\gamma + 1})$.
    \end{claim}
    \begin{claimproof}
        We construct an algorithm
        that solves $h(k)$-\textsc{Clique} using an oracle for $\NUM{}\cpindsubs{(\Phi,
        F)}{\star}$. Since $F$ is fixed, we can assume that our algorithm knows
        $F$. Fix a graph $G$, we use the algorithm from
        \cref{lem:tight:biclique:construction} to construct a graph $G'$ such that $G$
        contains a $h(k)$-clique if and only if $\NUM{}\cphoms{F}{G'}$ is not zero. The
        running time of
        this construction is in $O(|V(G)|^2)$ since $|V(F)| = k$ is constant. Further,
        we obtain $|V(G')| \leq 2 h(k) |V(G)| + k = O(|V(G)|)$.

        If we can compute $\NUM{}\cpindsubs{(\Phi, F)}{|V(G')|}$ in time $O(|V(G')|^{\gamma})$, then we
        can use \cref{lem:cphom:to:indsub} to compute $\NUM{}\cphoms{F}{G'}$ in time $O((2
        h(k) |V(G)| + k)^{\gamma + 1})$. Thus, we can solve $h(k)$-\textsc{Clique} in time
        $O(|V(G)|^2 + (2 h(k) |V(G)| + k)^{\gamma + 1})$ which is in $O(|V(G)|^{\gamma + 1})$ since $k$
        is fixed. Note, that we can assume that $\gamma \geq 1$ since otherwise the algorithm
        for $\NUM{}\cpindsubs{(\Phi, F)}{\star}$ would be sublinear. 
    \end{claimproof}
    By \cref{theo:k-clique:ETH}, there is a constant $\alpha >
    0$ such that no algorithm solves $h(k)$-\textsc{Clique} in time
    $O(|V(G)|^{\alpha h(k)})$ for a fixed $h(k) \geq 3$ unless ETH fails. Set $\beta \coloneqq
    \alpha/2$ and $N \coloneqq \max(3, 2/\alpha)$.
    If there is a $k$ with $h(k) \geq N \geq 3$,
    a graph $F$ with $k$ vertices such that $\ae{\Phi}{F} \neq 0$,
    and $F$ contains $K_{h(k), h(k)}$ as a subgraph, and an algorithm that solves $\NUM{}\indsubs{(\Phi,
    k)}{G'}$ in time $O(|V(G')|^{\beta h(k)})$, then we can use \cref{claim:clique:to:indsub}
    to solve $h(k)$-\textsc{Clique} in time $O(|V(G)|^{\beta h(k) + 1})$. Observe that $\beta
    h(k) + 1 \leq \alpha h(k)$ for $h(k) \geq 2/\alpha$. Thus, we can use solve
    $h(k)$-\textsc{Clique} in time $O(n^{\alpha h(k)})$. Hence, ETH fails.
\end{proof}

The last results show how to use an algorithm for $\NUM{}\cpindsub(\Phi)$
to solve $\NUM{}$-$\clique$. Now, the last step of our reduction chain (\ref{reduction:chain})
is to prove a reduction form $\NUM{}\cpindsub(\Phi)$ to $\NUM{}\indsubsprob(\Phi)$.

\begin{lemma}[Modification of {\cite[Lemma 10]{alge}}]\label{lem:redu:cpindsub:indsub}
    Given a graph parameter $\Phi \colon \graphs{} \to \Q$ and let $H$ be a $k$-vertex graph.
    Assume that there is an algorithm $\mathbb{A}$ that computes on input $(G, k)$
    the number $\NUM{}\indsubs{(\Phi, k)}{G}$ in time $g(k) \cdot |V(G)|^{\beta}$. Then, there
    is an algorithm $\mathbb{B}$ that given an $H$-colored graph $G$ as input,
    computes $\NUM{}\cpindsubs{(\Phi, H)}{G}$ in time $O(2^k \cdot g(k) \cdot |V(G)|^{\beta+2})$.
\end{lemma}
\begin{proof}
    Let $H \in \mathcal{H}$ be a graph and let $G$ be a $H$-colored graph with coloring $c$.
    Using~\cite[equation (19)]{alge}, we can rewrite each $\NUM{}\cpindsubsprob(\Phi^b, \mathcal{H})$
    as a sum of $\NUM{}\indsubs{\Phi, k}{G_J}$ terms. We combine this with (\ref{eq:cpindsub:param}) to obtain
    \begin{align*}
        \NUM{}\cpindsubs{(\Phi, H)}{G}
    &\stackrel{(\ref{eq:cpindsub:param})}{=} \sum_{b \in \image{\Phi}{k}} b \cdot \NUM{}\cpindsubs{\Phi^b, H}{G} \\
    &\stackrel{(19)}{=} \sum_{b \in  \image{\Phi}{k}} b \sum_{J \subseteq \setn{k}} (-1)^{\#J} \cdot \NUM{}\indsubs{\Phi^b, k}{G_J} \\
    &= \sum_{J \subseteq \setn{k}} (-1)^{\#J} \sum_{b \in \image{\Phi}{k}} b \cdot \NUM{}\indsubs{\Phi^b, k}{G_J} \\
    &\stackrel{(\ref{eq:indsub:param})}{=} \sum_{J \subseteq \setn{k}} (-1)^{\#J} \NUM{}\indsubs{\Phi, k}{G_J},
    \end{align*}
    where $G_J$ is the graph the is obtained from $G$ by deleting all vertices whose color is in $J$. Hence,
    $\NUM{}\cpindsubs{(\Phi, H)}{G}$ can be computed by using algorithm $\mathbb{A}$ at most $2^k$ times. Each
    graph $G_J$ can be computed in time $|V(G)|^2$ and we can compute $\NUM{}\indsubs{\Phi, k}{G_J}$ in
    $O(g(k) \cdot |V(G)|^{\beta})$.
\end{proof}

Lastly, we combine the results of this section to obtain a corollary that we can use for proving hardness 
and to get tight lower bounds under ETH. In the following, we write $h(G_k)$ instead of $h(|V(G_k)|)$ 
whenever $h \colon \nat \to \nat$ is a function that maps natural numbers to natural numbers.

\cortightbounds
\begin{proof}
    \begin{itemize}
        \item First, we show $\w$-hardness by constructing a parameterized Turing reduction from 
        $\NUM{}\clique{}$ to $\NUM{}\cpindsub(\Phi)$. Given a graph $G$ and a number $k$. We know 
        that there exists a nonvanishing graph $F$ that contains $K_{k, k}$ is a biclique since 
        there exists a graph $G_j \in A$ with $h(G_j) \geq k$. We can find $F$ by simply 
        iterating over all possible graphs (ordered by the number of vertices) and search for the
        first one that has these properties. The running time of this is in $O(h_1(k))$ for some 
        computable function $h_1$. Further, we can use \cref{lem:tight:biclique:construction} to 
        compute a $F$-colored graph $G'$ with $2k|V(G)| + (|V(F)| - 2k)$ many vertices in time
        $O(h_2(k) \cdot |V(G)|^2)$ such that the number of  $k$-clique in $G$ is equal to 
        $\NUM{}\cphoms{F}{G'}$. \cref{claim:hom:tw} ensures that $\NUM{}\cphoms{F}{G'}$ 
        can be computed by an algorithm that can queries $\NUM{}\cpindsub(\Phi)$ only
        on $F$ and $F$-colored graphs with size at most $h_3(k) \cdot |V(G)|$, 
        for some computable function $h_3(k)$. This shows that there is a parameterized Turing
        reduction from $\NUM{}\clique{}$ to $\NUM{}\cpindsub(\Phi)$, thus $\NUM{}\cpindsub(\Phi)$ is $\w$-hard. 
        Lastly, we use \cref{lem:redu:cpindsub:indsub} to get a parameterized Turing reduction from $\NUM{}\cpindsub(\Phi)$
        to $\NUM{}\indsubsprob(\Phi)$ which shows that $\NUM{}\indsubsprob(\Phi)$ is $\w$-hard, too.

        \item Let $N$ and $\beta$ be the constants from \cref{theo:lower:bound:bicliques}. Since $h(G) \in \Omega(|V(G)|^\psi)$, 
        there is a constant $c$ and integers $M$ such that $h(G) \geq |V(G)|^\psi/c$ for all $|V(G)| \geq M$. We define 
        $\gamma' \coloneqq \min(\beta/c, \beta/2) $. Now, let $k$ be some value such that $k = |V(G_j)|$
        for some $j$. If $k \geq C$, where $C \coloneqq \max((c N)^{1/\psi}, N, M, 4/\beta)$, then we can use 
        \cref{theo:lower:bound:bicliques} to show that no algorithm that computes for each $k$-vertex 
        graph $H$ and $H$-colored graph $G$ the value $\NUM{}\cpindsubs{(\Phi, H)}{G}$ in time 
        $O(|V(G)|^{\beta h(k)})$ unless ETH fails. This means that no algorithm solves the problem in
        time $O(|V(G)|^{\gamma' k^\psi})$ since $\gamma' k^\psi \leq \beta  k^\psi /c \leq \beta h(k)$.

        To obtain our lower bound also for \(k < C\), we set $\gamma \coloneqq \min(\gamma', 1/( C +1))$.
        Observe that for \(k < C\), we obtain
        \[
            O(|V(G)|^{\gamma k^\psi}) = o( |V(G)| ).
        \]
        since $k^\psi \leq k$. Now, this running time is unconditionally unachievable
        for an algorithm that reads the whole input. 
        
        Lastly, assume that we could compute the value
        $\NUM{}\indsubs{(\Phi, k)}{G}$ in time $O(|V(G)|^{\gamma h(k)})$. If $k \geq C$ we can use 
        \cref{lem:redu:cpindsub:indsub} to get an algorithm that computes $\NUM{}\cpindsubs{(\Phi, H)}{G}$
        in time $O(|V(G)|^{\gamma h(k) + 2})$. However, now we obtain that
        $\gamma h(k) + 2 \leq \beta/2 \cdot h(k) + 2 \leq \beta h(k)$ for
        $h(k) \geq C \geq 4 / \beta$. Otherwise $k < C$ and we would obtain a sublinear algorithm for
        $\NUM{}\indsubs{(\Phi, k)}{G}$ which is unconditionally unachievable
        for any algorithm that reads the whole input.

        \item This part is a direct application of the previous part.
            \qedhere
    \end{itemize}
\end{proof}

Our \cref{theo:edge_mono} only covers nontrivial edge-monotone graph parameters with
finite codomain. However, there are also many interesting graph parameters with
unbounded codomain. 

\begin{theorem}\label{theo:hard:common}
    For the following three graph parameter the problem $\indsubsprob(\Phi_j)$ and $\NUM{}\cpindsub(\Phi_j)$ are both \w-hard and
    there exists a constant $\alpha_j$ such that for all $k \geq 3$ no algorithm (that
    reads the whole input) computes for all $G$ the value $\indsubs{(\Phi_j, k)}{G}$ (or $\cpindsubs{(\Phi_j, \graphs{k})}{G}$)
    in time $O(|V(G)|^{\alpha_j k})$.
    \begin{enumerate}
        \item $\Phi_1(G) = \Delta(G)$, where $\Delta(G)$ is the maximal degree of $G$
        \item $\Phi_2(G) = $ number of connected components of $G$
        \item $\Phi_3(G) = \chi(G)$, where $\chi(G)$ is the chromatic number of $G$
    \end{enumerate}
\end{theorem}
\begin{proof}
    Let us first consider $\Phi_1$. For this let $p$ be an odd prime number, we consider the fixed points $\fp(\sylow_{p}, K_{p})$. 
    By \cref{lemma:fixed:points:p^m}, we know that these fixed points are precisely the difference graphs $\cgr{p}{A}$
    for $A \subseteq [s]$, where $s \coloneqq (p-1)/2$. Observe that $\Phi_1(\cgr{p}{A}) = 2 |A|$. Thus 
    \[\ae{\Phi_1}{K_p} = \ae{\Phi_1}{\cgr{p}{[s]}} \equiv_p \sum_{k = 0}^s 2 k \binom{s}{k} \equiv_p s 2^s,\]
    where the last equation follows from a simple induction. Observe that $s 2^s \not \equiv_p 0$ since $2 \not \equiv_p 0$
    and $s \not \equiv_p 0$. Hence, for each prime $p$, we get that $K_p$ is nonvanishing and it is easy to see that 
    $K_p$ contains $K_{s,s}$ as a subgraph. This means that we can use \cref{cor:tight:bounds} to show
    that $\NUM{}\indsubsprob(\Phi_1)$ and $\NUM{}\cpindsub(\Phi_1)$ are both \w-hard and have tight bounds under ETH.

    Next, for all natrual numbers $t$, we consider the fixed points $\fp(\sylow_{2^t}, K_{2^t})$. 
    By \cref{lemma:fixed:points:p^m}, we know that these fixed points are precisely
    the $t$-lexicographic product of difference graphs $\cgr{2}{A_1} \wrGraph \cdots \wrGraph \cgr{2}{A_t}$. 
    Let $s = 2^{t-1}$, and observe that 
    $K_{s, s} = \cgr{2}{\{1\}} \wrGraph \cgr{2}{\emptyset} \wrGraph  \cdots \wrGraph \cgr{2}{\emptyset}$
    and $\IS_{2^t} = \cgr{2}{\emptyset} \wrGraph  \cdots \wrGraph \cgr{2}{\emptyset}$. 
    Thus $\ae{\Phi_j}{K_{s, s}} \equiv_2 \Phi_j(K_{s, s}) + \Phi_j(\IS_{2^t})$. It is now easy to verify that 
    $\ae{\Phi_2}{K_{s, s}} \equiv_2 1 + 2^t \equiv_2 1$ and $\ae{\Phi_3}{K_{s, s}} \equiv_2 2 + 1 \equiv_2 1$.
    Hence, we can use \cref{cor:tight:bounds} to show
    that $\NUM{}\indsubsprob(\Phi_j)$ and $\NUM{}\cpindsub(\Phi_j)$ are both \w-hard and have tight bounds under ETH.
\end{proof}

\section{Hardness for edge-monotone graph parameters with finite codomain}\label{sec:edge_monotone}

To goal of this section is to show that for all edge-monotone nontrivial graph parameters with finite codomain the problems
$\NUM{}\indsubsprob(\Phi)$  and $\NUM{}\cpindsub(\Phi)$ are both $\w$-hard and have tight bounds assuming ETH.

\thmedgemon*

\subsection{Tight bounds for edge-monotone graph parameters with finite codomain on prime powers}\label{sec:prime:power}

The goal of this section is to prove $\w$-hardness and tight bounds for $\NUM{}\indsubsprob(\Phi)$ and
$\NUM{}\cpindsub(\Phi)$ when $\Phi$ is nontrivial on prime powers. We archive this by finding
nonvanishing graphs $H$ that contain large bicliques. To find these graphs we analyze the fixed points
of the $p$-Sylow group of $\aut(K_{p^t})$ (see \cref{sec:sylow}). We exploit the fact that these fixed
points contain large bicliques whenever they have a small empty-prefix (see \cref{lem:treewidth:wreath:level}).
Thus, our goal is to find a nonvanishing fixed point $H = \cgr{p}{A_1} \wrGraph \cdots \wrGraph  \cgr{p}{A_{t}}$ with empty-prefix
$\wrLevel(A_1, \dots, A_{t}) = 0$, which is done using the following approach.

Set $z \coloneqq \Phi(\IS_{p^t})$, then we search for the smallest $\ell$ such that there is a fixed point $H$
of level $\ell$ with $\Phi(H) < z$. Next, we show that $\wrLevel(A_1, \dots, A_{t}) = 0$ which is done by using
that graphs with a high empty-prefix are edge-subgraphs of graphs with small empty-prefix. Lastly, we use
\cref{lem:ae:nonvanishing:sub-points} to get that $\ae{\Phi}{H}$ is nonvanishing. The proof is similar to
the proof of Theorem 7.18 given in \cite{edge_monotone}.

\begin{theorem}[Modification of {\cite[Theorem 7.18]{edge_monotone}}]\label{theo:edge_mono:prime:power:bicliques}
    Let $\Phi \colon \graphs{} \to \{0, 1, \dots, c\}$ denote an edge-monotone graph parameter that is nontrivial on a prime power
    $p^t$ with $c < p$, then there is a nonvanishing fixed point $H$ of $\sylow_{p^t}$ in $K_{p^t}$ that
    contains $K_{p^{t-1}, p^{t-1}}$ as a subgraph. In particular, $\ae{\Phi}{H} \not \equiv_p 0$.
\end{theorem}
\begin{proof}
    Set $z \coloneq \Phi(\IS_{p^m})$.
    Our goal is to show that there is a fixed point $H$ with the following properties

    \begin{itemize}
        \item $H = \cgr{p}{A_1} \wrGraph \cdots \wrGraph  \cgr{p}{A_{t}}$ for
            $A_i \subseteq \field{p}^+$,
        \item $\Phi(H) < z$,
        \item $\wrLevel(A_1, \dots, A_{t}) = 0$, and
        \item $\Phi(\Tilde{H}) = z$ for all proper sub-points $\Tilde{H}$ of $H$.
    \end{itemize}

    As \(\Phi\) is nontrivial on \(p^t\), we have $\Phi(K_{p^t}) < z$. Further, we have
    and \(\Phi(\IS_{p^t}) = \Phi(\cgr{p}{\emptyset} \wrGraph \cdots \wrGraph \cgr{p}{\emptyset}) = z\).
    Now, let $i$ denote the smallest value such that
    there is a fixed point $H$ of level $i$ with \(\Phi(H) < z\),
    but all fixed points of level smaller than $i$ have value \(z\).
    As \(\Phi(\IS_{p^t}) = z\), we have $i > 0$.

    \begin{claim}\label{11.13-10}
        There is a fixed point $H$ of
        level $i$ with $\Phi(H) < z$ and whose empty-prefix is zero.
    \end{claim}
    \begin{claimproof}
        Toward an indirect proof,
        assume that all fixed points $\cgr{p}{A_1} \wrGraph \cdots \wrGraph \cgr{p}{A_{t}}$
        of level $i$ with $\wrLevel(A_1, \dots, A_{t}) = 0$ have value \(z\).
        We show that now, all fixed
        points of level $i$ have value $z$, which is a contradiction to our choice of $i$.

        To that end, fix an $F \coloneqq
        \cgr{p}{A_1} \wrGraph \dots \wrGraph \cgr{p}{A_{t}}$.
        If $\wrLevel(A_1, \dots, A_t) = 0$, then $\Phi(F) = z$ due to our assumption.
        Otherwise, $\wrLevel(A_1, \dots, A_t) > 0$, which means that $F$ has the
        form $F = \cgr{p}{\emptyset} \wrGraph \cdots \wrGraph \cgr{p}{\emptyset} \wrGraph
        \cgr{p}{A_1} \wrGraph \dots \wrGraph \cgr{p}{A_{t - j}}$ and is
        thus, by \cref{lem:wreath:sub:iso},
        isomorphic to an edge-subgraph of
        \[
            F' \coloneqq \cgr{p}{A_1} \wrGraph \cdots \wrGraph F_p^{A_{t - j}} \wrGraph \cgr{p}{\emptyset}
            \wrGraph \cdots \wrGraph \cgr{p}{\emptyset}.
        \]
        By \cref{lem:level:sylow:graphs},
        the level of $F$ is equal to the level of $F'$.
        Thus by assumption, $F'$ has value $z$, hence the value of $F$ is at least as large as $z$ since $\Phi$ is edge-monotone.
        This implies that $F$ has value $z$.
    \end{claimproof}

    By \cref{11.13-10},
    we may assume that there is a fixed point $H \coloneqq \cgr{p}{A_1} \wrGraph \cdots \wrGraph \cgr{p}{A_{t}}$ of
    level $i$ with $\wrLevel(A_1, \dots, A_{t}) = 0$ with value less than $z$.
    Now, as the empty-prefix of \(H\) is zero, \cref{lem:treewidth:wreath:level} yields that
    $H$ contains $K_{p^{t-1}, p^{t-1}}$ as a subgraph.
    Finally, the fact that $c < p $ allows us to apply
    \cref{lem:ae:nonvanishing:sub-points} which shows that $\ae{\Phi}{H} \not \equiv_p 0$. This completes the proof.
\end{proof}

\subsection{Reduction to Smaller Size Cases}\label{sec:red}

In this section, we describe a reduction from
$\NUM{}\cpindsubs{(\reductionGraph{\Phi}{C}{H_2, \dots, H_s}, \graphs{ p^t + c})}{\star}$ to
$\NUM{}\cpindsubs{(\Phi, \graphs{p^t})}{\star}$, where $\reductionGraph{\Phi}{C}{H_2, \dots, H_s}$ is a specific graph
parameter that is defined using $\Phi$ and inhabited graphs. We use said reduction to get from a value
$k = p^t +c$ to a prime power $p^t$.

\lemincexc
\begin{proof}
    Let $H$ be a graph and $G$ be a $H$-colored graph. We define a graph $\Tilde{H}$ with
    $V(\Tilde{H}) \coloneqq V(H) \unionSet V(H_2) \unionSet \dots \unionSet V(H_s)$ and the edge set of
    $\Tilde{H}$ consists of all edges in $H$. Further, all vertices $x$ of $\Tilde{H}$ with $x \notin H$
    are connected to all other vertices (including the vertices in $H$).
    Next, we define $\Tilde{G} \coloneqq \metaGraph{C}{G, H_2 \dots, H_s}$. We rename the
    vertices in $\Tilde{G}$ such that $V(\Tilde{G}) = V(G) \unionSet V(H_2) \unionSet \dots \unionSet V(H_s)$.
    Since $G$ is a $H$-colored graph
    we get a homomorphism $c \colon V(G) \to V(H)$ as part of our input. We extend $c$ by defining
    the following homomorphism
    \[\Tilde{c} \colon V(\Tilde{G}) \to V(\Tilde{H}), x \mapsto \begin{cases}
        c(x), \text{if $x \in V(G)$} \\
        x, \text{otherwise}
    \end{cases}. \]
    This means that $\Tilde{G}$ is a $\Tilde{H}$-colored graph.
    Observe that we can compute
    $\Tilde{H}$ in time $O((c + |V(H)|)^2)$, $\Tilde{G}$ in time $O((c + |V(G)|)^2)$ and $\Tilde{c}$ in time $O((c + |V(G)|)^2)$.
    Further, without loss of generality, we can assume that $|V(H)| \leq |V(G)|$.

    Lastly, we claim that $\NUM{}\cpindsubs{(\Tilde{\Phi}, H)}{G} = \NUM{}\cpindsubs{(\Phi, \Tilde{H})}{\Tilde{G}}$, thus we can use algorithm $\mathbb{A}$ to compute $\NUM{}\cpindsubs{(\Tilde{\Phi}, H)}{G}$ in time
    $O(g(c + k, c + |V(G)|) + (c + |V(G)|)^2)$. To prove our claim, we start with $\NUM{}\cpindsubs{(\Tilde{\Phi}, H)}{G}$ and unfold the definition of $\Tilde{\Phi}$. Observe that an induced subgraph $G[A]$
    is colorful with respect to $c$ if $G[A]$ hits all color classes meaning that $c(G[A]) = V(H)$.
    \[\NUM{}\cpindsubs{(\Tilde{\Phi}, H)}{G} = \sum_{\substack{A \in \binom{V(G)}{|V(H)|} \\
    c(G[A]) = V(H)}} \Phi(\metaGraph{C}{G[A], H_2, \dots, H_s})\]
    Next, we define $X \coloneqq V(H_2) \unionSet \dots \unionSet V(H_s)$ and observe
    that $c(G[A]) = V(H)$ is equivalent to $\Tilde{c}(\Tilde{G}[A \unionSet X]) = V(\Tilde{H})$ since
    $\Tilde{c}(\Tilde{G}[A \unionSet B]) = c(G[A]) \unionSet B$, for all $B \subseteq X$. Lastly,
    observe that $\Tilde{G}[A \unionSet X] = \metaGraph{C}{G[A], H_2, \dots, H_s}$.
    Thus we can continue our computation to obtain
    \[\sum_{\substack{A \in \binom{V(G)}{|V(H)|} \\
    c(G[A]) = V(H)}} \Phi(\metaGraph{C}{G[A], H_2, \dots, H_s}) = \sum_{\substack{A \in \binom{V(G)}{|V(H)|}  \\
    \Tilde{c}(\Tilde{G}[A \unionSet X]) = V(\Tilde{H})}} \Phi(\Tilde{G}[A \unionSet X]) =  \NUM{}\cpindsubs{(\Phi, \Tilde{H})}{\Tilde{G}}.
    \qedhere
\]
\end{proof}

\subsection{Finding Large Bicliques in Graphs}\label{sec:large:bicliques}

Following \cite[Section 6.4]{edge_monotone}, we introduce the following definition:

\defconsca

We show that each edge-monotone graph parameter $\Phi \colon \graphs{k} \to \{0, \dots, c\}$, with $c < p$,
that is nontrivial on $k \geq p^{t+1} + p^t$ is always \emph{concentrated on $k$ with respect to $p^t$}
or \emph{reducible on $k$ with respect to $p^{t+1}$}.

\lemlargebi
\begin{proof}
    First, we define $z \coloneq \Phi(\IS_k)$.
    Let $\sum_{i = 0}^d s_i \cdot p^i$ be the base $p$ representation of $(k - p^{t+1})$,
    that is $s_i \in \{0, \dots, p-1\}$ and $s_d \neq 0$.\footnote{
    For example, the base $3$ representation of 34 is $1 \cdot 3^0 + 2 \cdot 3^1 + 0 \cdot 3^2 + 1 \cdot 3^3$.}
    Observe that $p^d \geq p^{t}$. We define $s \coloneqq \sum_{j = t}^d s_j \cdot p^{j-t}$ and
    $r \coloneq (k-p^{t+1}) - s \cdot p^t$. Observe that $s \geq 1$, $k = p^{t+1} + s \cdot p^t + r$ and $r < p^t$.
    The idea is that $s$ represents the digits of $(k - p^{t+1})$ after the $t$ position.
    Our goal is to prove that $\Phi$ is concentrated on $k$ with respect to $p^t$ or reducible on $k$ with respect to $p^{t+1}$.
    For this, we assume that $\Phi$ satisfies neither of these properties and show that $\Phi(K_k) = z$.
    This would imply that $\Phi$ is trivial on $k$.

    Next, we consider the $p$-group\footnote{$\sylow_{1}$ is simply the trivial group
    containing one element. The group is isomorphic to $\aut(K_1) = \aut(\IS_1) = \sym{1}$. }
    \[\Gamma \coloneq \sylow_{p^{t+1}}  \times \underbrace{(\sylow_{p^t} \times \dots \times
        \sylow_{p^t})}_{s-\text{times}}
    \times \underbrace{(\sylow_{1} \times \dots \times \sylow_{1})}_{r-\text{times}} \]
    together with the fixed points $\fp(\Gamma, K_k)$. By \cref{theo:prod:fixed:points}, we get that
    the set of fixed points is equal to
    \[\fp(\Gamma, K_k) = \{\metaGraph{C}{A^1, A^2 \dots, A^{1+s}, \IS_1, \dots, \IS_1} \colon
    C \in \graphs{1+ s + r}, A^1 \in \fp(\sylow_{p^{t+1}}, K_{p^{t+1}}),  A^i \in \fp(\sylow_{p^t}, K_{p^t}) \}.\]
    Observe that the last $r$ elements of each fixed point are always independent sets of size $1$.
    Hence, we write $\metaGraph{C}{A^1, \dots, A^{1+s}}$ instead of $
    \metaGraph{C}{A^1 \dots, A^{1+s}, \IS_1, \dots, \IS_1}$ from now on.

    In our analysis, we mainly consider fixed points where each $A^i$ is either an independent set
    or a clique. We write $K_{1+s} \UnionGraph \IS_r$ for the graph with $1+s + r$
    vertices and edge set $\{\{i, j\} \colon 1 \leq i < j \leq 1+s \}$. Observe that a graph
    $G = \metaGraph{C}{A^1, \dots, A^{1+s}}$ contains $K_{p^t, p^t}$ as a subgraph whenever $C$
    contains an edge $\{u, v\}$ for $1 \leq u, v \leq 1+s$. This is because $G$ contains two
    components of size $p^t$ (or $p^{t+1}$) that are completely connected to each other.

    \begin{claim}\label{claim:k_s:is_r:z}
        It holds $\Phi(\metaGraph{(K_{1+s} \UnionGraph \IS_r)}{\IS_{p^t}, \dots, \IS_{p^t}}) = z$.
    \end{claim}
    \begin{claimproof}
        Assume that $\Phi(\metaGraph{(K_{1+s} \UnionGraph \IS_r)}{\IS_{p^t}, \dots, \IS_{p^t}}) < z$. Then, there is a graph
        $C \in \graphs{1+s}$ with $\Phi(\metaGraph{(C \UnionGraph \IS_r)}{\IS_{p^t}, \dots, \IS_{p^t}}) < z$ and $C$ has minimal
        number of edges amongst all graphs $C'$ with $\Phi(\metaGraph{(C' \UnionGraph \IS_r)}{\IS_{p^t}, \dots, \IS_{p^t}}) < z$.
        Observe that $C$ contains at least one edge since $\Phi(\metaGraph{\IS_{1+s + r}}{\IS_{p^t}, \dots, \IS_{p^t}}) = z$ and $1+s \geq 2$.
        Further, observe that all proper sub-points of $\metaGraph{(C \UnionGraph \IS_r)}{\IS_{p^t}, \dots, \IS_{p^t}}$ are fixed points of
        the form $\metaGraph{(C' \UnionGraph \IS_r)}{\IS_{p^t}, \dots, \IS_{p^t}}$, where $C'$ is an proper edge-subgraph of $C$ (see \cref{rem:sub-points:IS}).
        This means that $\Phi(\metaGraph{(C' \UnionGraph \IS_r)}{\IS_{p^t}, \dots, \IS_{p^t}}) = z$ for all proper edge-subgraphs. This
        allows us to apply \cref{lem:ae:nonvanishing:sub-points} to show that
        $\Phi(\metaGraph{(C \UnionGraph \IS_r)}{\IS_{p^t}, \dots, \IS_{p^t}})$ is nonvanishing mod $p$. Lastly, observe that
        $\metaGraph{(C \UnionGraph \IS_r)}{\IS_{p^t}, \dots, \IS_{p^t}}$ contains $K_{p^t, p^t}$ as a subgraph, meaning that $\Phi$ is
        concentrated on $k$ with respect to $p^t$. This implies that $\Phi(\metaGraph{(K_{1+s} \UnionGraph \IS_r)}{\IS_{p^t}, \dots, \IS_{p^t}}) = z$
        since we assume that $\Phi$ is not concentrated on $k$ with respect to $p^t$.
    \end{claimproof}

   Next, we use the following result to strengthen \cref{claim:k_s:is_r:z}.

    \begin{claim}\label{claim:complete:meta:graph}
        If $\Phi(\metaGraph{C}{\IS_{p^{t+1}}, \IS_{p^t}, \dots, \IS_{p^t}}) = z$; then either $\Phi(\metaGraph{C}{K_{p^{t+1}}, \IS_{p^t} \dots, \IS_{p^t}}) = z$;
        or $\Phi$ is reducible on $k$ with respect to $p^{t+1}$.
    \end{claim}
    \begin{claimproof}
        Assume otherwise, then there exists an $A \in \fp(\sylow_{p^{t+1}}, K_{p^{t+1}})$
        such that $\Phi(\metaGraph{C}{A, \IS_{p^t}, \dots, \IS_{p^t}}) < z$. We define
        $\Tilde{\Phi}(G) \coloneq \Phi(\metaGraph{C}{G, \IS_{p^t}, \dots, \IS_{p^t}})$.
        Observe that $\Tilde{\Phi}(\IS_{p^{t+1}}) = \Phi(\metaGraph{C}{\IS_{p^{t+1}}, \IS_{p^t}, \dots, \IS_{p^t}}) = z$
        and $\Tilde{\Phi}(A) = \Phi(\metaGraph{C}{A, \IS_{p^t}, \dots, \IS_{p^t}}) < z$.
        Thus $\Tilde{\Phi}$ is nontrivial on $p^{t+1}$ which shows that $\Phi$ is reducible on $k$ with respect to $p^{t+1}$.
    \end{claimproof}

    \cref{claim:complete:meta:graph} shows $\Phi(\metaGraph{(K_{1+s} \UnionGraph \IS_r)}{K_{p^{t+1}}, \IS_{p^t}, \dots, \IS_{p^t}}) = z$.
    Let $D \in \graphs{1+s + r}$ be a graph such that each vertex is connected to each other vertex
    expect for vertex $2$ which is isolated.

    \begin{claim}\label{eq:isolated:1}
        If $\Phi(\metaGraph{(K_{1+s} \UnionGraph \IS_r)}{K_{p^{t+1}}, \IS_{p^t}, \dots, \IS_{p^t}}) = z$ then
        $\Phi(\metaGraph{D}{\IS_{p^{t+1}}, \IS_{p^t} \dots, \IS_{p^t}}) = z$
    \end{claim}
    \begin{claimproof}
        First, observe that the graph $G \coloneqq \IS_r \UnionGraph \left(\metaGraph{K_{1 + s}}{K_{p^{t+1}}, \IS_{p^t}, \dots, \IS_{p^t}}\right)$
        is isomorphic to the graph
        $G' \coloneqq  \metaGraph{(K_{1+s} \UnionGraph \IS_r)}{K_{p^{t+1}}, \IS_{p^t}, \dots, \IS_{p^t}}$. Second, observe that
        $F \coloneqq  \metaGraph{D}{\IS_{p^{t+1}}, \IS_{p^t} \dots, \IS_{p^t}}$ is isomorphic to a subgraph of
        $F' \coloneqq  \IS_{p^t} \UnionGraph \left(\metaGraph{K_{1+s}}{\IS_{p^{t+1}}, \IS_{p^t}, \dots, \IS_{p^t}, K_r}\right)$.
        Lastly, observe that $F'$ is a subgraph of $G$ since $p^t > r$. This implies
        $\Phi(F) \geq \Phi(F') \geq \Phi(G) = \Phi(G') = z$ which proves the claim.
    \end{claimproof}

    \cref{eq:isolated:1} is the base case of the proof of \cref{claim:clique:meta} which we prove next.

    \begin{claim}\label{claim:clique:meta}
        If $\Phi(\metaGraph{(K_{1+s} \UnionGraph \IS_r)}{K_{p^{t+1}}, \IS_{p^t} \dots, \IS_{p^t}}) = z$
        then $\Phi(\metaGraph{K_{1 + s + r}}{\IS_{p^{t+1}}, \IS_{p^t}, \dots, \IS_{p^t}}) = z$.
    \end{claim}
    \begin{claimproof}
        Recall that the vertices of a graph $G \in \graphs{1 + s + r}$ are labeled from 1 to $1 + s + r$.
        For a graph $G \in \graphs{1 + s + r}$, we write $\deg_G(2)$ for
        the degree of the second vertex in $G$. We show the statement via induction by using the following induction hypothesis.
        For all $n \in \{0, \dots, s + r\}$ we get $\Phi(\metaGraph{G}{\IS_{p^{t+1}}, \IS_{p^t}, \dots, \IS_{p^t}}) = z$
        for all graphs $G \in \graphs{1 + s + r}$ with $\deg_G(2) \leq n$,

        The base case of the induction is already shown by \cref{eq:isolated:1} since for each graph $G$
        with $\deg_G(2) = 0$
        the graph $\metaGraph{G}{\IS_{p^t}, \dots, \IS_{p^t}}$ is a subgraph of
        $\metaGraph{D}{\IS_{p^{t+1}}, \IS_{p^t} \dots, \IS_{p^t}}$.

        For the induction step, we can assume $\Phi(\metaGraph{G}{\IS_{p^{t+1}}, \IS_{p^t}, \dots, \IS_{p^t}}) = z$ for all graphs
        $G \in \graphs{1+s + r}$ with $\deg_G(2) \leq n$. Our goal is to show the statement for all graphs $G$ with
        $\deg_G(2) = n+1$.
        \begin{citemize}
            \item In the first step, we consider the set $\mathcal{F}$ of all graphs
            in $\graphs{1+s + r}$ with $\deg_G(2) = n+1$ that contain at least
            one edge of the form $\{2, g\}$ for some $g \leq 1+s$. We show that
            $\Phi(\metaGraph{F}{\IS_{p^{t+1}}, \IS_{p^t}, \dots, \IS_{p^t}}) = z$ for all graphs $F \in \mathcal{F}$.
            Assume otherwise and $\Phi(\metaGraph{F}{\IS_{p^{t+1}}, \IS_{p^t}, \dots, \IS_{p^t}}) < z$
            for some graph $F \in \mathcal{F}$, then we can find a graph $F \in \mathcal{F}$ with minimal
            number of edges and $\Phi(\metaGraph{F}{\IS_{p^{t+1}}, \IS_{p^t}, \dots, \IS_{p^t}}) < z$. We show that
            all proper sub-points of $\metaGraph{F}{\IS_{p^{t+1}}, \IS_{p^t}, \dots, \IS_{p^t}}$ evaluate to
            $z$ on $\Phi$. According to \cref{rem:sub-points:IS}, these proper sub-points
            have the form $\metaGraph{F'}{\IS_{p^{t+1}}, \IS_{p^t}, \dots, \IS_{p^t}}$, where $F'$ is a proper edge-subgraph of $F$.

            First, let $F'$ be a proper edge-subgraph with $\deg_{F'}(2) = n$. In this case, our induction hypothesis ensures that
            $\Phi(\metaGraph{F'}{\IS_{p^{t+1}}, \IS_{p^t}, \dots, \IS_{p^t}}) = z$. Otherwise, $\deg_{F'}(2) = n+1$ which implies that
            $F'$ still contains the edge $\{2, g\}$, thus $F' \in \mathcal{F}$. Because $F'$ has strictly fewer edges than $F$,
            we obtain $\Phi(\metaGraph{F'}{\IS_{p^{t+1}}, \IS_{p^t}, \dots, \IS_{p^t}}) = z$.

            This implies that $\Phi(\metaGraph{F'}{\IS_{p^{t+1}}, \IS_{p^t}, \dots, \IS_{p^t}}) = z$ for all proper sub-points, thus
            \cref{lem:ae:nonvanishing:sub-points} shows that $\Phi(\metaGraph{F}{\IS_{p^{t+1}}, \IS_{p^t}, \dots, \IS_{p^t}})$
            is nonvanishing. Lastly, Observe that $\metaGraph{F}{\IS_{p^{t+1}}, \IS_{p^t}, \dots, \IS_{p^t}}$
            contains $K_{p^t, p^t}$ as a subgraph since $F$ contains the edge $\{2, g\}$ for $g \leq 1+s$. This means that $\Phi$
            is concentrated on $k$ with respect to $p^t$. However, this cannot happen since we assume that $\Phi$ is not concentrated
            on $k$ with respect to $p^t$, hence $\Phi(\metaGraph{F}{\IS_{p^{t+1}}, \IS_{p^t}, \dots, \IS_{p^t}}) = z$
            for all graphs $F \in \mathcal{F}$. Further, we get $\Phi(\metaGraph{F}{K_{p^{t+1}}, \IS_{p^t}, \dots, \IS_{p^t}}) = z$
            for all $F \in \mathcal{F}$. Otherwise, $\Phi$ is reducible on $k$ with respect to $p^{t+1}$
            due to \cref{claim:complete:meta:graph}.

            \item Due to the first case, we assume that $\Phi(\metaGraph{F}{K_{p^{t+1}}, \IS_{p^t} \dots, \IS_{p^t}}) = z$ for all $F \in \mathcal{F}$. In the following, we consider a graph $G$ with $\deg_{G}(2) = n+1$ that is not in $\mathcal{F}$ and show that
            $\metaGraph{G}{\IS_{p^t}, \dots, \IS_{p^t}}$ is isomorphic to a subgraph of
            $\metaGraph{F}{K_{p^{t+1}}, \IS_{p^t} \dots, \IS_{p^t}}$  for some $F \in \mathcal{F}$.
            This implies $\Phi(\metaGraph{G}{\IS_{p^t}, \dots, \IS_{p^t}}) = z$ as $\Phi$ is edge-monotone.

            First, observe that we can find an edge $\{2, g\}$ with $s + 2 \leq g$ because vertex $2$ has
            at least one adjacent edge in $G$. Next, we construct a graph
            $F \in \graphs{1 + s + r}$ with
            \[E(F) := \{\{1, 2\}\} \; \cup \; \{\{2, u\} : \{2, u\} \in E(G)\} \setminus \{\{2, g\}\} \; \cup \; \{\{u, v\} : u \neq 2 \neq v, \, u \neq v\}. \]
            We obtain $F$ by adding the edge $\{1, 2\}$, removing the edge $\{2, g\}$ and adding all possible
            edges $\{u, v\}$ with $u \neq 1 \neq v$. Observe that $\deg_{F}(2) = n+1$ and
            that $F \in \mathcal{F}$. We show that there is a bijective function
            $\pi \colon V(\metaGraph{G}{\IS_{p^t}, \dots, \IS_{p^t}}) \to V(\metaGraph{G}{\IS_{p^t}, \dots, \IS_{p^t}})$ such that
            $\pi(\metaGraph{G}{\IS_{p^t}, \dots, \IS_{p^t}})$ is a subgraph of $\metaGraph{F}{K_{p^{t+1}}, \IS_{p^t}, \dots, \IS_{p^t}}$.
            Note that the vertex set of $\metaGraph{G}{\IS_{p^t}, \dots, \IS_{p^t}}$ is equal to
            $\{1\} \times \{1, \dots, p^{t+1}\} \cup \{2, \dots, 1+s\} \times \{1, \dots, p^t\} \cup \{2+s, \dots, 1+s+r\} \times \{1\}$. Next, we define the permutation
            \begin{align*}
                \pi((x, y)) \coloneqq \begin{cases}
                    (1, x - s-1) &\text{if $x \geq s + 2$}, \\
                    (y + s+1, 1) &\text{if $x = 1$ and $y \leq r$}, \\
                    (x, y) &\text{otherwise}.
                \end{cases}
            \end{align*}
            Observe that this mapping is possible since $r < p^{t+1}$.
            Consult \cref{fig:pi} for an illustration of a demonstration of $\pi$.
            We show that $\pi(\metaGraph{G}{\IS_{p^t}, \dots, \IS_{p^t}})$ is a subgraph of
            $\metaGraph{F}{K_{p^{t+1}}, \IS_{p^t}, \dots, \IS_{p^t}}$ by considering the following
            case distinction. Let $\{(a, b), (x, y)\}$ be an edge in $\metaGraph{G}{\IS_{p^t}, \dots, \IS_{p^t}}$.
            Thus, $\{a, x\} \in E(G)$ which implies $a \neq x$.
            \begin{itemize}
                \item Case $a, x \geq 2 + s$: Observe that for $\pi((a, b))$ and $\pi((x, y))$ the first coordinate
                is equal to $1$ meaning that there are adjacent in  $\metaGraph{F}{K_{p^{t+1}}, \IS_{p^t}, \dots, \IS_{p^t}}$.

                \item Case $a, x \leq 1 + s$: This can only happen if $a \neq 2$ and $x \neq 2$ since
                otherwise $G \in \mathcal{F}$. We consider a case distinction
                \begin{itemize}
                    \item Case $a \neq 1$ or $x \neq 1$: With lose of generality, we assume that $x \neq 1$.
                    This implies $\pi((x, y)) = (x, y)$. Let $(c, d) \coloneqq \pi((a, b))$. Observe that
                    $c \neq x$ and $c \neq 2$ since $a \neq x$ and $a \neq 2$. Further, since $x \neq 2 \neq c$, we get that
                    $\{x, c\} \in F$ which implies that $\pi((a, b))$ and $\pi((x, y))$ are
                    adjacent in  $\metaGraph{F}{K_{p^{t+1}}, \IS_{p^t}, \dots, \IS_{p^t}}$.

                    \item Case $a = 1$ and $x = 1$: We get $\pi((a, b)) = (b+s+1, 1)$ and $\pi((x, y)) = (y+s+1, 1)$.
                    Since $\{b+s+1, y+s+1\} \in F$, we get that $\pi((a, b))$ and $\pi((x, y))$ are
                    adjacent in $\metaGraph{F}{K_{p^{t+1}}, \IS_{p^t}, \dots, \IS_{p^t}}$.
                \end{itemize}

                \item Case $a \leq 1 + s$ and $x \geq 2 + s$: We consider another case distinction:
                \begin{itemize}
                    \item Case $a \neq 1$: We get $\pi((a, b)) = (a, b)$ and $\pi((x, y)) = (1, x-s-1)$.
                    Observe that $\{a, 1\} \in F$ since $\{1, g\} \in F$, for all $2 \leq g \leq 1+s+r$.
                    This implies that $\pi((a, b))$ and $\pi((x, y))$ are
                    adjacent in  $\metaGraph{F}{K_{p^{t+1}}, \IS_{p^t}, \dots, \IS_{p^t}}$

                    \item Case $a = 1$: We get $\pi((a, b)) = (b+s+1, 1)$ and $\pi((x, y)) = (1, x-s-1)$.
                    Again, observe that $\{b+s+1, 1\} \in F$ since $\{1, g\} \in F$, for all $2 \leq g \leq 1+s+r$.
                    This implies that $\pi((a, b))$ and $\pi((x, y))$ are
                    adjacent in  $\metaGraph{F}{K_{p^{t+1}}, \IS_{p^t}, \dots, \IS_{p^t}}$
                \end{itemize}
                \item Case $a \geq 2 + s$ and $x \leq 2 + s$: Analog to the last case.
            \end{itemize}

            \begin{figure}[t]
                \centering
                \includegraphics[scale=1.15]{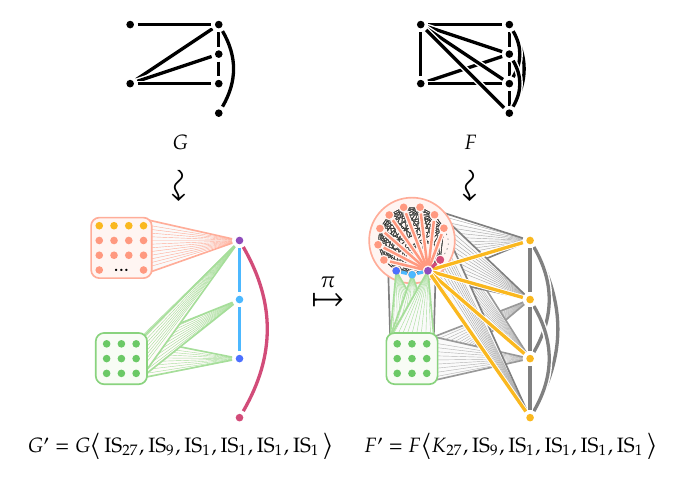}

                \caption{Graphs \(G\) and \(F\) that are transformed into the inhabited
                    graphs
                    $G' \coloneqq \metaGraph{G}{\IS_{27}, \IS_9, \IS_1, \IS_1, \IS_1, \IS_1}$ that is not in
                $\mathcal{F}$ and a graph $F' \coloneqq \metaGraph{F}{K_{27}, \IS_9, \IS_1, \IS_1, \IS_1, \IS_1}$ that is in
                $\mathcal{F}$.
                The graph \(G'\) is isomorphic to a subgraph of \(F'\); the isomorphism
                \(\pi\) maps vertices of \(G'\) to vertices of \(F'\) that have a common
                color; the colors of the edges demonstrate that \(\pi\) is indeed an
                isomorphism onto a subgraph of \(F'\). Gray edges of \(F'\) symbolize
                extra edges in \(F'\) that do not have corresponding edges in \(G'\).
                Note that the two orange graphs actually have 27 nodes.
            }\label{fig:pi}
            \end{figure}

            Hence, $\metaGraph{G}{\IS_{p^t}, \dots, \IS_{p^t}}$ is isomorphic to a subgraph of
            $\metaGraph{F}{K_{p^{t+1}}, \IS_{p^t} \dots, \IS_{p^t}}$. Since $\Phi$ is edge-monotone, we get that
            $\Phi(\metaGraph{G}{\IS_{p^t}, \dots, \IS_{p^t}}) = z$ which ends our induction step.
            \claimqedhere
        \end{citemize}
    \end{claimproof}

    Thus, we get $\Phi(\metaGraph{(K_{1 + s + r})}{\IS_{p^{t+1}}, \IS_{p^t}, \dots, \IS_{p^t}}) = z$.
    We show that this is enough to prove $\Phi(K_k) = z$.

    \begin{claim}
        If $\Phi(\metaGraph{(K_{1 + s + r})}{\IS_{p^{t+1}}, \IS_{p^t}, \dots, \IS_{p^t}}) = z$
        then $\Phi(\metaGraph{(K_{1 + s + r})}{K_{p^{t+1}}, K_{p^t}, \dots, K_{p^t}}) = K_{k} = z$.
    \end{claim}
    \begin{claimproof}
        First, we apply \cref{claim:complete:meta:graph} to our assumption
        \[\Phi(\metaGraph{(K_{1 + s + r})}{\IS_{p^{t+1}}, \IS_{p^t}, \dots, \IS_{p^t}}) =
        z\]
        and obtain that
        \[\Phi(\metaGraph{(K_{1 + s + r})}{K_{p^{t+1}}, \IS_{p^t}, \dots, \IS_{p^t}}) =
        z.\]

        Next, assume that the statement is false then there are fixed points $A^2, \dots, A^{s+1} \in \fp(\sylow_{p^t}, K_{p^t})$
        such that $\Phi(F) < z$, where $F \coloneqq \metaGraph{(K_{1 + s + r})}{K_{p^{t+1}}, A^2, \dots, A^{s+1}}$.
        Here, $A^2, \dots, A^{s+1}$ are chosen minimally.  We get that
        there is at least one fixed point $A^j$ with $A^j \neq \IS_{p^t}$. After reordering, we may assume
        $j = 2$.

        Next, we define the graph parameter
        \[\Tilde{\Phi}(G) = \Phi(\metaGraph{K_{1 + s + r}}{G, K_{p^t}, A^3, \dots,
        A^{s+1}}).\]
        It is now easy to see that $\Tilde{\Phi}(\metaGraph{K_p}{A^2, K_{p^t}, \dots, K_{p^t}})$ is isomorphic to $F$.
        However, this implies that
        \[\Tilde{\Phi}(\metaGraph{K_p}{A^2, K_{p^t}, \dots, K_{p^t}}) < z.\]
        Further, we have
        $\Tilde{\Phi}(\metaGraph{K_p}{\IS_{p^t}, K_{p^t}, \dots, K_{p^t}}) = z$ since $A^2, \dots, A^{s+1}$ were
        chosen minimally and $A^2 \neq \IS_{p^t}$. This shows that $\Tilde{\Phi}$ is nontrivial on $p^{t+1}$ which
        means that $\Phi$ is reducible on $k$ with respect to $p^{t+1}$. However, we assume that this is not the case,
        thus $\Phi(\metaGraph{(K_{1 + s + r})}{K_{p^{t+1}}, K_{p^t}, \dots, K_{p^t}}) = z$.
    \end{claimproof}

    This shows that $\Phi(K_k) = z$ which implies that $\Phi$ is trivial on $k$.
    Thus we get that if $\Phi$ is nontrivial on $k$ then $\Phi$ is either concentrated on $k$ with
    respect to $p^t$ or reducible on $k$ with respect to $p^{t+1}$.
\end{proof}

\subsection{Hardness and Tight Bounds for the General Case}\label{sec:main:theorem}

In this section, we prove \cref{theo:edge_mono}. The main idea is that
if $\Phi$ is nontrivial on $k$, \cref{lem:param:large:biclique}
ensures that either $\Phi$ is concentrated or reducible on $k$ with
respect to some prime power $p^t$ or $p^{t+1}$. If $\Phi$ is concentrated,
we use \cref{cor:tight:bounds} and we are done. Otherwise,
we use the reduction shown in \cref{lem:inc:exc} which allows us to use \cref{theo:edge_mono:prime:power:bicliques}.
The reduction constructs a graph parameter $\Scat{\Phi}{q}$
that is nontrivial on large prime powers.

\begin{definition}\label{11.5-2}
    Let \(\Phi \colon \graphs{} \to \Q\), denote an edge-monotone graph parameter and
    $q\colon \nat \to \nat$ a computable function
    and write \(\ScatSet{\Phi}{q} \) for the set of all integers $k$ such that \(\Phi\) is reducible on $k$ with respect to $q(k)$.
\end{definition}

The set $q(\ScatSet{\Phi}{q}) \coloneqq \{ q(k) : k \in \ScatSet{\Phi}{q} \}$ is the set of all values on
which $\Scat{\Phi}{q}$ is nontrivial. However, we need that the set $q(\ScatSet{\Phi}{q})$ is computable.

\begin{lemma}\label{lem:scatset:comp}
    Let $\Phi \colon \graphs{} \to \Q$ be a graph parameter and $q\colon \nat \to \nat$
    be a computable unbounded monotonically increasing function. Then
    the sets \(\ScatSet{\Phi}{q}\) and $q(\ScatSet{\Phi}{q})$ are computable. Further,
    there is an algorithm that for $x \in q(\ScatSet{\Phi}{q})$ computes $C \in \graphs{s}$ and
    $H_2, \dots, \dots, H_{s} \in \graphs{}$ such that the graph
    parameter $\reductionGraph{\Phi}{C}{H_2, \dots, H_{s}}(G) \coloneq \Phi(\metaGraph{C}{G, H_2 \dots, H_{s}})$ is
    nontrivial on $x$ with $k \coloneqq (x + \sum_{i = 2}^s |V(H_i)|) \in \ScatSet{\Phi}{q}$ and $k$
    is minimal under all values $j \in \ScatSet{\Phi}{q}$ with $q(j) = x$.
\end{lemma}
\begin{proof}
    First, we construct an algorithm $\mathbb{A}$ to check if $k \in \ScatSet{\Phi}{q}$. For this, we compute the value $q(k)$ and check if it is a prime power.
    We now iterate over all $s \in [k]$, all graphs $C \in \graphs{s}$ and graphs
    $H_2, \dots, \dots, H_{s} \in \graphs{}$ with $k = q(k) + \sum_{i = 2}^s |V(H_i)|$. Observe
    that only finitely many graphs with this property exist. For each set of graphs, we check whether the
    graph parameter $\reductionGraph{\Phi}{C}{H_2, \dots, H_{s}}(G) \coloneq \Phi(\metaGraph{C}{G, H_2 \dots, H_{s}})$ is
    nontrivial on $q(k)$ and stop if this the case. In this case, we return the
    graphs and say that $k \in \ScatSet{\Phi}{q}$. Otherwise, we return that $k \notin \ScatSet{\Phi}{q}$.

    To check $x \in q(\ScatSet{\Phi}{q})$, we check for $k = 2$, $k = 3$, $\dots$ if $k \in \ScatSet{\Phi}{q}$ and
    $q(k) = x$. We stop if we find an $k$ with $k \in \ScatSet{\Phi}{q}$ and
    $q(k) = x$ and return that $x$ is in $q(\ScatSet{\Phi}{q})$. Further, we use algorithm $\mathbb{A}$ to compute
    the graph $C \in \graphs{s}$ and graphs $H_2, \dots, \dots, H_{s} \in \graphs{}$. Alternatively, we stop if we find the first value
    $k$ with $q(k) > x$ and return that $x$ is not in $q(\ScatSet{\Phi}{q})$.  Observe that one of these two cases will always occur after
    a finite amount of steps since $q$ is an unbounded monotonically increasing function.
\end{proof}

Next, we show how the graph parameter $\Scat{\Phi}{p}$ is constructed and that there is a reduction
from $\NUM{}\cpindsub(\Scat{\Phi}{p})$ to $\NUM{}\cpindsub(\Phi)$.

\begin{lemma}\label{lem:scat:graph:parameter}
   Let $\Phi \colon \graphs{} \to \Q$ be an edge monotone graph parameter and let
   $q\colon \nat \to \nat$ be a computable unbounded monotonically increasing function. Then there
   is a computable graph parameter $\Scat{\Phi}{q}$ with:
    \begin{itemize}
        \item $\Scat{\Phi}{q}$ is nontrivial exactly on $q(\ScatSet{\Phi}{q})$.
        \item $\Scat{\Phi}{q}$ is edge-monotone.
        \item There is parameterized Turing reduction from $\NUM{}\cpindsub(\Scat{\Phi}{q})$ to $\NUM{}\cpindsub(\Phi)$.
    \end{itemize}
\end{lemma}
\begin{proof}
    Let $G$ be a graph. We set $k \coloneqq |V(G)|$. According to \cref{lem:scatset:comp} we can check
    if $k \in q(\ScatSet{\Phi}{q})$ and compute a graph $C^k \in \graphs{s_k}$ and graphs
    $H^k_2, \dots, \dots, H^k_{s_k} \in \graphs{}$ such that the graph
    parameter $\reductionGraph{\Phi}{C^k}{H^k_2, \dots, H^k_{m_s}}(G) \coloneq \Phi(\metaGraph{C^k}{G, H^k_2 \dots, H^k_{s_k}})$ is
    nontrivial on $k$. We use this to define the following
    graph parameter:
    \[\Scat{\Phi}{q}(G) =  \begin{cases}
        \Phi(\metaGraph{C^k}{G, H_2^k, \dots, H_{s_k}^k}) &\text{ if $k \coloneqq |V(G)| \in q(\ScatSet{\Phi}{q})$} \\
        0 &\text{ otherwise}.
    \end{cases}\]
    First, observe that $\Scat{\Phi}{q}$ is a computable graph parameter. Further, it is easy to check that $\Scat{\Phi}{q}$ is edge-monotone, whenever $\Phi$ is. Lastly, $\Scat{\Phi}{q}$ is nontrivial exactly on $q(\ScatSet{\Phi}{q})$ by construction.

    Next, we prove the reduction. For a $k$-vertex graph $H$ and a $H$-colored graph $G$, we want to
    compute $\NUM{}\cpindsubs{(\Scat{\Phi}{q}, H)}{G}$. First, check if $k \notin q(\ScatSet{\Phi}{q})$. If this is the
    case then $\NUM{}\cpindsubs{(\Scat{\Phi}{q}, H)}{G} = 0$ since $\Scat{\Phi}{q}$ is trivial on $k$. Since $q(\ScatSet{\Phi}{q})$
    is computable, this can be done in time $h_1(k)$ for some computable function $h$. We again use \cref{lem:scatset:comp} to compute
    the graph $C^k \in \graphs{s_k}$ and graphs $H^k_2, \dots, \dots, H^k_{s_k} \in \graphs{}$
    in time $h_1(k)$. Observe that $c_k = \sum_{i = 2}^{s_k} \# V(H^k_i)$
    only depends on $k$. Now, \cref{lem:inc:exc} allows us to compute $\NUM{}\cpindsubs{(\Scat{\Phi}{q}, H)}{G}$ using an oracle
    for $\NUM{}\cpindsubs{(\Phi, \graphs{k + c_k})}{\star}$. This shows that there is a parameterized Turing
    reduction from $\NUM{}\cpindsub(\Scat{\Phi}{q})$ to $\NUM{}\cpindsub(\Phi)$.
\end{proof}

We define the function $\primebase \colon \nat \to \nat$ that returns for each prime power $p^t$ to value $p^{t-1}$. Further
$\primebase(k) = 1$ if $k$ is not a prime power. Now,
because of \cref{theo:edge_mono:prime:power:bicliques}, we can find for all $q(k) \in \Scat{\Phi}{q}$
a nonvanishing graph $F$ that contain $K_{\primebase(q(k)), \primebase(q(k))}$ as a subgraph.
Thus we can use \cref{theo:lower:bound:bicliques} to get lower bounds for $\NUM{}\cpindsub(\Scat{\Phi}{q})$.
We show that these lower bounds propagate to $\NUM{}\indsubsprob(\Phi)$ and $\NUM{}\cpindsub(\Phi)$
due to \cref{lem:scat:graph:parameter}.

\begin{theorem}\label{theo:scat:hard}
   Let $\Phi \colon \graphs{k} \to \codo{\Phi}{k}$ be an edge monotone graph parameter and let
   $q\colon \nat \to \nat$ be a computable unbounded monotonically increasing function.
   Let $q' \colon \ScatSet{\Phi}{q} \to \nat$ be the restriction of $q$ to $\ScatSet{\Phi}{q}$.
   \begin{itemize}
    \item If $\primebase(q'(k)) \in \omega(1)$ then $\NUM{}\indsubsprob(\Phi)$ and $\NUM{}\cpindsub(\Phi)$ are
    both \w-hard.
    \item There is a constant $\gamma > 0$ such that for all
    $k \in \ScatSet{\Phi}{q}$, there is no algorithm (that reads the whole input) that for every $G$ computes
    $\NUM{}\indsubs{(\Phi, k)}{G}$ (or $\NUM{}\cpindsubs{(\Phi, \graphs{k})}{G}$)
    in time $O(|V(G)|^{\gamma \primebase(q(k))})$ unless ETH fails.
    Here $\gamma$ depends on $q$ but is independent of $\Phi$ and $k$.
   \end{itemize}
\end{theorem}
\begin{proof}
    We use \cref{lem:scat:graph:parameter} to construct the graph parameter $\Scat{\Phi}{q}$.
    Let $q(k) \in q(\ScatSet{\Phi}{q})$, then according to \cref{theo:edge_mono:prime:power:bicliques}
    there exists a nonvanishing graph $H$ that contains $K_{\primebase(q(k)), \primebase(q(k))}$ as a subgraph. Now,
    if $\primebase(q'(k)) \in \omega(1)$ then we can apply \cref{cor:tight:bounds} to
    $h \colon \ScatSet{\Phi}{q} \to \nat, k \mapsto \primebase(q'(k))$ which proves that
    $\NUM{}\cpindsub(\Scat{\Phi}{q})$ is \w-hard. \cref{lem:scat:graph:parameter}
    ensures that $\NUM{}\cpindsub(\Phi)$ is \w-hard, too. Lastly, \cref{lem:redu:cpindsub:indsub}
    shows that $\NUM{}\indsubsprob(\Phi)$ is \w-hard.

    Further, assuming ETH, we can use \cref{theo:lower:bound:bicliques} to get an integer $N$
    and a constant $\beta$ such that for all $q(k) \geq N$ with $q(k) \in q(\Scat{\Phi}{q})$,
    there is no algorithm (that reads the whole input)  that for every \(G\) computes
    $\NUM{}\cpindsubs{(\Scat{\Phi}{q}, \graphs{q(k)})}{G}$ in time $O(|V(G)|^{\beta \primebase(q(k))})$. Since $q$ is monotonically increasing
    and unbounded, we can find a value $N'$ such that $q(k) \geq N$ if and only if $k \geq N'$.
    Next, let $k \in \ScatSet{\Phi}{q}$ such that $k \geq N'$. Assume that there is an algorithm
    that solves $\NUM{}\cpindsubs{(\Phi, \graphs{k})}{G}$ in time $O(|V(G)|^{\beta \primebase(q(k))})$.
    Now using \cref{lem:scatset:comp}, we can compute a graph $C^k \in \graphs{s_k}$ and graphs
    $H^k_2, \dots, \dots, H^k_{s_k} \in \graphs{}$ such that the graph
    parameter $\reductionGraph{\Phi}{C^k}{H^k_2, \dots, H^k_{m_s}}(G) \coloneq \Phi(\metaGraph{C^k}{G, H^k_2 \dots, H^k_{s_k}})$
    coincides with $\Scat{\Phi}{q}$ on $q(k)$-vertex graphs.
    Now, let $c_k \coloneqq k - q(k)$. According to \cref{lem:inc:exc}, we can use our
    algorithm for $\NUM{}\cpindsubs{(\Phi, \graphs{k})}{G}$
    to solve $\NUM{}\cpindsubs{(\Scat{\Phi}{q}, \graphs{q(k)})}{G}$ in time
    \[O((|V(G)| + c_k)^{\beta \primebase(q(q(k) + c_k))}) = O(|V(G)|^{\beta \primebase(q(k))}),\]
    which is not possible unless ETH fails.
    Next, set $\beta' \coloneqq \min(\delta'_{s, p}, 1/(N'+1))$.
    Now for all \(k < N'\), we obtain
    \[
        O(|V(G)|^{\beta' \primebase(q(k))}) = o( |V(G)| ).
    \]
    since $\primebase(q(k)) \leq k$. Such a running time is unconditionally unachievable
    for any algorithm that reads the whole input. Thus for all
    $k \in \ScatSet{\Phi}{q}$, there is no algorithm (that reads the whole input) that for every $G$ computes
    $\NUM{}\cpindsubs{(\Phi, \graphs{k})}{G}$ in time $O(|V(G)|^{\beta' \primebase(q(k))})$ unless ETH fails.

    Lastly, we define $C \coloneqq 4/ \beta'$ and $\gamma \coloneqq \min(\beta'/2, 1/(C+1))$. Assume that we can compute
    on input $G$ the value
    $\NUM{}\indsubs{(\Phi, k)}{G}$ in time $O(|V(G)|^{\gamma \primebase(q(k))})$. If $\primebase(q(k)) \geq C$, then we can use
    \cref{lem:redu:cpindsub:indsub} to get an algorithm that computes $\NUM{}\cpindsubs{(\Phi, \graphs{k})}{G}$
    in time $O(|V(G)|^{\gamma \primebase(q(k)) + 2})$. Now, we obtain that
    $\gamma \primebase(q(k)) + 2 \leq \beta'/2 \cdot \primebase(q(k)) + 2 \leq \beta' \primebase(q(k))$ for
    $\primebase(q(k)) \geq C = 4 / \beta'$. This implies that we can compute $\NUM{}\cpindsubs{(\Phi, \graphs{k})}{G}$
    in time $O(|V(G)|^{\beta \primebase(q(k))})$ which is not possible unless ETH fails.
    Otherwise $\primebase(q(k)) < C$ and we would obtain a sublinear algorithm for
    $\NUM{}\indsubs{(\Phi, k)}{G}$ which is unconditionally unachievable
    for any algorithm that reads the whole input.
\end{proof}

\thmedgemon

\begin{proof}
    Without loss of generality, we assume that the codomain of $\Phi$ is equal to $\{0, \dots, c\}$.
    Otherwise, we use \cref{lem:finite:codomain}.
    Let $p$ be the first prime number such that $c < p$ and write $\nt{\Phi}$ for the set
    of values $k$ on which $\Phi$ is nontrivial. Let $\primepower{p}$ be the function
    that on input $k$ returns the largest prime power $p^t$ such that $p^t \leq k$.  We define the function
    $q(k) \coloneqq \primepower{p}(k)/p $ and $h(k) \coloneqq q(k)/p = \primebase(q(k))$. Observe that $q(k) \geq k/p^2$ thus
    $q(k), h(k) \in \Omega(k)$. Further, $q$ is a computable unbounded monotonically increasing function with
    $q(k)/p + q(k) \leq k$.  This means that we can apply \cref{lem:param:large:biclique}
    on all $k \in \nt{\Phi}$ to get that $\Phi$ is concentrated on $k$ with respect to $q(k)/p = h(k)$
    or $\Phi$ is reducible on $k$ with respect to $q(k)$.

    Write $\CoSet{\Phi}{q}$ for the set of all $k \in \nt{\Phi}$ such that $\Phi$ is
    concentrated on $k$ with respect to $q(k)/p$.
    Further, write $\ScatSet{\Phi}{q}$ for the set of all $k \in \nt{\Phi}$ such that $\Phi$ is
    reducible on $k$ with respect to $q(k)$. Now, \cref{lem:param:large:biclique}
    implies that $\ScatSet{\Phi}{q} \cup \CoSet{\Phi}{q} = \nt{\Phi}$.

    If we assume that $\nt{\Phi}$ is infinite then at least one of the sets $\ScatSet{\Phi}{q}$ or $\CoSet{\Phi}{q}$ has
    to be infinite as well. If $\ScatSet{\Phi}{q}$ is infinite, then we use \cref{theo:scat:hard}
    to show that both $\NUM{}\indsubsprob(\Phi)$ and  $\NUM{}\cpindsub(\Phi)$ are \w-hard.
    If $\CoSet{\Phi}{q}$ is infinite, then we can use \cref{cor:tight:bounds} to show that  both $\NUM{}\indsubsprob(\Phi)$ and  $\NUM{}\cpindsub(\Phi)$ are \w-hard.

    If we assume ETH, then we can use \cref{theo:scat:hard} to get a constant $\delta_{s, c}$ such that
    for all $k \in \ScatSet{\Phi}{q}$
    no algorithm computes $\NUM{}\indsubs{(\Phi, k)}{\star}$ (or $\NUM{}\cpindsubs{(\Phi, k)}{\star}$)
    in time $O(|V(G)|^{\delta_{s, c} h(k) })$.
    On the other side, we can use \cref{cor:tight:bounds} to get
    a constant $\delta_{c, c}$ such that for all $k \in \CoSet{\Phi}{q}$
    no algorithm computes $\NUM{}\indsubs{(\Phi, k)}{\star}$ (or $\NUM{}\cpindsubs{(\Phi, k)}{\star}$)
    in time $O(|V(G)|^{\delta_{c, c} h(k) })$.
    We obtain our tight bounds under ETH by simply picking $\delta_c = \min(\delta_{s, c}, \delta_{c, c}) / p$.
\end{proof}

\subsection{Hardness for unbounded codomain}
In this section, we show that $\NUM{}\indsubsprob(\Phi)$ and $\NUM{}\cpindsub(\Phi)$ are
still \w-hard when our codomain is smaller than $\lfloor (1 - \varepsilon) k^{1/\alpha} \rfloor$.
We also get lower bounds which are not tight anymore.

\thmedgemonDomain
\begin{proof}
    Fix $\beta \coloneqq 1 - {1}/{\alpha}$, and $\delta > 0$ such that $(1 + \delta)^\alpha (1 - \varepsilon)^\alpha < 1$.
    According to \cite[(22.19.2)]{number_theory_hardy}, there exists a number $N_{\delta}$ such that for
    all $k \geq N_{\delta}$ there exists a prime number $p$ with $k < p < (1 + \delta)k$.
    Let $N'_\delta$ be a number such that $k \geq N'_\delta$ implies $\codo{\Phi}{k} \geq N_\delta$.

    For all $k$, we write $g(k)$ for the smallest prime number $p$ with $\codo{\Phi}{k} < p = g(k)$.
    Note that \cite[(22.19.2)]{number_theory_hardy} implies that for all $k \geq N'_\delta$ we have $g(k) < (1+ \delta)\codo{\Phi}{k}$.
    Further, we define $q(k) \coloneqq g(k)^\alpha$ and $h(k) \coloneqq g(k)^{\alpha-1}$.
    Observe that $q(k)$ is a computable unbounded monotonically increasing function.

    Further, for $k \geq N'_\delta$ we get
    \[\frac{q(k)}{g(k)} + q(k) = \left(1 + \frac{1}{g(k)}\right)q(k) \leq
    (1 + \delta)^\alpha (1 - \varepsilon)^\alpha \left(1 + \frac{1}{g(k)}\right)k.\]
    For $g(k)$ large enough, we get that
    \[\underbrace{(1 + \delta)^\alpha (1 - \varepsilon)^\alpha}_{< 1} \cdot \left(1 + \frac{1}{g(k)}\right) < 1 \]
    implying $q(k)/g(k) + q(k) \leq k$.

    Since the function $g(k)$ is unbounded and monotonically increasing,
    there is a constant $N$ such that, for all
    $k \geq N$, we can use \cref{lem:param:large:biclique}. This means that for all $k \in \nt{\Phi}$ with $k \geq N$
    we get that $\Phi$ is concentrated on $k$ with respect to $q(k)/g(k) = h(k)$
    or $\Phi$ is reducible on $k$ with respect to $q(k)$.

    Write $\CoSet{\Phi}{q}$
    for the set of all $k \in \nt{\Phi}$ such that $\Phi$ is concentrated on $k$ with respect to $h(k)$.
    Further, write $\ScatSet{\Phi}{q}$ for the set of all $k \in \nt{\Phi}$ such that $\Phi$ is
    reducible on $k$ with respect to $q(k)$.

    If we assume that $\nt{\Phi}$ is infinite then at least one of the sets $\ScatSet{\Phi}{q}$ or $\CoSet{\Phi}{q}$ has
    to be infinite as well. If $\ScatSet{\Phi}{q}$ is infinite, then we use \cref{theo:scat:hard}
    to show that both $\NUM{}\indsubsprob(\Phi)$ and  $\NUM{}\cpindsub(\Phi)$ are \w-hard.
    If $\CoSet{\Phi}{q}$ is infinite, then we can use \cref{cor:tight:bounds} to show that  both $\NUM{}\indsubsprob(\Phi)$ and  $\NUM{}\cpindsub(\Phi)$ are \w-hard.

    If we assume ETH, then we can use \cref{theo:scat:hard} to get a constant $\delta_{s, c}$ such that
    for all $k \in \ScatSet{\Phi}{q}$ with $k \geq N$
    no algorithm computes $\NUM{}\indsubs{(\Phi, k)}{\star}$ (or $\NUM{}\cpindsubs{(\Phi, k)}{\star}$)
    in time $O(|V(G)|^{\delta_{s, c} q(k) / g(k) }) = O(|V(G)|^{\delta_{s, c} h(k) })$.

    Observe that $h \in \Omega(k^{\beta})$ since $h(k)^{1/\beta} = g(k)^\alpha \geq (1 - \varepsilon)^\alpha k$.
    This allows us to use \cref{cor:tight:bounds} to get a
    constant $\delta_{c, c}$ such that for all $k \in \CoSet{\Phi}{q}$ with $k > N$
    no algorithm computes $\NUM{}\indsubs{(\Phi, k)}{\star}$ (or $\NUM{}\cpindsubs{(\Phi, k)}{\star}$)
    in time $O(|V(G)|^{\delta_{c, c} h(k) })$.
    Next, we set $\delta = (1-\varepsilon)^{\alpha - 1}\min(\delta_{s, c}, \delta_{c, c}, 1/(N+1))$.
    Observe that for $k \geq N$ we obtain
    \[\min(\delta_{s, c}, \delta_{c, c}) \cdot h(k) \geq \min(\delta_{s, c}, \delta_{c, c}) \cdot \codo{\Phi}{k}^{\alpha-1} \geq
    \min(\delta_{s, c}, \delta_{c, c}) \cdot (1 - \varepsilon)^{\alpha-1} k^{\frac{\alpha - 1}{\alpha}} \geq \delta k^{\beta},\]
    where we used $g(k) >  \codo{\Phi}{k}$.
    Thus, for all $k \in \nt{\Phi}$ with $k \geq N$, no algorithm computes for every graph
    \(G\) the number $\NUM{}\indsubs{(\Phi, k)}{\star}$
    (or $\NUM{}\cpindsubs{(\Phi, k)}{\star}$)
    in time $O(|V(G)|^{\delta k^{\beta} })$.
    For $k < N$, we obtain
    \[
        O(|V(G)|^{\delta k^{\beta}}) = o( |V(G)| ).
    \]
    since $k^{\beta} \leq k$. Such a running time is unconditionally unachievable
    for any algorithm that reads the whole input. Thus for all nontrivial $k \geq 3$,
    no algorithm (that reads the whole input) computes the function $\NUM{}\indsubs{(\Phi, k)}{\star}$
    (or $\NUM{}\cpindsubs{(\Phi, \graphs{k})}{\star}$) in time $O(|V(G)|^{\delta k^\beta})$.
\end{proof}

\section{The Sub-Basis Representation of $\NUM{}\indsubsprob(\Phi)$}\label{sec:sub:base}

Many counting problems $\NUM{}P$ in graphs can be expressed as a linear combination of graph homomorphisms
counts. This means that we can find a finite number of non-isomorphic graphs $H$ and coefficients $\alpha_H$
such that for all graphs \(G\), we have
\begin{align}\label{eq:hom:basis}
    \NUM{}P(G) = \sum_{H}  \alpha_H \cdot \NUM{} \homs{H}{G}.
\end{align}
We say that (\ref{eq:hom:basis}) is the \emph{hom-basis representation} of $\NUM{}P$.

In \cite{hom:basis}, the authors used that
the set $\{\NUM{}\homs{H}{\star} : H \in \graphs{}\}$ forms a basis of an infinite
dimension vector space which allowed them to rewrite sums of homomorphisms using
different basis. One basis they used was the sub-basis $\{\NUM{}\subs{H}{\star} : H \in \graphs{}\}$
where $\NUM{}\subs{H}{G}$ is the number of subgraphs in $G$ that are isomorphic to $H$. This
means that we can do a change of basis and rewrite (\ref{eq:hom:basis}) using the sub-basis
representation.

In this section we analyze the sub-basis representation of $\Phi$ and $\NUM{}\indsubsprob(\Phi)$ and
see that the coefficients the sub-basis are in direct relationship to the alternating enumerator.
This enables us to find nontrivial graph parameters $\Phi$ such that the $\NUM{}\indsubsprob(\Phi)$
becomes easy to compute. Lastly, in \cref{sec:indsub:harder}, we compare the sub-basis representation of
$\NUM{}\indsubsprob(\Phi)$ with the $\cphom$-basis representation of $\NUM{}\cpindsubsprob(\Phi)$
which helps us to understand why there are graph parameters $\Phi$ such that
$\NUM{}\indsubsprob(\Phi)$ is \w-hard while $\NUM{}\cpindsubsprob(\Phi)$ is FPT solvable.

\subsection{Using the Sub-Basis for Graph Parameters}

\begin{theorem}\label{theo:sub:base}
    Let $\Phi_k \colon \graphs{k} \to \Q$ be any graph parameter.
    Then, there are coefficients $\subCoeff{\Phi}{k}{H}$ such that
    \[\Phi_k(G) =  \sum_{H \in \graphs{k}^* }  \subCoeff{\Phi}{k}{H} \cdot \NUM{}\subs{H}{G}.\]
    Further, we have $\subCoeff{\Phi}{k}{H} = (-1)^{\#E(H)} \cdot \ae{\Phi}{H}$.
\end{theorem}
\begin{proof}
    We first prove that we can write $\Phi_k$ as a linear combination of subgraph counts. For this,
    we define for all $H \in \graphs{k}^*$ the values
    \begin{align}\label{eq:coeff:sub:base}
        \subCoeff{\Phi}{k}{H} \coloneqq \Phi_k(H) - \sum_{\substack{F \neq H \\ F \in \graphs{k}^*}} \subCoeff{\Phi}{k}{F} \cdot \NUM{}\subs{F}{H}.
    \end{align}
    Observe that $\NUM{}\subs{F}{H}$ is only nonvanishing if $F$ is isomorphic to a subgraph of $H$.
    Thus, \cref{eq:coeff:sub:base} is a well-defined recursive definition. Further, we obtain
    \[ \Phi_k(H) = \subCoeff{\Phi}{k}{H} + \sum_{\substack{F \neq H \\ F \in \graphs{k}^*}} \subCoeff{\Phi}{k}{F} \cdot \NUM{}\subs{F}{H} =
    \sum_{\substack{F \in \graphs{k}^*}} \subCoeff{\Phi}{k}{F} \cdot \NUM{}\subs{F}{H}, \]
    where the last equation is justified since $\NUM\subs{H}{H} = 1$. Lastly, we show that for each $k$-vertex graph $H$, we obtain
    $\subCoeff{\Phi}{k}{H} = (-1)^{\#E(H)} \cdot \ae{\Phi}{H}$.
    For this, observe that $H \in \graphs{k}^\ast$, we obtain
    \begin{align*}
        \ae{\Phi}{H} = \sum_{S \subseteq E(H)} \Phi(\ess{H}{S}) (-1)^{\#S} &= \sum_{F \in \graphs{k}^*} (-1)^{\#E(F)} \cdot \Phi(F) \cdot \NUM{}\subs{F}{H} \\
        &= \sum_{F \in \graphs{k}^*} \sum_{F' \in \graphs{k}^*}  (-1)^{\#E(F)}  \cdot \subCoeff{\Phi}{k}{F'} \cdot  \NUM{}\subs{F'}{F} \cdot \NUM{}\subs{F}{H} \\
        &= \sum_{F' \in \graphs{k}^*} \subCoeff{\Phi}{k}{F'}   \cdot (-1)^{\#E(F')}   \sum_{F \in \graphs{k}^*}  (-1)^{\#E(F)- \#E(F')} \cdot  \NUM{}\subs{F'}{F} \cdot \NUM{}\subs{F}{H} ,
    \end{align*}
    where the last step is archived by swapping the summation. Next, we use that for all graphs $F', H \in \graphs{k}^\ast$
    \begin{align}\label{eq:ae:sub:2}
        \sum_{F \in \graphs{k}^* } (-1)^{\#E(F) - \#E(F')}   \cdot  \NUM{}\subs{F'}{F} \cdot \NUM{}\subs{F}{H}  =
        \begin{cases}
        1 &\text{if $F' = H$,} \\
        0 &\text{otherwise.}
    \end{cases}
    \end{align}
    This equation
    is justified since the sum of \cref{eq:ae:sub:2} is the same as the matrix product of
    $\mathrm{Ext}^{-1}$ and $\mathrm{Ext}$ evaluated at position $(F', H)$
    (see \cite[Section 3.1, (9)]{hom:basis} for the definition of $\mathrm{Ext}$).\footnote{$\mathrm{Ext}$ is an infinite matrix whose rows and columns are enumerated
    by $\graphs{}^\ast$. For a graph $H$ and a graph $G$ we define
    $\mathrm{Ext}(H, G) = [V(H) = V(G)] \cdot \NUM{}\subs{H}{G}$, where $[V(H) = V(G)]$ is
    one if $V(H) = V(G)$  and zero otherwise. }
    Since $\mathrm{Ext}^{-1} \cdot \mathrm{Ext}$
    is the identity matrix, we obtain that the left side of \cref{eq:ae:sub:2}
    is equal to $1$ if $F' = H$ and is equal to zero otherwise. Using \cref{eq:ae:sub:2},
    we continue our computation and obtain $\ae{\Phi}{H} = (-1)^{\#E(H)} \subCoeff{\Phi}{k}{H}$.
\end{proof}

Now \cref{theo:sub:base} allows us to rewrite $\Phi_k$ as a linear combination of subgraph counts.
This in turn allows us to write $\NUM{}\indsubs{(\Phi_k, k)}{\star}$ as a linear linear combination
of subgraph counts. We use this fact later to show that $\NUM{}\indsubs{(\Phi_k, k)}{\star}$ is
easy to compute for certain graph parameters.

\begin{theorem}\label{theo:indsub:sub:base}
    For each $\Phi_k \colon \graphs{k} \to \Q$ with $\Phi_k(G) = \sum_{H \in \graphs{k}^* } \subCoeff{\Phi}{k}{H} \cdot \NUM{}\subs{H}{G}$,
    we have
    \[\NUM{}\indsubs{(\Phi_k, k)}{G} =  \sum_{H \in \graphs{k}^* }  \subCoeff{\Phi}{k}{H} \cdot \NUM{}\subs{H}{G}
    = \sum_{H \in \graphs{k}^* }  (-1)^{\#E(H)} \cdot \ae{\Phi}{H} \cdot \NUM{}\subs{H}{G},\]
    for all $G \in \graphs{}$.
\end{theorem}
\begin{proof}
    We first show the statement for $\Phi_k(G) = \NUM{}\subs{H}{G}$ where $H \in \graphs{k}^*$ is a fixed graph.
    For each graph $G$, we obtain
    \[\NUM{}\indsubs{(\Phi_k, k)}{G} = \sum_{A \in \binom{V(G)}{k}} \Phi_k(G[A]) = \sum_{A \in \binom{V(G)}{k}} \NUM{}\subs{H}{G[A]}. \]
    Next, we show that $\NUM{}\subs{H}{G}$ is equal to $\sum_{A \in \binom{V(G)}{k}} \NUM{}\subs{H}{G[A]}$.
    Note that $\NUM{}\subs{H}{G}$ counts the number of subgraphs $H'$ in $H$ that are isomorphic to $G$ and each $H'$ is in the set
    $\subs{H}{G[A]}$  if and only if $A = V(H')$. Thus $\subs{H}{G[A]}$ forms a partition of the set $\subs{H}{G}$. This shows $\NUM{}
    \indsubs{(\Phi_k, k)}{G} = \NUM{}\subs{H}{G}$.

    For the general case,
    we use
    \[\NUM{}\indsubs{(\alpha \Phi_k + \beta \Psi_k, k)}{G} = \alpha \NUM{}\indsubs{(\Phi_k , k)}{G} + \beta \NUM{}\indsubs
    {(\Psi_k , k)}{G}.\]
    We obtain \[\sum_{H \in \graphs{k}^* }  \subCoeff{\Phi}{k}{H} \cdot \NUM{}\subs{H}{G}
    = \sum_{H \in \graphs{k}^* }  (-1)^{\#E(H)} \cdot \ae{\Phi}{H} \cdot
\NUM{}\subs{H}{G}\] due to \cref{theo:sub:base}.
\end{proof}

\subsection{$\NUM{}\indsubsprob(\Phi)$ is Easy for Bounded Vertex Cover Size}

In this section, we show that there are nontrivial graph parameters for which
we can compute $\NUM{}\indsubsprob(\Phi)$ easily. For this, the
vertex-cover size of a graph $H$ plays an important role (that is, the size of the smallest
vertex-cover of $H$). We write $\VC{H}$ for the vertex-cover size of $H$.
Recall that the vertex-cover size is an upper bound for the treewidth of $H$.

By \cite[Theorem 1.1]{hom:basis}, we can compute $\NUM{}\subs{H}{G}$ in time
$O(|V(H)|^{O(|V(H)|)} |V(G)|^{t+ 1})$, where $t$ is the maximum treewidth in the spasm of $H$.
However, this
value is upper bounded by the vertex-cover size of $H$ (see \cite[Fact 3.4]{hom:basis}).
This in turn
means that we can compute $\NUM{}\subs{H}{G}$ in time $O(|V(H)|^{O(|V(H)|)}  |V(G)|^{\VC{H} + 1})$.

\begin{theorem}\label{theo:indsub:easy}
    If $\Phi \colon \graphs{} \to \Q$ is a graph parameter such that there is a constant $\tau$
    with $\ae{\Phi}{G} = 0$ for all graphs $G$ with $\VC{G} > \tau$, then
    $\NUM{}\indsubsprob(\Phi)$ is FPT and there is an algorithm that computes
    $\NUM{}\indsubs{(\Phi, k)}{G}$ in time $O(f(k) \cdot |V(G)|^{\tau + 1})$
    for some computable function $f$.
\end{theorem}
\begin{proof}
    For each $k \in \nat$, we use $\Phi_k \colon \graphs{k} \to \Q$ to denote the restriction of $\Phi$ to $k$-vertex graphs. Observe that
    $\NUM{}\indsubs{(\Phi, k)}{G} = \NUM{}\indsubs{(\Phi_k, k)}{G}$. From \cref{theo:indsub:sub:base}, we obtain
    \[\NUM{}\indsubs{(\Phi_k, k)}{G} =  \sum_{H \in \graphs{k}^* }  (-1)^{\#E(H)} \cdot \ae{\Phi}{H} \cdot \NUM{}\subs{H}{G} \qquad \text{for all $G \in \graphs{}$.} \]
    For each graph $H$ with $\ae{\Phi}{H} \neq 0$, we can compute $\NUM{}\subs{H}{G}$ in time
    $O(k^{O(k)} |V(G)|^{\VC{H} + 1})$. Thus, we can compute
    $\NUM{}\indsubs{(\Phi_k, k)}{G}$ in time $O(f(k) \cdot |V(G)|^{\tau + 1})$.
\end{proof}

\begin{corollary}
    For each $c \in \nat$, let $\Phi_c(G) \coloneq (\#E(G))^c$ be a graph parameter. There is an algorithm that computes
    $\NUM{}\indsubs{(\Phi_c, k)}{G}$ in time $O(f(k) \cdot |V(G)|^{c + 1})$ for some computable function $f$.
\end{corollary}
\begin{proof}
    For each fixed $c$, we compute alternating enumerator of $\Phi_c$ which yields
    \begin{align}\label{eq:ae:edge_count}
        \ae{\Phi_c}{H} =  \sum_{S \subseteq E(H)} \Phi_c(\ess{H}{S}) (-1)^{\#S} = \sum_{k = 0}^{k} (-1)^{k} \binom{\#E(H)}{k} k^c =
    (-1)^{\#E(G)} \sum_{k = 0}^{\#E(G)} (-1)^{\#E(G) - k} \binom{\#E(H)}{k} k^c,
    \end{align}
    where we used that we can group together all $\binom{\#E(H)}{k}$ many graphs subgraphs of $H$ that contain $k$ edges. By
    \cite[Chapter 2.4, (9)]{enumeration_book}, the sum (\ref{eq:ae:edge_count}) is zero whenever $\#E(H) > c$.
    Meaning, that $\aek{\Phi}{c}{H} = 0$ for all graphs
    that contain more than $c$ edges, which implies that $\aek{\Phi}{c}{H} \neq 0 $ is only possible if $\VC{H} \leq \#E(H) \leq c$.
    Now, \cref{theo:indsub:easy} shows that
    we can compute $\NUM{}\indsubs{(\Phi_c, k)}{G}$ in time $O(f(k) \cdot |V(G)|^{c + 1})$ for some computable function $f$.
\end{proof}

In the following, we write $S_n \in \graphs{n}$ for the star graph with $n$ vertices. The edge set $S_n$ is
equal to $\{\{1, j\} : 2 \leq j \leq n\}$. Next, we show that there are graph parameters whose codomain
on $k$-vertex graphs is equal to $\{0, \dots, k\}$ but $\NUM{}\indsubsprob(\Phi)$ is easy to compute.

\begin{corollary}
    Let $\Phi(G) = \NUM{}\subs{S_{|V(G)|}}{G}$ denote the graph parameter that counts the
    number of universal vertices, then $\Phi_k \colon \graphs{k} \to \{0, \dots k\}$ and
    $\NUM{}\indsubsprob(\Phi)$ can be computed in time $O(f(k) \cdot |V(G)|^2)$.
\end{corollary}
\begin{proof}
    It is easy to see that $\Phi_k$ maximal for $K_k$ and that $\Phi_k(K_k) = k$. This
    immediately implies that the codomain of $\Phi_k$ is equal to $\{0, \dots k\}$.
    Further, due to \cref{theo:indsub:sub:base} we get $\ae{\Phi}{G} \neq 0$ if
    and only if $G = S_{|V(G)|}$. This allows to use \cref{theo:indsub:easy}
    which shows that $\NUM{}\indsubsprob(\Phi)$ can be computed in time $O(f(k) \cdot |V(G)|^2)$ since $\VC{S_n} = 1$.
\end{proof}

\subsection{$\NUM{}\indsubsprob(\Phi)$ is Harder than $\NUM{}\cpindsubsprob(\Phi)$}\label{sec:indsub:harder}

We write $P_n$ for the path graph that contains $n$ vertices.

\begin{lemma}
    Let $\Phi(G) = \NUM{}\subs{P_{|V(G)|}}{G}$ denote the graph parameter that counts the
    number of Hamiltonian paths, then
    \begin{itemize}
        \item $\NUM{}\indsubsprob(\Phi)$ is \w-hard.
        \item $\NUM{}\cpindsubsprob(\Phi, \mathcal{H})$ is FPT for any recursively enumerable class of graphs $\mathcal{H}$.
        Further, there is an algorithm that computes $\NUM{}\cpindsubs{(\Phi, H)}{G}$ in time $O(f(|V(H)|) \cdot |V(G)|^{2})$
        for some computable function $f$.
    \end{itemize}
\end{lemma}
\begin{proof}
    \begin{itemize}
        \item Observe that the vertex-cover size of the $k$-vertex path $P_k$ is equal to $\lceil k/2 \rceil$. Let
        $\mathcal{P} = \{P_k : k \in \nat \}$ be the set of all paths. Then, $\NUM{}\subprob(\mathcal{P})$ is \w-hard due
        to \cite[Theorem 1.1]{sub}. Further, for each $k \in \nat$, we use $\Phi_k \colon \graphs{k} \to \Q$
        to denote the restriction of $\Phi$ on $k$-vertex graphs, thus
        $\Phi_k(G) =  \NUM{}\subs{P_{k}}{G}$. We now apply \cref{theo:indsub:sub:base} to each $\Phi_k$ and obtain
        \[\NUM{}\indsubs{(\Phi, k)}{G} = \NUM{}\indsubs{(\Phi_k, k)}{G} =  \NUM{}\subs{P_k}{G} \qquad \text{for all $G \in \graphs{}$.}\]
        This implies that $\NUM{}\indsubsprob(\Phi)$ is also \w-hard.
        \item For each $H \in \mathcal{H}$ and graph $G$, we use \cref{lem:chi:cpindub} and obtain
        \[\NUM{}\cpindsubs{(\Phi, H)}{G}
        = \sum_{A \subseteq E(H)}  (-1)^{\#A} \cdot \ae{\Phi}{\ess{H}{A}} \cdot \NUM{}\cphoms{\ess{H}{A}}{G}.\]
        Next, we use $\Phi_k(G) = \subs{P_k}{G}$ and obtain $\ae{\Phi}{\ess{H}{A}} \neq 0$ if and only
        if $\ess{H}{A} \cong P_{|V(H)|}$ due to \cref{theo:sub:base}. This means that
        \[\NUM{}\cpindsubs{(\Phi, H)}{G} =
           \NUM{}\subs{P_{|V(H)|}}{H} \cdot \NUM{}\cphoms{P_{|V(H)|}}{G}. \]
        Thus, in order to compute $\NUM{}\cpindsubs{(\Phi, H)}{G}$, we first compute $\NUM{}\subs{P_{|V(H)|}}{H}$ in time
        $O(f_1(|V(H)|))$ for some computable function $f_1$. Afterward, we compute $\NUM{}\cphoms{P_{|V(H)|}}{G}$,
        which can be done in time $O(f_2(|V(H)|) \cdot |V(G)|^{2})$ since the treewidth of $P_{|V(H)|}$ is
        equal to $1$ (see \cite[Theorem 2.35]{roth_phd}. Note that the DP algorithm has to be
        modified to work with $\NUM{}\cphom$). Thus, we can compute $\NUM{}\cpindsubs{(\Phi, H)}{G}$
        in time $O(f(|V(H)|) \cdot |V(G)|^{2})$ for some computable function~$f$.
        \qedhere
    \end{itemize}
\end{proof}

\begin{remark}
    The difference of $\NUM{}\indsubsprob(\Phi)$ and $\NUM{}\cpindsubsprob(\Phi, \mathcal{H})$ can be best understood by choosing
    $\mathcal{H} = \{K_k : k \in \nat\}$. Now, according to \cref{lem:chi:cpindub} and \cref{theo:indsub:sub:base} we obtain
    \begin{align}
        \NUM{}\indsubs{(\Phi, k)}{G} &= \sum_{H \in \graphs{k}^* }  (-1)^{\#E(H)} \cdot
        \ae{\Phi}{H} \cdot \NUM{}\subs{H}{G}\qquad\text{and} \label{eq:rem:indsub} \\
        \NUM{}\cpindsubs{(\Phi, K_k)}{G} &= \sum_{H \in \graphs{k}^* }  (-1)^{\#E(H)}
        \cdot \frac{k!}{\NUM{}\aut(H)} \cdot  \ae{\Phi}{H} \cdot \NUM{}\cphoms{H}{G}.\label{eq:rem:cpindsub}
    \end{align}
    This means that $\NUM{}\indsubs{(\Phi, k)}{G}$ can be understood as a sum of $\NUM{}\subs{H}{G}$ where a term appears
    whenever $\ae{\Phi}{H} \neq 0$ and $\NUM{}\cpindsubs{(\Phi, K_k)}{G}$ can be understood as a sum of $\NUM{}\cphoms{H}{G}$
    where a term appears whenever $\ae{\Phi}{H} \neq 0$. Now, the problem $\NUM{}\subprob(\mathcal{F})$ is \w-hard if and
    only if the class $\mathcal{F}$ has unbounded vertex cover size. However, $\NUM{}\cphomsprob(\mathcal{F})$ is
    \w-hard if and only if the class $\mathcal{F}$ has unbounded treewidth. Note, that the treewidth is a lower bound for
    the vertex cover size and that there are classes $\mathcal{F}$ of graphs with unbounded vertex cover size but bounded
    treewidth. This explains why $\NUM{}\cpindsubsprob(\Phi, \mathcal{H})$ could be FPT in cases where
    $\NUM{}\indsubsprob(\Phi)$ is \w-hard. Still, there is a parameterized Turing reduction
    from  $\NUM{}\cpindsubsprob(\Phi, \mathcal{H})$ to $\NUM{}\indsubsprob(\Phi)$ (see \cite[Lemma 10]{alge}).

    Nevertheless, it is often much simpler to use $\NUM{}\cpindsubsprob(\Phi, \mathcal{H})$ since we can apply
    \emph{complexity monotonicity} on the $\NUM{}\cphom$-terms (see \cite[Lemma 7]{alge}). Hence, the hardness of
    $\NUM{}\cpindsubs{(\Phi, K_k)}{G} $ is governed by the term with highest treewidth that appears in the sum
    of \cref{eq:rem:cpindsub}. However, it might be that nontrivial cancellations occur in the sum
    of \cref{eq:rem:indsub}, which is why this sum is much harder to analyze. See \cite[Example 1.12]{hom:basis}
    for an example of nontrivial cancelations.
\end{remark}

\section{Hardness Results for Modular Counting}\label{sec:mod}

In this section, we show that our hardness results from
\cref{sec:reduction,sec:edge_monotone} stay true when considering modulo counting.
To that end, we introduce the following two classes of problems.
\begin{definition}
    For each prime number $p$, and parameterized counting problem
    $(\NUM{}A, \kappa)$, we define the following two problems:
    \begin{itemize}
        \item The parameterized decision problem $(\MODP{p}A, \kappa)$ decides on
        input $x$ if $\NUM{}A(x)$ is divisible by $x$.
        \item The parameterized problem $(\NUMP{p}A, \kappa)$ computes on input
        $x$ the values $\NUM{}A(x) \bmod p$.
        \qedhere
    \end{itemize}
\end{definition}

The definition of $\MODP{p}A$ was introduced by Beigel, Gill, and Hertramp in \cite{Beigel_MOD_p}
and is a generalization of $\MODP{2}A = \oplus A$ which was already studied by
Papadimitriou and Zachos back in the 1980s (see \cite{parity_p}).
However, Faben \cite{Faben_sharp_p} realized that the definition of $\MODP{p}A$ does not
completely capture the complexity of modulo counting which is why
he introduced the definition $\NUMP{p}A$.

It is easy to see that an FPT algorithm for $(\NUMP{p}A, \kappa)$ yields an
FPT algorithm for $(\MODP{p}A, \kappa)$. However, already
in \cite[p. 3]{Faben_sharp_p} Faben mentioned that one can design artificial problems
for which $\NUMP{p}A$ is substantially harder than $\MODP{p}A$, even though
he was not able to find natural examples for this phenomenon.
Observe that for $p = 2$, the problems $(\NUMP{2}A, \kappa)$
and $(\MODP{2}A, \kappa)$ are equivalent. A detailed discussion of the
differences between the complexity classes $\NUMP{p}P$ and $\MODP{p}P$
can be found in Faben's PhD thesis \cite[Section 3.1.2]{Faben_PhD}.

Due to the differences between $\NUMP{p}A$ and $\MODP{p}A$, we introduce two different complexity classes
for which we rely on the $\NUM{}\clique{}$ problem.

\begin{definition}
    For each prime number $p$, we define the following classes of problems.
    \begin{itemize}
        \item The class of $\wMODP{p}$-hard problems consists of all problems $(\MODP{p}A, \kappa)$ that
        have a parameterized Turing reduction from $\MODP{p}\clique{}$ to $\MODP{p}P$.
        \item The class of $\wP{p}$-hard problems consists of all problems $(\NUMP{p}A, \kappa)$ that
        have a parameterized Turing reduction from $\NUMP{p}\clique{}$ to $\NUMP{p}P$.
        \qedhere
    \end{itemize}
\end{definition}

In \cref{lem:clique:mod:equiv}, we show that the problem $\MOD{}\clique$ and
$\NUMP{p}\clique$ are equivant to each other under parameterized Turing reduction.
However, this does not imply that the class of $\wMODP{p}$-hard problems is equivalent to
the class of $\wP{p}$-hard problems.
Further, assuming rETH, we show in \cref{sec:appendix:mod} that there is a constant $\delta$,
such that for all prime numbers $p$, no algorithm decides $\MOD{}\clique{}$
on input $(G, k)$ in time $O(|V(G)|^{\alpha k})$.

\begin{restatable*}{corollary}{corcliquemodeth}\label{cor:k-clique:mod:ETH}
    Assuming rETH, there is a constant $\gamma > 0$ such that for all prime numbers $p$, no
    algorithm (that reads the whole input) solves $\MODP{p}\clique{}$ on input $(G, k)$ in time $O(|V(G)|^{\alpha k})$.
\end{restatable*}

To show the hardness of $\NUMP{p}\indsub(\Phi)$, we use the following corollary
which is a variation of \cref{cor:tight:bounds}.

\begin{corollary}\label{cor:tight:bounds:mod}
    Let $\Phi \colon \graphs{} \to \{0, \dots, c\}$ be a graph parameter, $0 < \psi \leq 1$ and $p$ be a prime number.
    Assume that there set of graphs $A \coloneqq \{G_k\}$ and a function $h \colon A \to \nat$ such that each $G_k$
    contains $K_{h(G_k), h(G_k)}$ as a subgraph and $\ae{\Phi}{G_k} \not \equiv_p 0$.
    \begin{itemize}
        \item If $h(G) \in \omega(1)$, then the problem $\NUMP{p}\indsub(\Phi)$ (and $\NUMP{p}\cpindsubsprob(\Phi)$) is $\wMOD$-hard and
        $\wP{p}$-hard.
        \item If $h(G) \in \Omega(|V(G)|^\psi)$, then there is a global constant $\gamma > 0$
        such that for each fixed $k = |V(G_j)|$, there is no algorithm (that reads the
        whole input)
        that for every \(G\) computes the number $\NUMP{p}\indsubs{(\Phi, k)}{G}$ (or
        $\NUMP{p}\cpindsubs{(\Phi, \graphs{k})}{G}$) in time $O(|V(G)|^{\gamma  k})$
        unless rETH fails.
        Here, the constant $\gamma$ depends on $h$ but is independent of $\Phi$ and $k$.
        \item If $h(G) \in \Omega(|V(G)|^\psi)$, then there is no algorithm (that reads
            the whole input) that for every \(G\) computes
        $\NUMP{p}\indsubs{(\Phi, k)}{G}$  (or $\NUMP{p}\cpindsubs{(\Phi, \graphs{k})}{G}$)
        in time $O(f(k) \cdot |V(G)|^{o(k)})$ for all computable functions $f$ unless rETH
        fails.
    \end{itemize}
\end{corollary}
\begin{proof}
    We show that an algorithm for $\NUMP{p}\indsubs{(\Phi, k)}{\star}$ (or $\NUMP{p}\cpindsubs{(\Phi, \graphs{k})}{\star}$)
    can be used to compute $\NUMP{p} \clique{}$. This is enough since $\NUMP{p} \clique{}$
    and $\MODP{p} \clique{}$ are equivalent.

    We emulate the proof \cref{cor:tight:bounds}. For this, observe that the chain of reductions
    (\ref{reduction:chain}) still holds when taking all computations modulo $p$. This was already
    shown in \cite[Appendix A]{alge}. In particular, we observe that we can still write
    $\NUMP{p}\cpindsub(\Phi)$ as a linear combination of $\NUMP{p}\cphom$-terms
    over the finite field $\field{p}$ (see \cref{lem:chi:cpindub}). This means that our
    reduction is possible if we can find graphs with large bicliques whose
    alternating enumerator is nonvanishing mod $p$.
    This is exactly the case when there is a graph $G_k$ that contains $K_{h(G_k), h(G_k)}$
    as a subgraph and $\ae{\Phi}{G_k} \not \equiv_p 0$. We obtain the tight bounds under
    rETH by replacing \cref{theo:k-clique:ETH} with
    \cref{cor:k-clique:mod:ETH} in the proof of \cref{theo:lower:bound:bicliques}.
\end{proof}

\thmedgemonMOD
\begin{proof}
    The proof of this theorem works exactly like the proof of \cref{theo:edge_mono} with
    the only difference that we replace
    \cref{cor:tight:bounds} with \cref{cor:k-clique:mod:ETH} and we observe that the
    reduction from \cref{lem:inc:exc}
    still works for modulo counting. Hereby it is important to point out that
    in the proof of \cref{theo:edge_mono}, we picked a prime $p > c$ and only
    considered prime powers of the form $q(k) \coloneqq p^t$. This allows us to again use
    \cref{theo:edge_mono:prime:power:bicliques}
    to get our hardness results for the prime power case. Note, that \cref{theo:edge_mono:prime:power:bicliques}
    ensures that we can find graphs that contain large bicliques and whose alternating
    enumerator is nonvanishing modulo $p$ (see definition of concentrated);
    or we can find a reduction to $\Scat{\Phi}{q}$ which is nontrivial on prime
    powers of the form $p^t$.
\end{proof}

Lastly, observe that each graph property $\Phi$ can be viewed as a graph parameter with
codomain $\{0, 1\}$.
Thus, we can apply \cref{cor:hardness:mod} to obtain that $\NUMP{p}{}\indsubsprob(\Phi)$ is
$\wP{p}$-hard for all prime numbers $p$.

\coredgemonMOD*

{
    \bibliographystyle{alphaurl}
    \bibliography{main}
}

\appendix
\clearpage
\normalsize
\section{The Fixed Points of Automorphism Groups}

The following explanations, definitions, and statements are taken from \cite{edge_monotone}. However,
the idea of using fixed points to analyze the alternating enumerator was originally developed in \cite{alge}.

\subsection{The Fixed Points of Graph Automorphisms} \label{app:fix:basics}

For each graph $G$, the set of automorphisms $\auts{G}$ forms a group with composition as
the group operation. Observe that the automorphism group of a clique \(\auts{K_n}\) is just
the symmetric group \(\sym{n}\).

Let \(G\) denote a graph and consider a subgroup \(\Gamma \subseteq \auts{G}\).
Any \(g \in G\) is a bijection on \(V(G)\), which we may interpret as \(g\) permuting the
vertices of \(G\).
We may also interpret $g$ as an operation on the edges of \(G\).
However, as we wish to use results from the literature (such as
the Orbit-stabilizer Theorem), we need to describe this operation on the edges of \(G\)
using the language of group actions.

Formally, we can turn the operation on vertices into a {\em
group action} $\cdot : \Gamma\times E(G)\to E(G)$ that tells us how each member of
$\Gamma$ moves the edges of $G$.
Specifically, we define $g\cdot \{u,v\} \coloneqq \{g(u),g(v)\}$.
We interpret said group action \(\cdot\) multiplicatively and typically write just \(g\{u, v\}\).

Extending the previous group action, \(\Gamma\) also acts on
edge-subgraphs of \(G\) via
\begin{align*}
    V(g\cdot A) &\coloneqq V(A),\; \text{and}\\
    E(g\cdot A) &\coloneqq \{\{g(u), g(v)\} \mid \{u, v\} \in E(A)\},
\end{align*}
where \(g \in \Gamma\) and \(A\) is an edge-subgraphs of \(G\).
Again, we interpret said group actions \(\cdot\) multiplicatively and typically write just \(gA\).

For each $\{u, v\} \in E(G)$ the set $\Gamma \cdot \{u, v\} \coloneqq \{g \cdot \{u, v\}
\mid g \in \Gamma\}$ is the \emph{orbit} of $\{u, v\}$.
Two different edges either have disjoint orbits or equal orbits.
We write \(E(G)/\Gamma\) to denote the set of all orbits of \(\cdot\).
Recall that \(E(G)/\Gamma\) forms a partition of \(E(G)\).
In a slight abuse of notation, for an orbit \(O\), we also write \(O\) for the
edge-subgraph \(\ess{G}{O}\); similarly, we use \(E(H)/\Gamma\) to denote the set of all
such edge-subgraphs.

We say that an edge-subgraph $A$ is a \emph{fixed point} of $\Gamma$ in $G$
if $gA = A$ for all $g \in \Gamma$.
We write $\fps{G}$ for the set of all fixed points of $\Gamma$ in $G$. As it turns out, a graph \(A\)
being a fixed point of \(\Gamma\) in \(G\) is a strong property that is
very useful to us. For now and as a first useful observation, we see that \(H\) inherits key
properties from \(G\).

\begin{lemmaq}[See {\cite[Lemma 3.1]{edge_monotone}}]\label{remark:subgroup:fixpoint}
    Let \(G\) denote a graph and let $\Gamma \subseteq \auts{G}$ denote a group.
    For any $H \in \fps{G}$ all of the following hold.
    \begin{enumerate}[(1)]
        \item We have $\Gamma \subseteq \auts{H}$.
        \item Any fixed point of \(\Gamma\) in \(H\) is a fixed point of \(\Gamma\) in
            \(G\) and
            $\fps{H} = \{A \in \fps{G} \mid E(A) \subseteq E(H)\}$.
            \qedhere%
    \end{enumerate}
\end{lemmaq}

Intuitively, \cref{remark:subgroup:fixpoint} implies that fixed points of fixed points are again fixed points,
which allows us to define a structure on the set of fixed points. Now, the following lemma shows that
this structure of fixed points is made up of orbits, meaning that orbits are the
\emph{basic building blocks} of $\fp(\Gamma, H)$.

\begin{lemmaq}[See {\cite[Lemma 4.1]{edge_monotone}}]\label{remark:fixed:point:union}
    Let \(H\) denote a graph and let \(\Gamma \subseteq \aut(H)\) denote a group.
    Further, let \(\Gamma\) act on \(E(H)\) and write \(E(H)/\Gamma\) for the set of all
    resulting orbits.
    Finally, let \(\Gamma\) act on edge-subgraphs of \(H\) and write \(\fp(\Gamma, H)\) for
    the set of all resulting fixed points of \(\Gamma\) in \(H\).

    Then, the edge set of each fixed point $F \in \fp(\Gamma, H)$ is the (possible empty)
    disjoint union of orbits $O_1, \dots, O_s$, where $O_i \in E(H)/\Gamma$
    and each disjoint union of orbits yields a fixed point.
    The partition into orbits is unique.
\end{lemmaq}

\cref{remark:fixed:point:union} induce an order of the fixed points in
\(\fps{H}\): we say a fixed point \(F\) is a \emph{sub-point} of another fixed point
\(G\) if we can obtain \(G\) from the union of \(F\) and (potentially multiple) orbits in
\(E(H)/\Gamma\). This allows us to define fixed points as the union of orbits.

\begin{definition}[See {\cite[Definition 4.5]{edge_monotone}}]\label{def:10.4-1}
    Let \(H\) denote a graph and let \(\Gamma \subseteq \auts{H}\) denote a group.
    Further, let \(\Gamma\) act on \(E(H)\) and write \(E(H)/\Gamma\) for the set of all
    resulting orbits.
    Finally, let \(\Gamma\) act on edge-subgraphs of \(H\) and write \(\fps{H}\) for
    the set of all resulting fixed points of \(\Gamma\) in \(H\).

    For a fixed point \(F\in \fps{H}\), its \emph{orbit factorization} \(\od{F}\) is the
    unique subset of \(E(H)/\Gamma\) whose union is~\(F\):
    \[
    F = \bigcup \od{F}.
    \]
    The \emph{level} of \(F\), denoted by \(\hl{F}\), is the size of the orbit
    factorization of \(F\)
    \[
    \hl{F} \coloneqq |\od{F}|.
    \]
    Finally, for two fixed points \(F_1, F_2  \in \fps{H}\), we say that \(F_2\) is a \emph{sub-point}
    of \(F_1\), denoted by \(F_2 \subseteq F_1\), if the orbit factorization of \(F_2\) is
    a subset of the orbit factorization of \(F_1\). If the inclusion is strict, we say
    that \(F_2\) is a \emph{proper sub-point} of \(F_1\).
\end{definition}

The next lemma makes it possible to group the sub-points of a fixed point according to their level. It can
be proven by observing that each orbit divides the group order $|\Gamma|$, which is a $p$ power, due to the orbit-stabilizer-theorem.

\begin{lemmaq}[See {\cite[Lemma 4.6]{edge_monotone}}]\label{lem:10.4-2}
    Let \(H\) denote a graph and let \(\Gamma \subseteq \auts{H}\) denote a \(p\)-group.

    For any fixed point \(F \in \fps{H}\), we have \[
        (-1)^{\NUM{}E(F)} \equiv_p (-1)^{\hl{F}}. \qedhere%
    \]
\end{lemmaq}

Finally, we are at the point where we can use fixed points to compute the alternating enumerator
modulo a prime number. For this, we use the following lemma which was originally proven in
\cite{alge}.

\begin{lemma}[\cite{alge}]\label{lem:chi:comp}
    Let \(H\) denote a graph, let \(\Gamma \subseteq \auts{H}\) denote a \(p\)-group, and
    let  $\Phi \colon \graphs{} \to \Z$ denote a graph parameter, then
    \begin{align*}
        \ae{\Phi}{H} \equiv_p \sum_{ A \in \fp(\Gamma, H) } \Phi(A) (-1)^{\NUM{}E(A)} \equiv_p \sum_{ A \subseteq H} \Phi(A) (-1)^{\hl{A}} .
        \tag*{\qedhere}
    \end{align*}
\end{lemma}
\begin{proof}
    The proof originates from the proof of \cite[Lemma~1]{alge}. The main idea is
    that we can group together edge-subgraphs of $H$ that are in the same orbit
    of the group action $\cdot$ that takes as input a group element $g \in \Gamma$
    and an edge-subgraph $A$ of $H$, and outputs
    $g \cdot A$. Observe that graphs from the same orbit are always isomorphic
    to each other, thus $\Phi(A) = \Phi(g \cdot A)$, which allows us to group them
    together in the sum. Due to the orbit-stabilizer-theorem, the size of each orbit
    divides the group order of $\Gamma$, which is a $p$ power. Thus, if we compute
    the alternating enumerate mod $p$, then only the orbits of size $1$ (that is, fixed
    points) remain.
\end{proof}

We mainly use \cref{lem:chi:comp} to show that the alternating enumerator is non-vanishing in some special cases.

\lemcompae
\begin{proof}
    We show that $\ae{\Phi}{A} \equiv_p (-1)^{\hl{A} + 1}(b - a)$ which is sufficient since
    $b - a \not \equiv_p 0$.
    By definition, the fixed point \(A\) has exactly \(\binom{\hl{A}}{i}\)
    sub-points that have a level of exactly \(i\).
    Starting from \cref{lem:chi:comp}, we rewrite the alternating enumerator and obtain
    \begin{align*}
        \ae{\Phi}{A}
        &\equiv_p \sum_{ B \subseteq A } \Phi(B) (-1)^{\hl{B}} \\
        &\equiv_p -(-1)^{\hl{A}}(b - a) + \sum_{ B \subseteq A } b (-1)^{\hl{B}}
        & \text{(\(\Phi(A) = a\) and \(\Phi(B) = b\) for all \(B \subsetneq A\))}\\
        &\equiv_p  -(-1)^{\hl{A}}(b - a) + b \sum_{i = 0}^{\hl{A}} \binom{{\hl{A}}}{i} (-1)^i
        & \text{(group by level)}\\
        &\equiv_p -(-1)^{\hl{A}}(b - a) + b (-1 + 1)^{\hl{A}}
        & \text{(Binomial Theorem)} \\
        &\equiv_p (-1)^{\hl{A} + 1}(b - a).
    \end{align*}
    In total, this yields the claim.
\end{proof}

\subsection{The $p$-Sylow groups of \(\aut(K_{p^m})\) and its Fixed Point Structure}\label{sec:sylow}

We intend to construct a special subgroup of \(\auts{K_{p^m}}\) that we call the $p$-Sylow groups of \(\aut(K_{p^m})\).
To make things simpler, we identify the vertex set of $K_{p^m}$ with \(\fragmentco{0}{p}^{m}\). We start by
defining special types of bijections on \(\fragmentco{0}{p}^{m}\).

\begin{definition}[See {\cite[Definition 7.1]{edge_monotone}}]\label{11.12-3}
    Consider a prime $p$ and a positive integer $m$.
    For each \(j \in \fragmentco{0}{m}\), write \(\varphi_i\) for a function
    \(\fragmentco{0}{p}^{j} \to \fragmentco{0}{p}\) and set \(\varphi \coloneqq
    (\varphi_0,\dots,\varphi_{m-1})\).
    We define the function $\sylelm : \fragmentco{0}{p}^{m} \to \fragmentco{0}{p}^{m}$ via\footnote{We write $\varphi_0$ for
    $\varphi_0(())$ since $\varphi_0$ is a function that is defined on a single element}
    \[\sylelm(x_1, \dots, x_m) \coloneqq (x_1 + \varphi_0, x_2 + \varphi_1(x_1), x_3 + \varphi_2(x_1,
    x_2), \dots, x_m + \varphi_{m-1}(x_1, \dots, x_{m-1})), \] where all computations are done
    modulo $p$.

    We write \(\overline{p}^m\) for the set of all functions $\sylelm :
    \fragmentco{0}{p}^{m} \to \fragmentco{0}{p}^{m}$ that are obtained in this fashion,
    that is,
    \[
        \overline{p}^m \coloneqq \{ \sylelm \mid \varphi_i \in \fragmentco{0}{p}^{j} \to
        \fragmentco{0}{p} \}.
        \qedhere
    \]
\end{definition}

Now, it is easy to check that, together with composition $\circ$ as a group operation,
these specific bijections form a transitive $p$-group.

\begin{lemmaq}\label{11.13-1}
    The pair $\sylow_{p^m} \coloneqq (\overline{p}^m, \circ)$ is a transitive $p$-group on the set $V(K_{p^m})$.
\end{lemmaq}

We see that the fixed point structure of $\fp(\sylow_{p^m}, K_{p^m})$ has some very nice properties (e.g. containing large
bicliques). However, we need the lexicographic product of graphs in order to describe the fixed points of $\sylow_{p^m}$ in $K_{p^m}$.

\begin{definition}[See {\cite[Definition 7.8]{edge_monotone}}]
    For graphs $G_1, \dots, G_m$,
    we define their \emph{lexicographic product} $G_1 \wrGraph \cdots \wrGraph G_m$
    via
    \begin{align*}
        V(G_1 \wrGraph \cdots \wrGraph G_m) &\coloneqq
        V(G_1) \times \cdots \times V(G_m) \quad\text{and}\\
        E(G_1 \wrGraph \cdots \wrGraph G_m) &\coloneqq
        \{ \{(u_1, \dots, u_m), (v_1, \dots, v_m) \}\\&\hskip3em
        \mid \text{there is an $i \in \setn{m}$ with  $u_j = v_j$ for all $j < i$ and $\{u_i,
        v_i\} \in E(G_i)$}\}.
        \tag*{\qedhere}
    \end{align*}
\end{definition}
As an easy example, observe that we have  $G_1 \wrGraph G_2 \cong \metaGraph{G_1}{G_2,
\dots, G_2}$ and $G_1 \wrGraph G_2    \wrGraph G_3 \cong \metaGraph{G_1}{G_2 \wrGraph G_3,
\dots, G_2 \wrGraph G_3}$. Now, it turns out that the fixed points of $\fp(\sylow_{p^m}, K_{p^m})$
are the lexicographic product of so-called \emph{difference graphs}.

\begin{definition}[See {\cite[Definition 2.4]{edge_monotone}}]\label{def:difference:graphs}
    For a prime \(p\), an integer $m > 0$, and a set \(A \subseteq \field{p^m}^+\),
    we define the \emph{difference graph} \(\cgr{p^m}{A}\) via
    \begin{align*}
        V( \cgr{p^m}{A} ) &\coloneqq \field{p^m} \quad\text{and}\quad
        E( \cgr{p^m}{A} ) \coloneqq \{ \{u, v\} \mid u, v \in \field{p^m},  (u - v) \in A \cup (-A) \},
    \end{align*}
    where $-A = \{-x \mid x \in A\}$. Observe that $\cgr{p^m}{A} = K_{p^m}$ whenever \(A = \field{p^m}^+\).

    The \emph{level} of a difference graph $\cgr{p^m}{A}$ is the cardinality of $A$; we write
    \(\hl{\cgr{p^m}{A}}\) for the level of \(\cgr{p^m}{A}\).
\end{definition}

Note that for all odd prime numbers $p$ the set $\field{p}^+$ is equal to $\{1, \dots, (p-1)/2\}$ and
that $\field{2}^+$ is equal to $\{1\}$.

\begin{lemmaq}[See {\cite[Lemma 7.11]{edge_monotone}}]\label{lemma:fixed:points:p^m}
    For any prime \(p\) and any positive integer \(m\), we have
    \begin{align*}
        \fpb{\sylow_{p^m}}{K_{p^m}}
        = \{\cgr{p}{A_1} \wrGraph \cdots \wrGraph \cgr{p}{A_m} \mid A_i \subseteq \field{p}^+\}.
        \tag*{\qedhere}
    \end{align*}
\end{lemmaq}

Additionally, it is useful to know that the level of a fixed point $\cgr{p}{A_1} \wrGraph \cdots \wrGraph \cgr{p}{A_m}$ in $\fpb{\sylow_{p^m}}{K_{p^m}}$
(that is, the number of orbits) is simply $\sum_{i = 1}^m |A_i|$.

\begin{lemmaq}[See {\cite[Lemma 7.12]{edge_monotone}}]\label{lem:level:sylow:graphs}
    For any prime \(p\) and any positive integer \(m\),
    the level of $\cgr{p}{A_1} \wrGraph \dots
    \wrGraph \cgr{p}{A_m} \in \fpb{\sylow_{p^m}}{K_{p^m}}$ is
    \[\Hasselevel(\cgr{p}{A_1} \wrGraph \dots
    \wrGraph \cgr{p}{A_m}) = \sum_{i = 1}^m |A_i|. \qedhere\]
\end{lemmaq}

\subsection{Fixed Points of Sylow Groups Contain Large Bicliques}

The goal of this section is to understand the fixed points structure of $p$-Sylow groups.
We start by defining the \emph{empty-prefix} of a fixed point in $\fpb{\sylow_{p^m}}{K_{p^m}}$.

\begin{definition}[See {\cite[Definition 7.14]{edge_monotone}}]
    Write $H \coloneqq \cgr{p}{A_1} \wrGraph \dots \wrGraph \cgr{p}{A_m}$ for a fixed
    point of \(\fpb{\sylow_{p^m}}{K_{p^m}}\).
    The \emph{empty-prefix} of \(H\)
    is the smallest index \(i\) with $A_i \neq \emptyset$, minus one;
    we write $\wrLevel(A_1, \dots, A_m) \coloneqq i-1$.
\end{definition}

Our first observation is that each fixed point with a low \emph{empty-prefix} contains a large
biclique as an edge-subgraph. This allows us to use \cref{cor:tight:bounds} which proves \w-hardness and tight bounds.
The proof itself is a simple computation.

\begin{lemmaq}[See {\cite[Lemma 7.15]{edge_monotone}}]\label{lem:treewidth:wreath:level}
    Let $p$ denote a prime number and let $m$ denote a positive integer.
    For each \(i \in \setn{m}\), let $A_i \subseteq \field{p}^+$ denote a subset and set
    $A \coloneqq (A_1, \dots, A_m)$.
    Then, $\cgr{p}{A_1} \wrGraph \cdots \wrGraph \cgr{p}{A_{m}}$ contains
    $K_{p^{m - 1 - \wrLevel(A)}, p^{m - 1 - \wrLevel(A)}}$ as a subgraph.
\end{lemmaq}

Our second observation is that fixed points with a high empty-prefix are
always isomorphic to edge-subgraphs of fixed points with a low
empty-prefix. This allows us to consider only fixed points with low
empty-prefix (see \cref{theo:edge_mono:prime:power:bicliques}).

\begin{lemma}[See {\cite[Lemma 7.17]{edge_monotone}}]\label{lem:wreath:sub:iso}
    Let $p$ denote a prime number and let $m$ denote a positive integer.
    For each \(i \in \setn{m}\), let $A_i \subseteq \field{p}^+$ denote a subset.
    Then, for all $j \in \setn{m}$, the graph
    $\cgr{p}{\emptyset} \wrGraph \cdots\cgr{p}{\emptyset} \wrGraph
    \cgr{p}{A_1} \wrGraph \cdots \wrGraph \cgr{p}{A_{m-j}}$ is isomorphic to an
    edge-subgraph of $\cgr{p}{A_1} \wrGraph \cdots \wrGraph \cgr{p}{A_{m}}$.
\end{lemma}
\begin{proof}
    The intuition behind this proof is that we can use the following bijection
    \[\push_{p^m} \colon V(K_{p^m}) \to V(K_{p^m}); (a_1 \dots, a_m) \mapsto (a_m, a_1, \dots, a_{m-1})\]
    to show that $\push_{p^m}(\cgr{p}{\emptyset} \wrGraph \cgr{p}{A_1} \wrGraph \cdots \wrGraph \cgr{p}{A_{m-1}})$
    is an edge-subgraph of $\cgr{p}{A_1}  \wrGraph \cdots \wrGraph \cgr{p}{A_{m}}$.
    Now, the iteration of this result yields the claim.
\end{proof}

\subsection{Product Groups, Graphs Unions, and Graph Joins}\label{sec:product}

In the last two sections, we consider fixed points of $K_m$ whenever $m$ is a $p$ power.
However, if $m$ is not a $p$-power then we can use the product of multiple $p$-groups to get a $p$-subgroup of
$\aut(K_m)$. For (not necessarily disjoint) sets $X_1, \dots, X_m$, the disjoint union $X_1 \unionSet \cdots \unionSet X_m$
is the set $\{(i, x) \mid i \in \setn{m}, x \in X_i\}$.

\begin{definition}[See {\cite[Definition 6.3]{edge_monotone}}]
    For permutation groups $\Gamma_1 = (G_1, \circ), \dots, \Gamma_m = (G_m, \circ)$
    with $\Gamma_i \subseteq \sym{X_i}$,
    their product group $\Gamma \coloneqq \Gamma_1 \prodGroup \cdots \prodGroup \Gamma_m$ is
    the set $G \coloneqq G_1 \times \cdots \times G_m$ together with the component-wise
    function composition, that is, for \((g_1, \dots, g_m), (g_1', \dots,
    g_m') \in \Gamma\), we set
    \[
        (g_1, \dots, g_m) \circ (g_1', \dots, g_m')
        \coloneqq (g_1 \circ  g_1',\dots, g_m \circ g_m').
    \]

    We let \(\Gamma\) act on \(X \coloneqq X_1 \unionSet \cdots \unionSet X_m\) via
    \[
        g(j, x_j) \coloneqq (j, g_j(x_j)), \text{ for all } g = (g_1, \dots, g_m) \in G
        \text{ and } x = (j, x_j) \in X.
        \tag*{\qedhere}
    \]
\end{definition}

Let $\Gamma_1 \subseteq \aut(G_1), \dots, \Gamma_s \subseteq \aut(G_m)$ denote $p$-groups.
Now, it is easy to see that $\Gamma \coloneqq \Gamma_1 \prodGroup \cdots \prodGroup \Gamma_m$
is also a $p$-group, and that it is a subgroup of $\aut(G)$ for $G \coloneqq G_1 \times \cdots \times G_m$.

Next, we try to understand the fixed points structure of $\fp(\Gamma, G)$. For this, we need some
way to \emph{merge} multiple graphs together.

\definhab

Now, our observation is that the fixed points of $\fp(\Gamma, G)$ are \emph{inhabited graphs}.
In particular, each fixed point $A  \coloneqq \metaGraph{C}{A_1, \dots, A_m}$ consists of
$m$ blocks, where the $i$-th block is a fixed point $A^i \in
\fpb{\Gamma_i}{G_i}$ and we fully connect two blocks $i$ and $j$ with each other if
$\{i, j\} \in E(C)$.

\begin{corollaryq}[See {\cite[Corollary 6.8]{edge_monotone}}]\label{theo:prod:fixed:points}
    For \(i  \in \setn{m}\), let \(G_i\) denote a graph
    and let \(\Gamma_i
    \subseteq \auts{G_i}\) denote a transitive permutation group. Then, we have
    \begin{align*}
        \fpb{\BigProdGroup_{i = 1}^m \Gamma_i}{\BigJoinGraph_{i = 1}^m G_i} =
    \left\{ \metaGraph{C}{A^1, \dots, A^m} \mid C \in \graphs{m}  , A^i \in \fp(\Gamma_i,
    G_i)\right\}.  \tag*{\qedhere}%
    \end{align*}
\end{corollaryq}

The structure of these fixed points is very useful whenever we use large graphs. For example, consider a fixed point
$\metaGraph{C}{A^1, \dots, A^m} \in \fpb{\BigProdGroup_{i = 1}^m \Gamma_i}{\BigJoinGraph_{i = 1}^m G_i}$.
Now, if $\{i, j\} \in E(C)$, then $K_{|V(G_i)|, |V(G_j)|}$ is a subgraph of $\metaGraph{C}{A^1, \dots, A^m}$.
This implies that for large graphs $G_i$ and $G_j$, we can use \cref{cor:tight:bounds} to show \w-hardness and tight bounds.
Lastly, the following remark is quite useful.

\begin{remark}\label{rem:sub-points:IS}
    For \(i  \in \setn{m}\), let \(G_i\) denote a graph, let \(\Gamma_i \subseteq \auts{G_i}\) denote a transitive permutation group, and let
    $C \in \graphs{m}$ be a graph. The sub-points of a fixed point $\metaGraph{C}{\IS, \dots, \IS}$ are precisely the fixed points
    $\metaGraph{C'}{\IS, \dots, \IS}$, where $C'$ is an edge-subgraph of $C$.
\end{remark}

Using \cite[Lemma 6.5]{edge_monotone}, we obtain that the level of the fixed point
\[\metaGraph{C}{A^1, \dots, A^m} \in \fpb{\BigProdGroup_{i = 1}^m
\Gamma_i}{\BigJoinGraph_{i = 1}^m G_i}\]
is equal to
\[\hl{\metaGraph{C}{A^1, \dots, A^m}} = |E(C)| + \sum_{i = 1}^m \hl{A^i}.\]

\section{The Parameterized Complexity of $\MOD{}\clique{}$}\label{sec:appendix:mod}

For our modular counting results from \cref{sec:mod}, we rely on the hardness
$\MODP{p}\clique{}$. However, to show that $\MODP{p}\clique{}$ is hard,
we introduce the following two problems.

\begin{definition}
    We define the following two decision problems
    \begin{itemize}
        \item The problem $\UniSat{k}$ gets as input a $\sat{k}$ instance with the guarantee that this instance either has 0 or 1 solutions.
        The problem is to decide whether a satisfying assignment exists.
        \item The problem $\UniKClique{k}$ gets as input a graph $G$ with the guarantee that this instance either has 0 or 1 many $k$-cliques.
        The problem is to decide whether a $k$-clique exists.
        \qedhere
    \end{itemize}
\end{definition}

To show that there is a constant $\alpha > 0$ such that
no algorithm solves $\UniKClique{k}$ in times $O(|V(G)|^{\alpha k})$, we use
a reduction from $\UniSat{3}$. According to \cite{CALABRO2008386} there
is a constant  $\gamma > 0$ such that no algorithm solves  $\UniSat{3}$ in time $O(2^{\gamma n})$
unless the randomized exponential time hypothesis fails. The randomized exponential time
hypothesis (rETH) states that there is a constant $\delta_3$ such that no randomized
algorithm solves \sat{3} in expected time $O(2^{\delta_3 n})$.

\begin{theoremq}[\cite{CALABRO2008386}]\label{theo:ETH:unique:SAT}
    Assuming rETH, there is a constant $\gamma > 0$ such that no
    algorithm solves $\UniSat{3}$ in time $O(2^{\gamma n})$, where $n$ is the number of variables.
\end{theoremq}

Next, we combine \cref{theo:ETH:unique:SAT} with the sparsification lemma to obtain the following statement.

\begin{corollary}\label{cor:ETH:unique:SAT}
    Assuming rETH, there is a constant $\delta > 0$ such that no algorithm
    solves $\UniSat{3}$ in time $O(2^{\delta (n + m)})$,
    where $n$ is the number of variables and $m$ is the number of clauses.
\end{corollary}
\begin{proof}
    Suppose that there is a randomized algorithm that solves the $\UniSat{3}$ problem in
    expected time $O(2^{\delta (n + m)})$.
    Then, we can use the Sparsification Lemma (see \cite[Theorem 14.4]{param_algo}) to get an
    algorithm that solves $\UniSat{3}$ in expected time $O(2^{\delta' n})$ for some $\delta' > 0$.
\end{proof}

Using \cref{cor:ETH:unique:SAT}, we can show tight bounds for $\UniKClique{k}$ assuming rETH.
The proof constructs a parsimonious reduction from $\sat{3}$ to $k$-\clique{} that is fast
enough to keep the exponent.

\begin{theorem}\label{thm:uni-k-clique:ETH}
    Assuming rETH, there is a constant $\gamma > 0$ such that, no
    algorithm (that reads the whole input) solves $\UniKClique{k}$
    in time $O(|V(G)|^{\alpha k})$.
\end{theorem}
\begin{proof}
    The proof is a modification of the proof of Theorem 14.21 in \cite{param_algo} and \cite[Lemma B.2]{edge_monotone}.
    Originally, the proof works by first using a reduction from $\sat{3}$ to $3$-$\textsc{Coloring}$ and later
    a reduction from $3$-$\textsc{Coloring}$ to $k$-$\clique{}$. The problem with this approach is that this
    reduction is not parsimonious since each valid satisfying assignment generates multiple valid
    3 colorings. We circumvent this problem by combining both reductions into a single one.
    This gives us more freedom in defining the colorings, which ensures that each
    satisfying assignment only yields a single \emph{valid} coloring that later leads to exactly one $k$-clique. This
    parsimonious reduction proves the claim.

    Give a $\sat{3}$ formula $\phi$ with $n$ variables and $m$ clauses, we construct a graph $G_{\phi}$ using the
    standard reduction \cite[Theorem 2.1]{GAREY1976237} from  $\sat{3}$ to $3$-$\textsc{Coloring}$.

    Let $x_1, \dots, x_n$ be the variables of $\phi$ and $C_1, \dots, C_m$ be the clauses of $\phi$. The $i$-th clause
    is equal to $C_i \coloneqq \ell_{i, 1} \lor \ell_{i, 2} \lor \ell_{i, 3}$, where $\ell_{i, 1}, \ell_{i, 2}, \ell_{i, 3}$
    are literals. The vertex set of $G_{\phi}$ is defined as
    \[V(G_{\phi}) = \{T, F, B\} \cup \{x_i, \overline{x}_i : i \in [n] \} \cup \{y_{i, 1}, \dots, y_{i, 6} : i \in [m]\}\]
    this is, we introduce two vertices $x_i$ and $\overline{x}_i$ per variable and six vertices per clause. We call the vertices
    $x_i$ and $\overline{x}_i$ \emph{literal vertices} since they represent the literals of $\Phi$. Further, we
    introduce three auxiliary vertices $T, F, B$.

    Next, we describe the edge set of $G_{\phi}$. First, we connect the vertices
    $T, F, B$ with each other to form a triangle. Next, we connect $x_i$ with $\overline{x}_i$ and connect both
    $x_i$ and $\overline{x}_i$ with $B$. Further, we connect each vertex $y_{i, 6}$ to both $F$ and $B$. Lastly,
    for each clause $C_i$, we add the clause gadget from \cref{fig:clause:gadget} to the graph $G_{\Phi}$.

    \begin{figure}[t]
        \centering
        \begin{subfigure}{.35\textwidth}
        \centering
        \includegraphics{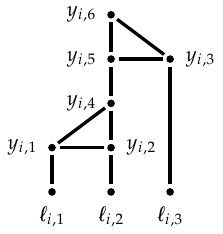}
        \label{fig:clause:gadget}
        \caption{The clause gadget of the clause $C_i$}
        \end{subfigure}\qquad%
        \begin{subfigure}{.55\textwidth}
        \centering
        \includegraphics{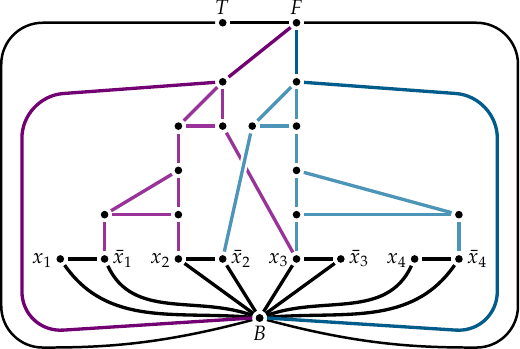}
        \label{fig:clause:gadget}
        \caption{The graph $G_\phi$ for \(\phi = (\bar{x}_1 \vee x_2 \vee x_3) \wedge
        (\bar{x}_4 \vee x_3 \vee \bar{x}_2)\).}
        \end{subfigure}
        \caption{How to construct the graph $G_{\phi}$.}\label{fig:clause:gadget}
    \end{figure}

    Observe, that the vertex set of each clause gadget $C_i$ contains three literal vertices
    and six additional vertices $y_{i, 1}, \dots, y_{i, 6}$. This means that a literal vertex
    $x_i$ or $\overline{x_i}$ might occur in multiple clause gadgets.

    We introduce three colors $\colorT, \colorF, \colorB$. Let $f \colon V(G_{\phi}) \to \{\colorT, \colorF, \colorB\}$
    be a coloring of $G_{\phi}$.
    Assume that $f$ is a proper coloring with $f(z) = z$, for all $z \in \{\colorT, \colorF, \colorB\}$,
    then $f(x_{i}), f(\overline{x}_i) \in \{\colorT, \colorF\}$
    and $f(x_i) \neq f(\overline{x}_i)$.
    Intuitively, this allows us to interpret the color $\colorT$ as \emph{true} and the color $\colorF$ as
    \emph{false}. The color $\colorB$ stands for \emph{base}.

    Next, we define what it means to be a valid coloring. We say that $f$ is a \emph{valid} coloring if
    \begin{itemize}
        \item $f(z) = z$, for all $z \in \{\colorT, \colorF, \colorB\}$,
        \item and the coloring of each clause gadget $C_i$ is valid.
    \end{itemize}

    We define what it means that a clause gadget $C_i$ has a valid color. By \cite[Theorem 2.1]{GAREY1976237}, if
    each literal $\ell_{i, 1}, \ell_{i, 2}, \ell_{i, 3}$
    is colored using the colors $\colorT$ and $\colorF$, then each proper coloring $f$ of a clause gadget of $C_i$
    has the property that $y_{i, 6}$ can only be colored with the color $\colorT$ if at least one
    literal is colored with $\colorT$. Intuitively this means that the gadget simulates a disjunction
    of the literals $\ell_{i, 1}, \ell_{i, 2}, \ell_{i, 3}$. However, if we have an assignment of
    variables such that $\ell_{i, 1} \lor \ell_{i, 2} \lor \ell_{i, 3}$ is true,
    then they might exist multiple proper colorings $f$ of the clause gadget $C_i$ with $f(y_{i, 6}) = \colorT$.
    We avoid this ambiguity by using the following approach.

    For each assignment of the literals  $\ell_{i, 1}, \ell_{i, 2}, \ell_{i, 3}$ (i.e. colorings using $\colorT$ or $\colorF$)
    such that $\ell_{i, 1} \lor \ell_{i, 2} \lor \ell_{i, 3}$
    is true, we consider all proper colors $f$ of the clause gadget with $f(y_{i, 6}) = \colorT$. Now, we choose one of those
    proper colors and define it to be a valid coloring. This way, we get exactly 7 valid colors for clause gadgets.

    Next, we show that there is a one-to-one mapping between satisfying assignments and colorings that are
    both proper and valid.
    \begin{claim}
        The number of satisfying assignments in $\phi$ is equal to the number of proper and valid colorings
        $f \colon V(G_{\phi}) \to \{\colorT, \colorF, \colorB\}$.
    \end{claim}
    \begin{claimproof}
        Let $f$ be a proper and valid coloring of $G_{\phi}$ then $f$ defines a satisfying assignment $\psi(f)$ of $\phi$
        by setting the $i$th variable to true if and only if $f(x_i) = \colorT$. Since $f$ is a valid coloring
        we get that $f(z) = z$ for $z \in \{T, F, B\}$. However, this implies $f(y_{i, 6}) = \colorT$ for all
        clause gadgets $C_i$ because $y_{i, 6}$ is connected to the vertices $T$ and $F$. This is only possible
        if at least one of the literals $\ell_{i, 1}, \ell_{i, 2}, \ell_{i, 3}$ is colored using $\colorT$
        which implies that $\psi(f)$ is an assignment that satisfies all clauses $C_i$ of $\phi$.

        On the other side, let $A \coloneqq \{z_1, \dots, z_n\}$ be a satisfying assignment of $f_A$. We show that this
        assignment corresponds to exactly one valid and proper coloring $f_A$. First, we define
        $f_A(z) = z$ for $z \in \{T, F, B\}$. Next, we define $f_A(x_i) = \colorT$ and $f_A(\overline{x}_i) = \colorF$
        if $z_i$ is true, and  $f_A(x_i) = \colorF$ and $f_A(\overline{x}_i) = \colorT$ if $z_i$ is false.
        Lastly, let $C_i$ be a literal. Since $A$ is satisfying, we get that $\ell_{i, 1} \lor \ell_{i, 2} \lor \ell_{i, 3}$ is true.
        This means that there exists exactly one valid coloring of $C_i$ with the property that $y_{i, 6}$ gets the color
        $\colorT$. We use this coloring to color the vertices $y_{i, 1}, \dots, y_{i, 6}$. This defines our coloring $f_A$.
        It is now easy to check that $f_A$ is a valid and proper coloring.

        Observe that for all colorings $f$ we get $f_{\psi(f)} = f$ and for all satisfying assignments $A$, we obtain
        $\psi(f_A) = A$. This shows that the set of proper and valid colorings is equal to the set of satisfying assignments.
    \end{claimproof}

    We show how to compute the number of proper and valid $3$-colorings by counting cliques of size $(2k + 1)$.

    \begin{claim}\label{claim:clique:sat}
        If we can solve $k$-$\clique{}$ in time $O(n^{\alpha k})$, then we can solve \sat{3}
        in time $O( b^{ 2(n + m) / k} \cdot (n + m)^2 + ( b^{2\alpha(n + m)}))$ where $b = 3^9$.
    \end{claim}
    \begin{claimproof}
    First, we split the variables and clauses of the
    formula $\phi$ into different groups. For this, we split the variables $x_1, \dots, x_n$ into $k$ groups such
    that each group contains at most $\lceil n / k \rceil$ many variables. Next, we do the same for the clauses and
    obtain $k$ groups of clauses such that each group contains at most $\lceil m / k \rceil$ many clauses.
    We use this to define $2k + 1$ many groups on $V(G_\phi)$:
    \begin{itemize}
        \item We define one group containing the vertices $\colorT$, $\colorF$, and $\colorB$.
        \item Each group of variables $x_{a_1}, \dots, x_{a_s}$ defines a group of \emph{literal vertices}
        $x_{a_1}, \overline{x}_{a_1}, \dots, x_{a_s}, \overline{x}_{a_s}$.
        \item Each group of clauses $C_{b_1}, \dots, C_{b_s}$ defines a group of clause gadget vertices.
        The vertices of this group are precisely those vertices that are part of some clause gadget $C_{b_i}$ (see \cref{fig:clause:gadget}).
    \end{itemize}
    Note that a vertex $x_i \in V(G_\phi)$ might be contained in multiple groups. We use $V_1, \dots, V_{2k +1}$ to enumerate the groups.
    Observe that $|V_i| \leq 9 (n + m) / k$. Next, we create a graph $\Tilde{G}_\phi$ in the following.

    The idea is that the vertices of $\Tilde{G}_\phi$ represent valid colorings of subgraphs of $G_\phi$, meaning that
    each vertex of $\Tilde{G}_\phi$ has the form $(V_i, f_i)$, where $f_i$ is a valid and proper coloring of $V_i$. For the group
    $\{\colorT, \colorF, \colorB\}$, we create a single vertex with the coloring $f_i(z) = z$. Next, let $V_i$ be a group of \emph{literal vertices}
    then we create a vertex $(V_i, f)$ for each valid and proper coloring $f_i$ on $V_i$ meaning that $f_i$ can only use colorers
    in $\{\colorT, \colorF\}$. Lastly, let $V_i$ be a group of clause gadget vertices, then we create a vertex $(V_i, f_i)$
    for each valid and proper coloring $f_i$, meaning that $f_i$ is a valid color of each clause gadget $C_j$ that is part of the group $V_i$.
    Note that $V_i$ either contains all vertices of a clause gadget or no vertex of a clause gadget. Observe that $\Tilde{G}_\phi$ has
    at most $(2k + 1) \cdot 3^{9 (n + m) / k} = (2k + 1) \cdot b^{(n + m) / k}$ many vertices, where $b \coloneqq 3^9$.

    Next, we add an edge between $(V_i, f_i)$ and $(V_j, f_j)$ if
    $i \neq j$, $f_i(x) = f_j(x)$ for all $x \in V_i \cap V_j$, and $f$ is a proper $3$-coloring of $G_\phi(V_i \cup V_j)$,
    where $f$ is defined as $f(x) = f_i(x)$ if $x \in V_i$ and $f(x) = f_j(x)$ if $x \in V_j$. It is easy to see that can
    construct the graph $\Tilde{G}_\phi$ in time
    \[O( ((2k + 1) \cdot b^{ (n + m) / k})^2 \cdot (n + m)^2 )\]
    Further, it is easy to see that each proper and valid coloring of $G_\phi$ corresponds to exactly one
    $(2k + 1)$-clique in $\Tilde{G}_\phi$. Assume next that we can solve $k$-clique in time $O(n^{\alpha k})$, then
    we can compute the number of satisfying assignments of $\phi$ in time
    \[O(((2k + 1) \cdot b^{ (n + m) / k})^2 \cdot (n + m)^2 + ((2k + 1) \cdot b^{(n + m) / k})^{\alpha (2k + 1)})\]
    For a fixed $k$, this time is in $O(b^{ 2(n + m) / k} \cdot (n + m)^2 + ( b^{2\alpha(n + m)}))$.
    \end{claimproof}

    Let $\delta > 0$ be constant from \cref{cor:ETH:unique:SAT}. We set $\alpha \coloneqq \log_b(2) \delta /3$.
    If there is an $k \geq 1 / \alpha$ such that we can solve $k$-$\clique{}$ in time $O(n^{\alpha k})$,
    then we use \cref{claim:clique:sat} to solve \sat{3} in time
    \[O(b^{ 2(n + m) / k} \cdot (n + m)^2 + ( b^{2\alpha(n + m)}))
    \subseteq O(b^{ 2\alpha(n + m) } \cdot (n + m)^2)
    \subseteq  O(2^{2 \delta(n + m) / 3 } \cdot (n + m)^2) \subseteq  O(2^{ \delta(n + m) }). \]
    However, according to \cref{cor:ETH:unique:SAT} this is only possible if ETH fails.
    If $k < 1 / \alpha$ then $O(n^{\alpha k})$ would be sublinear and there is no
    sublinear algorithm that reads the whole input.
\end{proof}

Now, \cref{thm:uni-k-clique:ETH} gives us tight lower bounds for $\UniKClique{k}$, which immediately yields
tight lower bounds for \cref{cor:k-clique:mod:ETH}.

\corcliquemodeth
\begin{proof}
    Let $\alpha$ be the constant from \cref{thm:uni-k-clique:ETH}.
    Assume that we have an algorithm that decides if the number of $k$-cliques is divisiable by $p$ in time $O(|V(G)|^{\alpha k})$
    then this algorithm can be used to solve $\UniClique$ in time $O(|V(G)|^{\alpha k})$. However, this is not possible
    unless rETH fails.
\end{proof}

Lastly, we show that $\MODP{p}\clique{}$ and $\NUMP{p}\clique{}$ are equivalent under parameterized
Turing reductions

\begin{lemma}\label{lem:clique:mod:equiv}
    The problems $\MODP{p}\clique{}$ and $\NUMP{p}$-$\clique{}$ are equivalent under parameterized Turing reductions.
\end{lemma}
\begin{proof}
    Let $A$ be an algorithm that computes $\NUMP{p}\clique{}$. On input $(G, k)$, we use the
    algorithm to compute the number of $k$-cliques modulo $p$. If this number is zero, then we return
    true. Otherwise, we return false. This algorithm decides if the number $k$-cliques is divisible
    by $k$.

    On the other side, let $A$ be an algorithm that decides $\MODP{p}\clique{}$. For a given input
    $(G, k)$, let $G_i$ be the graph that we get by adding $i$ many $k$-cliques to $G$. Let $x$
    be the number of $k$-cliques in $G$, then observe that the number  of $k$-cliques in $G_i$
    is equal to $x + 1$ since we get $x$ many $k$-cliques in $G$ plus one $k$-clique for
    each $k$-clique that we add to $G$. Thus $G_i$ contains $x + i$ many $k$-cliques.
    We now simply run $A$ on input $(G_0, k)$, $(G_1, k)$, $\dots$ $(G_{p-1}, k)$.
    Observe that $A$ returns true exactly once and let $G_i$ be this graph.
    This means that $x + i \equiv_p 0$ which implies that $x \equiv_p -i$. Thus
    we simply return $-i + p$ which is equal to the number of $k$-cliques
    in $G$ modulo $p$.
\end{proof}

\end{document}